\documentclass[11pt,a4paper]{article}

\usepackage{hyperref}
\usepackage{url}

\usepackage[margin=1in]{geometry}

\usepackage[latin2]{inputenc}
\usepackage[english]{babel}
\usepackage[cmex10]{amsmath}
\usepackage{amsthm}
\usepackage{amssymb}
\usepackage{tabularx}
\usepackage{graphicx}
\usepackage{tikz}
\usepackage{xspace}
\usepackage{comment}
\usepackage{enumitem}

\setlist{noitemsep}

\newcommand{\C}{\H}

\renewcommand{\H}{\mathcal{H}}

\newcommand{\executeiffilenewer}[3]{%
\ifnum\pdfstrcmp{\pdffilemoddate{#1}}%
{\pdffilemoddate{#2}}>0%
{\immediate\write18{#3}}\fi%
} 

\newcommand{\svg}[2]{\def\svgwidth{#1}%
\executeiffilenewer{#2.svg}{#2.pdf}%
{inkscape -z -D --file=#2.svg %
--export-pdf=#2.pdf --export-latex}%
{\input{#2.pdf_tex}}}

\newtheorem{theorem}{Theorem}[section]
\newtheorem{lemma}[theorem]{Lemma}
\newtheorem{claim}[theorem]{Claim}
\newtheorem{corollary}[theorem]{Corollary}
\newtheorem{proposition}[theorem]{Proposition}
\theoremstyle{definition}
\newtheorem{definition}[theorem]{Definition}
\newtheorem{remark}[theorem]{Remark}
\newcommand{\cqed}{\renewcommand{\qedsymbol}{$\lrcorner$}}

\newcommand{\tw}{\textup{tw}}

\newcommand{\B}{\mathcal{B}}

\newcommand{\hX}{\widehat{X}}
\newcommand{\hS}{\widehat{S}}

\newcommand{\stwo}{\psi}
\newcommand{\ktwo}{\kappa_2}
\newcommand{\boundary}{\partial}
\newcommand{\fdegree}{f_d}
\newcommand{\fstar}{f_s}
\newcommand{\fmain}{f}
\newcommand{\Dh}{D_1}
\newcommand{\Dl}{D_2}
\newcommand{\kh}{k_H}
\newcommand{\kl}{k_L}
\newcommand{\kx}{k_X}

\newcommand\sharpP{\textup{\#P}}
\newcommand\sharpWone{\textup{\#W[1]}}
\newcommand\Wone{\textup{W[1]}}
\newcommand\FPT{\textup{FPT}}
\newcommand\leqT{\leq_{\mathsf{fpt}}^{\mathrm{T}}}
\newcommand\leqLin{\leq_{\mathsf{fpt}}^{\mathrm{T,\ell}}}
\newcommand\leqFpt{\leq_{\mathsf{fpt}}}

\newcommand\pClique{\mathsf{\#Clique}}
\newcommand\pBiclique{\mathsf{\#Biclique}}
\newcommand\pUndirCycle{\mathsf{\#UndirCycle}}
\newcommand\pDirCycle{\mathsf{\#DirCycle}}
\newcommand\pUndirPath{\mathsf{\#UndirPath}}
\newcommand\pDirPath{\mathsf{\#DirPath}}
\newcommand\pMatch{\mathsf{\#Match}}
\newcommand\pColMatch{\mathsf{\#ColMatch}}
\newcommand\pSub{\mathsf{\#Sub}}
\newcommand\pEmb{\mathsf{\#Emb}}

\newcommand\pRestr{\mathsf{\#PartitionedSub}}
\newcommand\pGridTiling{\mathsf{Grid Tiling}}

\newcommand\Sub{\mathsf{Sub}}
\newcommand\Emb{\mathsf{Emb}}

\newcommand\SubPart{\mathsf{PartitionedSub}}

\newcommand\pure{\mathrm{pure}}

\makeatletter
\def\namedlabel#1#2#3{\begingroup
	#2%
	\def\@currentlabel{#3}%
	\label{#1}\endgroup
}
\makeatother

\begin{document}
\title{Complexity of counting subgraphs: only the boundedness of the vertex-cover number counts}
\author{Radu Curticapean\thanks{Dept. of Computer Science, Saarland University, Saarbr\"ucken,
Germany, curticapean@cs.uni-sb.de}, D\'{a}niel Marx\thanks{Institute for Computer Science and Control, Hungarian Academy of Sciences (MTA SZTAKI), Budapest,
Hungary, dmarx@cs.bme.hu}
}
\maketitle

\begin{abstract}
  For a class $\H$ of graphs, $\pSub(\H)$ is the counting problem
  that, given a graph $H\in \H$ and an arbitrary graph $G$, asks for
  the number of subgraphs of $G$ isomorphic to $H$. It is
  known that if $\H$ has bounded vertex-cover number (equivalently,
  the size of the maximum matching in $\H$ is bounded), then $\pSub(\H)$ is
  polynomial-time solvable. We complement this result with a corresponding
  lower bound: if $\H$ is {\em any} recursively enumerable class of
  graphs with unbounded vertex-cover number, then $\pSub(\H)$ is
  \sharpWone-hard parameterized by the size of $H$ and hence not
  polynomial-time solvable and not even fixed-parameter tractable,
  unless $\FPT=\sharpWone$.

  As a first step of the proof, we show that counting $k$-matchings in
  bipartite graphs is \sharpWone-hard. Recently, Curticapean [ICALP 2013]
  \cite{DBLP:conf/icalp/Curticapean13} proved the \sharpWone-hardness of
  counting $k$-matchings in general graphs; our result strengthens
  this statement to bipartite graphs with a considerably simpler proof
  and even shows that, assuming the
  Exponential Time Hypothesis (ETH), there is no $f(k)n^{o(k/\log k)}$
  time algorithm for counting $k$-matchings in bipartite graphs for
  any computable function $f(k)$. As a consequence, we obtain an
  independent and somewhat simpler proof of the classical result of
  Flum and Grohe [SICOMP 2004] \cite{MR2065338} stating that counting
  paths of length $k$ is \sharpWone-hard, as well as a similar
  almost-tight ETH-based lower bound on the exponent.

\end{abstract}
\section{Introduction}
Counting the number of solutions is often a considerably more
difficult task than deciding whether a solution exists or finding a
single solution.  A classical example is the case of perfect matchings
in bipartite graphs: there are well-known polynomial-time algorithms
for finding a perfect matching, but the seminal result of Valiant
\cite{MR526203} showed that counting the number of perfect matchings
in bipartite graphs is \#P-hard, and hence unlikely to be
polynomial-time solvable.  This phenomenon has been systematically
analyzed, for example, in the context of Constraint Satisfaction
Problems (CSPs), where dichotomy theorems characterizing the
polynomial-time solvable and \#P-hard cases
\cite{DBLP:journals/jacm/Bulatov13,DBLP:journals/jacm/BulatovDGJM13,DBLP:journals/jcss/BulatovDGJJR12,DBLP:journals/siamcomp/DyerGJ09,DBLP:journals/tcs/BulatovDGJR09}
show that very restrictive conditions are needed to ensure that
not only the decision problem is polynomial-time solvable, but the
counting problem is as well.

Our goal in the present paper is to systematically analyze the
tractable cases of counting subgraphs.  Counting the number of times a
certain pattern appears in a graph is a fundamental theoretical
problem that has been explored intensively also on real-world large
graphs
\cite{DBLP:conf/wea/SchankW05,DBLP:conf/icdm/Tsourakakis08,DBLP:conf/kdd/BecchettiBCG08,DBLP:journals/im/KolountzakisMPT12,alon-network-motifs}.
Formally, given graphs $H$ and $G$, the task is to count the number of
subgraphs of $G$ that are isomorphic to the pattern graph $H$; we
would like to understand which graphs $H$ make this problem easy or
hard. However, we have to be careful how we formulate the framework of
our investigations. For every fixed pattern graph $H$, the
number of subgraphs of $G$ isomorphic to $H$ can be determined in 
polynomial-time by brute force: it suffices to check each of
the $|V(G)|^{|V(H)|}$ mappings from the vertices of $H$ to the
vertices of $G$, resulting in a simple polynomial-time algorithm for fixed
$H$. There is a line of research devoted to finding nontrivial
improvements over brute-force search for specific patterns
\cite{DBLP:journals/siamcomp/ItaiR78,DBLP:journals/jacm/AlonYZ95,DBLP:journals/algorithmica/AlonYZ97,DBLP:journals/ipl/KloksKM00,DBLP:conf/esa/BjorklundHKK09,DBLP:journals/jcss/FominLRSR12,DBLP:journals/siamcomp/WilliamsW13,DBLP:conf/isaac/FloderusKLL13,DBLP:conf/soda/BjorklundKK14}. Besides
improvements for specific small graphs $H$, these papers identified
structural properties, such as boundedness of treewidth, pathwidth,
and vertex-cover number, that can give improvements for some infinite
classes $\H$ of graphs $H$. Our goal is to exhaustively characterize which
graph properties are sufficiently strong to guarantee polynomial-time
solvability.

The search for graph properties that make counting easy can be formally
studied in the following framework. For every class $\H$ of graphs,
$\pSub(\H)$ is the counting problem where, given a graph $H\in \H$ and
arbitrary graph $G$, the task is to count the number of (not necessarily induced) subgraphs of
$G$ isomorphic to $H$. Rather than asking which fixed graphs $H$ make
counting easy (as we have seen, the problem is polynomial-time
solvable for every fixed $H$), we ask which {\em classes} $\H$ of
graphs make $\pSub(\H)$ polynomial-time solvable. Furthermore, as
many of the theoretical results and applications involve counting a
small fixed pattern graph $H$ in a large graph $G$, an equally natural
question to ask is whether $\pSub(\H)$ can be solved in time
$f(|V(H)|)\cdot n^{O(1)}$ for some computable function $f$ depending
only on the size of $H$. That is, we may ask whether $\pSub(\H)$ for a
particular class $\H$ is {\em fixed-parameter tractable (FPT)}
parameterized by $|V(H)|$.

\textbf{Main result.} The vertex-cover number $\tau(H)$ of a graph $H$ is the minimum size
of a set of vertices that contains at least one endpoint of every
edge. It is well known that if $\nu(H)$ is the size of a maximum
matching in $G$, then $\nu(H)\le \tau(H)\le 2\nu(H)$.  
If the class $\H$ has bounded vertex-cover number (or equivalently on the maximum matching size), then it follows from a result of
Vassilevska Williams and
Williams~\cite{DBLP:journals/siamcomp/WilliamsW13} that $\pSub(\H)$ is
FPT and it follows from a result of Kowaluk, Lingas, and
Lundell~\cite{DBLP:journals/siamdm/KowalukLL13} that $\pSub(\H)$ is
actually polynomial-time solvable (we also present a simple
self-contained argument for the polynomial-time solvability of
$\pSub(\H)$ in Section \ref{sec:algo-vc}). Our main result complements these
algorithms by showing that boundedness of the vertex-cover number is
the only property of $\H$ that guarantees tractability of $\pSub(\H)$.
\begin{theorem}\label{th:main}
Let $\H$ be a recursively enumerable class of graphs. Assuming $\FPT\neq \sharpWone$, the following are equivalent:
\begin{enumerate}
\item $\pSub(\H)$ is polynomial-time solvable.
\item $\pSub(\H)$ is fixed-parameter tractable parameterized by $|V(H)|$.
\item $\H$ has bounded vertex-cover number.
\end{enumerate}
\end{theorem}
Let us review some results from the literature that are of similar
form as Theorem~\ref{th:main}. A result of Grohe, Schwentick, and
Segoufin~\cite{380867} can be interpreted as characterizing the
complexity of finding a vertex-colored graph $H\in \H$ in $G$; they
show that the tractability criterion is the boundedness of the
treewidth of $\H$. Grohe~\cite{1206036} considered the problem of
deciding if there is a homomorphism from a graph $H\in \H$ to $G$;
here the tractability criterion is the boundedness of the treewidth of
the core of $H$. For the problem of counting homomorphisms, Dalmau and
Jonsson \cite{DBLP:journals/tcs/DalmauJ04} showed that it is again the
boundedness of the treewidth that matters. Chen, Thurley, and Weyer
\cite{DBLP:conf/icalp/ChenTW08} studied the problem of finding induced
subgraphs, which is apparently the most difficult of these problems,
as the problem is easy only if the class $\H$ contains only graphs of
bounded size. In all of these results, similarly to
Theorem~\ref{th:main}, polynomial time and fixed-parameter
tractability coincide. An example where polynomial time and FPT is
not known to be equivalent is the result of Marx
\cite{DBLP:journals/jacm/Marx13}, which can be interpreted as
characterizing the complexity of finding vertex-colored hypergraphs. For
this problem, bounded submodular width is the property that guarantees
fixed-parameter tractability, but it is not known if it implies
polynomial-time solvability.

Very recently, Jerrum and Meeks
\cite{DBLP:journals/corr/Meeks14,DBLP:journals/corr/JerrumM13,DBLP:journals/corr/JerrumM13a}
studied problems related to counting induced subgraphs isomorphic to a
given graph $H$ and counting induced subgraphs satisfying certain
fixed properties. As these investigations are in the very
different setting of induced subgraphs, they are not directly related
to our results.

We remark that there have been investigations of finding and counting
subgraphs in a framework when the pattern graph $H$ is arbitrary and
the host graph $G$ is restricted to a certain class; some of these
results appear in the more general context of evaluating first-order
logical sentences
\cite{DBLP:journals/jacm/FrickG01,DBLP:journals/corr/GroheKS13,DBLP:journals/jacm/DvorakKT13}.
Needless to say, these results are very different from our setting.

\textbf{FPT vs.~polynomial time.} There are at least two reasons why it is very natural to study the
fixed-parameter tractability of $\pSub(\H)$ along with its
polynomial-time solvability. As mentioned earlier, there is a large
body of previous work focusing on counting small patterns in large
graphs, hence, for example, the question whether there is a
$2^{2^{O(k)}}\cdot n^{O(1)}$ time algorithm for counting cycles of
length $k$ fits naturally into the framework of previous
investigations. Ruling out polynomial-time algorithms would not, on
its own, answer whether such algorithms exist and therefore would
give only a partial picture of the complexity of counting
subgraphs. Moreover, it seems that understanding fixed-parameter
tractability is a prerequisite for understanding polynomial-time
solvability. In all the results mentioned in the previous paragraph,
the families of problems considered contain problems that seem to be
\#P-intermediate: they are unlikely to be polynomial-time solvable, but
they are unlikely to be \#P-hard either.  We face a similar situation
in the characterization of $\pSub(\H)$ (see the examples in the next
two paragraphs).  Due to the existence of such \#P-intermediate
problems $\pSub(\H)$, we cannot hope for a P vs. \#P-hard dichotomy.
It is a very fortunate coincidence that the polynomial-time solvable
and fixed-parameter tractable cases of $\pSub(\H)$ coincide, and hence
the characterization of the latter gives a characterization of the
former as well. 

As a first example, let us define a class $\H$ the following way:
for every $k\ge 1$, let $\H$ contain the graph $H_k$ consisting of a
clique of size $k$, padded with $2^k$ isolated vertices. We can count the
number of copies of $H_k$ in $G$ in time $|V(G)|^{O(k)}=|V(G)|^{O(\log
  |V(H)|)}$, hence $\pSub(\H)$ is solvable in quasi-polynomial time,
but there does not seem any way of improving this to polynomial
time. This suggests that the problem is NP-intermediate, as it is not
believed that NP-hard problems can be solved in quasi-polynomial
time.

More importantly, Chen et al.\ showed an analogue of Ladner's Theorem for induced subgraph counting problems $\#\mathsf{IndSub}(\H')$ under the assumption that $\mathrm{P}\neq \mathrm{P}^\sharpP$: Define a reflexive and transitive relation (a quasiorder) on the set of polynomial-time decidable graph classes by declaring $\H \leq \H'$ iff $\#\mathsf{IndSub}(\H)$ admits a polynomial-time Turing reduction to $\#\mathsf{IndSub}(\H')$. This relation is indeed reflexive and transitive, and it orders subgraph counting problems $\#\mathsf{IndSub}(\H)$ according to their complexity in the non-parameterized sense. In this quasiordered set, Chen et al.\ showed the existence of a dense linear order, similar to Ladner's theorem that establishes such a linear order between P and NP. This implies that, when counting induced subgraphs, there exist problems that are \#P-intermediate, i.e., they are neither in P nor \#P-complete.

\textbf{Complexity of counting $k$-matchings.} The study of the
fixed-parameter tractability of counting problems was initiated by
Flum and Grohe \cite{MR2065338}. Finding paths and cycles of length
$k$ is well known to be fixed-parameter tractable
\cite{DBLP:journals/jacm/AlonYZ95,DBLP:conf/wg/KneisMRR06,DBLP:journals/ipl/Williams09,DBLP:conf/focs/Bjorklund10},
but Flum and Grohe \cite{MR2065338} proved the surprising result that
counting paths and cycles of length $k$ is \sharpWone-hard, and
hence unlikely to be fixed-parameter tractable. They raised as an
open question whether counting $k$-matchings (i) in general graphs or (ii) in bipartite graphs is fixed-parameter
tractable. Very recently, Curticapean
\cite{DBLP:conf/icalp/Curticapean13} (based on the earlier work of
Bl\"aser and Curticapean \cite{DBLP:conf/iwpec/BlaserC12}) used quite
involved algebraic techniques to answer the first question in the negative
by showing that counting $k$-matchings is \sharpWone-hard on general graphs. Our proof
of Theorem~\ref{th:main} is based on a reduction from counting $k$-matchings. In fact, the proof technique requires the stronger result 
that counting $k$-matchings is \sharpWone-hard even in bipartite
graphs. Therefore, in Section~\ref{sec:bipart-edge-colorf}, we prove this stronger result using a proof that relies only on basic
linear algebra (the rank of the Kronecker
product of matrices) and is significantly simpler than the proof of
Curticapean \cite{DBLP:conf/icalp/Curticapean13}. Our proof also shows
the hardness of the ``edge-colorful'' variant where the edges of $G$
are colored with $k$ colors and we need to count the $k$-matchings in
$G$ where every edge has a different color. An additional benefit of
our proof is that, combined with a lower bound of Marx
\cite{marx-toc-treewidth} for \textsc{Subgraph Isomorphism}, it gives
an almost-tight lower bound on the exponent of $n$. The Exponential
Time Hypothesis (ETH) of Impagliazzo, Paturi, and Zane
\cite{MR1894519} implies that $n$-variable 3SAT cannot be solved in
time $2^{o(n)}$. Our result shows that, assuming ETH, the number of
$k$-matchings in a bipartite graph cannot be counted in time
$f(k)n^{o(k/\log k)}$ for any computable function $f$. There are
simple reductions from counting $k$-matchings to counting paths and cycles of
length $k$, thus our proof gives an independent and somewhat simpler
proof of the results of Flum and Grohe \cite{MR2065338} on counting
paths and cycles, together with almost-tight ETH-based lower bounds on
the exponent that were not known previously.
\begin{theorem}\label{th:main-matching}
The following problems are \sharpWone-hard and, assuming ETH, cannot be solved in time $f(k)\cdot n^{o(k/\log k)}$ for any computable function $f$:
\begin{itemize}
\item Counting (directed) paths of length $k$.
\item Counting (directed) cycles of length $k$.
\item Counting $k$-matchings in bipartite graphs.
\item Counting edge-colorful $k$-matchings in bipartite graphs.
\end{itemize}
\end{theorem}

\textbf{Hereditary classes.} As a warm up, In Section~\ref{sec:hereditary-classes}, we give a very simple proof
of Theorem~\ref{th:main} in the special case when $\H$ is hereditary,
that is, when $H\in \H$ implies that every induced subgraph of $H$ is also
in $\H$. If $\H$ is hereditary and has unbounded vertex-cover number,
then a Ramsey argument shows that $\H$ contains either every clique,
or every complete bipartite graph, or every matching (i.e., 1-regular
graph). In each case, \sharpWone-hardness follows. While this proof is
very simple and intuitively explains what the barrier is that we hit
when going beyond bounded vertex-cover number, it leaves many natural
questions unanswered. In principle, the counting problem where the
pattern is a set of $k$ disjoint triangles can be simpler than
counting $k$-matchings, but there is no hereditary class $\H$ such
that $\pSub(\H)$ expresses exactly the complexity of the former
problem: if a hereditary class $\H$ contains the disjoint union of $k$
triangles, then it also contains $k$-matchings, and hardness of $\H$ could follow from matchings alone.  Therefore, the setting of hereditary classes 
cannot answer if counting disjoint triangles is easier than counting matchings.
While intuitively it
would seem obvious that counting more complicated objects should not be easier
 (and, in particular, counting $k$ disjoint triangles should not be easier than counting $k$-matchings), there is no a priori
theoretical justification for this. In fact, in followup work, we
study edge-colored versions of the problem and identify cases where
removing vertices from the pattern can actually make the problem
harder. Therefore, it is a nontrivial conclusion of
Theorem~\ref{th:main} that adding edges and vertices to the pattern
does not make $\pSub(\H)$ any easier.

\textbf{Proof overview.} We proceed the following way for general (not
necessarily hereditary) classes $\H$. First, if $\H$ has unbounded
treewidth, then the arguments underlying the previous work of Grohe,
Schwentick, and Segoufin \cite{380867}, Grohe \cite{1206036}, Dalmau
and Jonsson \cite{DBLP:journals/tcs/DalmauJ04}, and Chen, Thurley, and
Weyer~\cite{DBLP:conf/icalp/ChenTW08} go through (see Section~\ref{sec:unbo-treew-graphs}). Essentially, we need
two reductions. First, there is a simple reduction from counting
cliques to counting colored grids. If $\H$ has unbounded treewidth,
then the Excluded Grid Theorem of Robertson and Seymour
\cite{MR854606} shows that the graphs in $\H$ have arbitrary large
grid minors. Therefore, we can embed the problem of counting colored
grids into $\pSub(\H)$. As these techniques are fairly standard by
now, the main part of our proof is handling the case when $\H$ has
bounded treewidth. This is the part where we have to deviate from
previous results (where bounded treewidth always implied tractability)
and have to use the fact that counting $k$-matchings is hard.

If $\H$ has bounded treewidth, then the Ramsey argument mentioned in
the discussion of hereditary classes shows (as the members of $\H$
cannot contain large cliques and complete bipartite graphs) that there
are graphs in $\H$ containing large induced matchings. Our goal is to
use these large induced matchings to reduce counting $k$-matchings in
bipartite graphs to $\pSub(\H)$. Suppose that there is a graph $H\in
\H$ such that $V(H)$ has a partition $(X,Y)$ where $H[Y]$ is a
$k$-matching. By simple inclusion/exclusion arguments, it is
sufficient to prove hardness for the more general problem where we
count only those subgraphs of $G$ isomorphic to $H$ that contain
certain specified vertices/edges of $G$. This suggests the following
reduction: let us extend $G$ to a graph $G'$ by introducing a copy of
$H[X]$ fully connected to every original vertex of $G$ and then
consider the problem of counting subgraphs of $G'$ isomorphic to $H$
that contain every vertex and edge of this copy of $H[X]$. As $H[Y]$
is a $k$-matching (that is, attaching to a $H[X]$ a $k$-matching in a certain way
extends it to $H$), any $k$-matching of $G$ can be used to extend
the copy of $H[X]$ to a subgraph of $G'$ isomorphic to $H$.  It could seem now
that the number of subgraphs of $G'$ isomorphic to $H$ and containing
$H[X]$ is exactly the number of $k$-matchings in $G$. Unfortunately,
this is not true in general due to a seemingly unlikely problem: if we
extend $H[X]$ to a copy of $H$, then it is not necessarily true that
the extension forms a $k$-matching. That is, it is possible that
$V(H)$ has another partition $(X',Y')$ such that $H[X']$ is isomorphic to
$H[X]$, but $H[Y']$ is not a $k$-matching. While this can be perhaps
considered counterintuitive, there are very simple examples where this
can happen. Consider, for example, the graph $H$ on vertices $a$, $b$, $c$, $d$, where any two vertices are adjacent, except $a$ and $d$. Now $X=\{a,b\}$ and $Y=\{c,d\}$ is a partition where $H[Y]$ is an edge. Consider now the partition $X'=\{b,c\}$, $Y'=\{a,d\}$. We have $H[X]\simeq H[X']$, but $H[Y']$ contains two independent vertices. (The reader may easily find larger examples of this flavor, for example, by taking several disjoint copies of $H$.) Arguing about the isomorphism of extensions of
graphs is notoriously counterintuitive: for example, the
reconstruction conjecture of Kelly \cite{MR0087949} and Ulam
\cite{MR0120127} (the {\em deck} of graph is the multiset of graphs
obtained by removing one vertex in every possible way; the conjecture
says that if two graphs have the same deck, then they are isomorphic)
has been open for more than 50 years.

Our goal is to find graphs $H\in \H$ and partitions $(X,Y)$ where the
problem described in the previous paragraph does not occur. We say
that $H\in \H$ and a partition $(X,Y)$ is a {\em $k$-matching gadget} if $H[Y]$
is a $k$-matching, and whenever $(X',Y')$ is a partition of $V(H)$
such that $H[X]\simeq H[X']$ and $H[Y']$ satisfies some technical
conditions that we enforce in the reduction (such as $H[Y']$ is
bipartite and has no isolated vertices), then $H[Y']$ is also a
$k$-matching. If the class $\H$ has such $k$-matching gadgets for every $k\ge 1$,
then we can reduce counting $k$-matchings to $\pSub(\H)$ with a
reduction similar to what was sketched in the previous
paragraph (Section~\ref{sec:reduc-match-psubh}). We prove the existence of $k$-matching gadgets in $\H$ by a
detailed graph-theoretic study, where we first consider bounded-degree
graphs (Section~\ref{sec:bound-degr-graphs}), then move on to graphs that have unbounded degree, but do not contain
large subdivided stars (Section~\ref{sec:graphs-with-no}), and then finally consider graphs where only
the treewidth is bounded (Section~\ref{sec:bound-treew-graphs}). Together with the hardness proof for classes with unbounded treewidth (Section~\ref{sec:unbo-treew-graphs}) and the algorithm for bounded vertex-cover number (Section~\ref{sec:algo-vc}), this completes the proof of Theorem~\ref{th:main}.

\section{Preliminaries}

If $A$ is a set, we will sometimes write $\# A := |A|$ for the cardinality of $A$.
For $\ell \in \mathbb N$ and an indeterminate $x$, let $(x)_\ell := (x)(x-1)\ldots (x-\ell+1)$ denote the falling factorial.

In this paper, graphs are undirected, unweighted and simple, unless stated otherwise.
We write $H \simeq H'$ if $H$ and $H'$ are isomorphic.
If $\H$ is a class of graphs and $f$ is a function from graphs to $\mathbb N$, such as the vertex-cover number $\tau(H)$, then we call $f$ \emph{bounded} on $\H$ if there is a fixed $b\in\mathbb N$ such that every $H\in \H$ satisfies $f(H) \leq b$. Otherwise, we call $f$ \emph{unbounded} on $\H$. 

The graph $H$ is a \emph{minor}
of $G$, written $H\preceq G$, if $H$ can be obtained from $G$
by edge/vertex-deletions and edge-contractions. The contraction of an edge
$uv\in E(G)$ identifies $u,v\in V(G)$ to a single vertex adjacent to
the union of the neighborhoods of $u$ and $v$ in $G$. Equivalently, $H$ is a minor of $G$ if it has a {\em minor model} in $G$, which is an assignment of a {\em branch set} $B_v\subseteq V(G)$ of vertices to every $v\in V(H)$ such that these sets are pairwise disjoint, $G[B_v]$ is connected, and if $uv$ is an edge of $H$, then there is at least one edge between $B_u$ and $B_v$ in $G$.

\begin{definition}\label{def:tw}
A \emph{tree decomposition} of a
graph $G$ is a pair $(T,\B)$ in which $T$ is a tree\footnote{We often assume that $T$ is rooted.} and
$\B=\{B_t\:|\:t\in V(T)\}$ is a family of subsets of $V(G)$ such that
\begin{enumerate}
\item $\bigcup_{t\in V(T)}B_i = V$; 
\item for each edge $e=uv\in E(G)$, there
exists a $t\in V(T)$ such that both $u$ and $v$ belong to $B_t$; and
\item the set of nodes $\{t\in V(T)\:|\:v\in B_t\}$
forms a connected subtree of $T$ for every $v\in V(G)$.
\end{enumerate}
To distinguish between vertices of the original graph $G$ and
vertices of $T$ in the tree decomposition, we call vertices of $T$
\emph{nodes} and their corresponding $B_i$'s \emph{bags}.  The \emph{width}
of the tree decomposition is the maximum size of a bag in $\B$ minus
$1$.  The \emph{treewidth} of a graph $G$, denoted by $\tw(G)$, is the
minimum width over all possible tree decompositions of $G$. 
\end{definition}

For the purpose of this paper, parameterized problems are problems that ask for some output on input $(x,k)$, where $x$ is an instance and $k\in \mathbb N$ is a parameter.
A problem is \emph{fixed-parameter tractable} (FPT) if it admits an algorithm with runtime $f(k)n^{\mathcal O(1)}$ for a computable function $f$.
For parameterized problems $A,B$, we write $A \leqT B$ if $A$ admits a parameterized Turing reduction to $B$, i.e., 
given oracle access for $B$, we can solve an instance $(x,k)$ to $A$ in time $f(k)n^{\mathcal O (1)}$, calling the oracle only on queries $(y,k')$ with $k'\leq g(k)$. Here, both $f$ and $g$ are computable functions. We write $\leqLin$ if such a reduction exists with $g \in \mathcal O (k)$. It is known that if $A\leqT B$ and $B$ is FPT, then it follows that $A$ is FPT as well.
For our purposes, a parameterized problem $A$ is $\sharpWone$-hard if there is a reduction $\pClique \leqT A$, where $\pClique$ is the problem of counting $k$-cliques in a graph $G$ on input $(G,k)$. It is a standard assumption of complexity theory that $\mathsf{FPT} \neq \sharpWone$, parallel to the classical assumption that $\mathsf{P} \neq \mathsf{\#P}$.
Thus, assuming $\mathsf{FPT} \neq \sharpWone$, no $\sharpWone$-hard problem admits an FPT-algorithm.

It is known that if $A\leqLin B$ and $B$ can be solved in time $h_1(k)\cdot n^{h_2(k)}$ for some computable functions $h_1, h_2$, then $A$ can be solved in time $h_3(k)\cdot n^{\mathcal O (h_2(k))}$, that is, with the same asymptotic growth in the exponent of $n$. This fact can be used to transfer lower bounds on the exponent of $n$: if $A\leqLin B$ and $B$ has no $f(k)n^{o(h(k))}$ algorithm for any computable function $f$, then $A$ does not have such an algorithm either.

\begin{definition}
\label{def:uncolored-problems}Let $\mathcal H$ be a class of graphs, and let $H,G$ be graphs.
\begin{enumerate}
\item Let $\Sub(H\to G)$ denote the set of all (not necessarily induced) subgraphs $F\subseteq G$ with $F \simeq H$.
The problem $\pSub(\H)$ asks, given as input a graph $H\in\H$ and an arbitrary graph $G$, for the number $\#\Sub(H\to G)$. The parameter in this problem is $|V(H)|$.
\item A {\em subgraph embedding} of $H$ into $G$ is an injective function $f: V(H) \to V(G)$ such that $uv\in E(H)$ implies $f(u)f(v)\in E(G)$.
Let $\Emb(H\to G)$ denote the set of all subgraph embeddings of $H$ into $G$.
The problem $\pEmb(\mathcal{H})$ is defined as follows:
On input $H\in\mathcal{H}$ and $G$, we ask for $\#\Emb(H\to G)$. The parameter in this problem is $|V(H)|$.
\item In $\pMatch$, we are given a bipartite graph $G$ and $k\in\mathbb N$ and ask for $\#\Sub(M_k\to G)$, where $M_k$ denotes the matching of size $k$, i.e., the $1$-regular graph on $2k$ vertices with $k$ edges. The parameter in this problem is $k$.
\end{enumerate}
\end{definition}

\begin{remark}
\label{rem:sub-auto-emb}
Observe that the elements of $\Emb(H\to H)$ are exactly the automorphisms of $H$,
i.e., the isomorphisms from $H$ to $H$.
Therefore $\#\Emb(H\to G) = \#\Emb(H\to H) \cdot \#\Sub(H\to G)$ for all graphs $H,G$.
We can thus solve the problem $\pSub$ with two oracle calls to $\pEmb$. This means that it is sufficient to prove hardness results for $\pSub$, as this also implies hardness for $\pEmb$.
\end{remark}

\subsection{Colored graphs}
In the subsequent arguments, we will sometimes count occurrences of
colored graphs $H$ within colored graphs $G$.  While such problems can
be defined in full generality and indeed lead to problems and
questions that are interesting on their own right, here we treat
problems on colored graphs only as technical tools helpful in
obtaining results for problems on uncolored graphs.  Therefore, we
chose to limit our exposition to the two specific settings occurring
in this paper. Firstly, we will count copies of vertex-colored graphs
$H$ within vertex-colored graphs $G$, where each vertex of $H$ has a
distinct color. Secondly, we will count edge-colored matchings $M$ in
edge-colored graphs $G$.
\begin{definition}
\label{def:colored-graphs}Let $\Gamma$ be a set of colors. A \emph{colored
graph} is a graph $G$ together with a coloring $c_G:V(G)\to\Gamma$ or $c_G:E(G)\to\Gamma$. In the
first case, we call $G$ vertex-colored, otherwise edge-colored. 
For $\gamma\in\Gamma$, let $V_{\gamma}(G)$ denote the set of all
$\gamma$-colored vertices of $G$. For $S\subseteq\Gamma$, let $V_{S}(G):=\bigcup_{\gamma\in S}V_{\gamma}(G)$. Define $E_\gamma$ and $E_S$ likewise.

We call $G$ \emph{colorful} if $c_G$ is bijective. In such cases,
it will be convenient to identify $\Gamma$ with $V(G)$ or $E(G)$,
depending on whether $G$ is vertex- or edge-colored.

Two $\Gamma$-colored graphs $H$ and $H'$ are \emph{color-preserving isomorphic}
if there is an isomorphism from $H$ to $H'$ that maps each $\gamma$-colored vertex (or edge) of $H$ to a $\gamma$-colored vertex (or edge) of $H'$.
\end{definition}

The following counting problems associated with colored graphs will occur in the paper. 

\begin{definition}
\label{def:colored-problems}
\begin{enumerate} 
\item 
For $\Gamma$-vertex-colored graphs $H,G$ with colorful $H$,
let $\SubPart(H\to G)$ denote the set of all subgraphs $F\subseteq G$ such that $F$ is color-preserving isomorphic to $H$.
Given a class $\H$ of uncolored graphs, the problem $\pRestr(\H)$ asks for $\#\SubPart(H\to G)$, where $H$ is a $\Gamma$-vertex-colorful graph whose underlying uncolored graph is contained in $\H$, and $G$ is $\Gamma$-vertex-colored. The parameter is $|V(H)|$.
\item For a $\Gamma$-edge-colored graph $G$ and $X\subseteq \Gamma$, let $\mathcal M_X[G]$ denote the set of all $X$-colorful matchings in $G$, i.e., matchings in $G$ that choose exactly one edge from each color in $X$. 
In $\pColMatch$, we are given a \emph{bipartite} $\Gamma$-edge-colored graph $G$ and $X \subseteq \Gamma$ and 
ask for $\#\mathcal M_X[G]$. The parameter is $|X|$.
\end{enumerate}
\end{definition}

Note that $\pRestr(\H)$ is defined for a class $\H$ for uncolored graphs, while its inputs are vertex-colored graphs.

\begin{remark}
\label{rem:partitioned-colorful}
Let $H,G$ be $\Gamma$-vertex-colored graphs and let $F$ be a subgraph of $G$ that is color-preserving isomorphic to $H$. If $uv\in E(F)$ is an edge with endpoints of color
$\gamma_u,\gamma_v\in\Gamma$, then there is an edge between vertices
of colors $\gamma_u,\gamma_v$ in $E(H)$.  We may therefore assume
that, whenever $uv\in E(G)$ is an edge with endpoints of color
$\gamma_u,\gamma_v\in\Gamma$ in $G$, then $\{\gamma_u,\gamma_v\}\in
E(H)$.  In other words, we may assume that $G$ has edges between two
color classes if $H$ has an edge with endpoints of this color,
otherwise the edges between the classes are clearly useless.
\end{remark}

The principle of inclusion and exclusion will be an important ingredient of reduction between the colored and the uncolored versions of the problems defined above. It will always be invoked in the following form: Given
a set $\Omega$ and $\mathcal{A}_{1},\ldots,\mathcal{A}_{k}\subseteq\Omega$,
we are interested in the cardinality of $\Omega\setminus\bigcup_{i\in[k]}\mathcal{A}_{i}$.
It is a well-known fact that
\begin{equation}
\label{eq:incl-excl}
\left|\Omega\setminus\bigcup_{i\in[k]}\mathcal{A}_{i}\right|=|\Omega|+\sum_{t=1}^{k}(-1)^{t}\sum_{1\leq i_{1}<\ldots<i_{t}\leq k}|\mathcal{A}_{i_{1}}\cap\ldots\cap\mathcal{A}_{i_{t}}|.
\end{equation}
Note that we can apply \eqref{eq:incl-excl} only if we have an
efficient way of computing the size of the intersections
$\mathcal{A}_{i_{1}}\cap\ldots\cap\mathcal{A}_{i_{t}}$. As a first demonstration of this
principle, we obtain a reduction from the colorful problem to the
uncolored problem.

\begin{lemma}
\label{lem:colorful2uncolored}
The following reductions between colored and uncolored problems hold:
\begin{enumerate}
\item $\pRestr(\mathcal{H})\leqLin\pSub(\mathcal{H})$, for any class $\mathcal{H}$.
\item $\pColMatch \leqLin \pMatch$.
\end{enumerate}
\end{lemma}
\begin{proof}

For the first statement, let $H$ with $H \in \mathcal C$ be $\Gamma$-vertex-colorful with $\Gamma = V(H)$, and let $G$ be $\Gamma$-vertex-colored by $c_G : V(G) \to \Gamma$. Assume oracle access for $\#\Sub(H\to G')$ for graphs $G'\subseteq G$. 
The parameter trivially remains unchanged when making calls to this oracle.

By Remark \ref{rem:partitioned-colorful}, assume that the endpoints of
every $e\in E(G)$ have colors $\gamma,\gamma'\in \Gamma$ with
$\{\gamma,\gamma'\}\in E(H)$. We claim that $F\in\SubPart(H\to G)$ if
and only if (i) $F$ is isomorphic to $H$ when ignoring vertex-colors
of both graphs, and (ii) $F$ is colorful under $c_G$. The forward
direction is trivial. To see that (i) and (ii) imply that $F\in
\SubPart(H\to G)$, observe that if $F$ is colorful under $c_G$, then
$F$ can have at most one edge between any two color classes of $G$. By
the assumption of Remark~\ref{rem:partitioned-colorful}, there are at
most $|E(H)|$ pairs of color classes in $G$ having an edge between
them. Therefore, $|E(F)|=|E(H)|$ is possible only if $F$ contains an
edge between each such color class. Now we obtain a color-preserving
isomorphism from $H$ to $F$ by mapping the vertex of $H$ with color
$\gamma\in \Gamma$ to the unique vertex of $F$ of color $\gamma$.

We use inclusion-exclusion to count subgraphs $F\subseteq G$
satisfying (i) and (ii).  For $S\subseteq\Gamma$, let
$\mathcal{A}_{S}:=\Sub(H\to G[V_{S}])$, that is, the copies of $H$ in
$G$ using only vertices whose color is contained in $S$. Observe that we can compute
$|\mathcal{A}_{S}|$ for $S\subseteq\Gamma$ by an oracle call to
$\#\Sub(H\to G[V_{S}])$.  A subgraph $F\in\Sub(H\to G)$ satisfies (i)
by definition, and (ii) if and only if
$F\in\mathcal{A}_{\Gamma}\setminus\bigcup_{S\subsetneq\Gamma}\mathcal{A}_{S}$.  This
allows to compute $\#\SubPart(H\to G)$ by inclusion-exclusion using
\eqref{eq:incl-excl}. To determine the size of $\Sub_{S}\cap\Sub_{T}$,
which is needed in \eqref{eq:incl-excl}, note that
$\mathcal{A}_{S}\cap\mathcal{A}_{T}=\mathcal{A}_{S\cap T}$.

The second statement is shown in a similar (but simpler) way without using Remark \ref{rem:partitioned-colorful}: 
If $G$ is $\Gamma$-edge-colored and we wish to compute $\#\mathcal M_X[G]$, then define $\mathcal{A}_{S}:=\Sub(H\to G[E_{S}])$ for $S\subseteq X$.
An uncolored matching $M$ of size $|X|$ in $G$ chooses exactly one edge from each color in $X$ if and only if
$M\in\mathcal{A}_{X}\setminus\bigcup_{S\subsetneq X}\mathcal{A}_{S}$. This allows to invoke inclusion-exclusion as before.
\end{proof}

\subsection{Bounded vertex-cover number}
\label{sec:algo-vc}
We conclude the preliminaries with a simple self-contained polynomial-time algorithm for determining $\# \Sub(H\to G)$ in time polynomial in $|V(H)|$ and $|V(G)|$ when the vertex-cover number $\tau(H)$ (or equivalently, the size of the largest matching $\nu(H)$) can be assumed to be constant.
As already stated in the introduction, more efficient algorithms are known \cite{DBLP:journals/siamcomp/WilliamsW13, DBLP:journals/siamdm/KowalukLL13}.
We include the following theorem only for sake of completeness.

\begin{theorem}
\label{thm:algo-vc}Let $H$ be a graph on $k$ vertices with vertex-cover number $\tau=\tau(H)$ 
and let $G$ be a graph on $n$ vertices. Then we can compute $\#\Emb(H\to G)$ and $\#\Sub(H\to G)$ in time
$k^{2^{\mathcal{O}(\tau)}}n^{\tau+\mathcal{O}(1)}$.
\end{theorem}
\begin{proof}
Let $C=\{c_{1},\ldots,c_{\tau}\}$ be a vertex cover of $H$. 
For every $X\subseteq C$, let $R^{H}_X$ be the set of vertices in $V(H)\setminus C$ with $N_{H}(v)=X$. 
Note that $\sum_{X\subseteq C}|R^H_X|=k-\tau$.
For $\mathbf{s}=(s_{1},\ldots,s_{\tau})$ with $s_{i}\in V(G)$ for each $i\in[\tau]$, let 
\[
\mathcal{A}_{\mathbf{s}}=\{f\in\mathsf{Emb}(H\to G)\mid\forall i\in[\tau]:f(c_{i})=s_{i}\},
\]
that is, the set of all subgraph embeddings that map the vertices of
$C$ as prescribed by $\mathbf{s}$ (note that if a vertex $v\in V(G)$
appears more than once in $\mathbf{s}$, then clearly
$\mathcal{A}_\mathbf{s}=\emptyset$).  Since the sets
$\mathcal{A}_{\mathbf{s}}$ partition $\Emb(H\to G)$, it suffices to
compute $\#\mathcal{A}_{\mathbf{s}}$ for each $\mathbf{s}$.  Then
$\#\Emb(H\to G)=\sum_{\mathbf{s}}\#\mathcal{A}_{\mathbf{s}}$, where
the sum is over the $n^{\tau}$ tuples
$\mathbf{s}=(s_{1},\ldots,s_{\tau})$ with $s_{i}\in V(G)$ for
$i\in[\tau]$.

We show how to compute $\#\mathcal{A}_{\mathbf{s}}$ in time $k^{2^{\mathcal{O}(\tau)}}n^{\mathcal O (1)}$, which implies the claimed total runtime. 
Since $V(H)\setminus C$ is an independent set, we can safely delete all edges in $G$ that are not incident with any $s_{i}$ for $i\in[\tau]$. 
The resulting graph $G'$ has the vertex cover $S=\{s_{1},\ldots,s_{\tau}\}$. 
For every $Y\subseteq S$, let $R^G_Y$ be the set of vertices in $V(G')\setminus S$ with $N_{G'}(v)=Y$.

Construct a bipartite directed graph $I$ on $2^{\tau}+2^{\tau}$
vertices, with a left vertex $\ell_{Y}$ for each $Y\subseteq S$
and a right vertex $r_{X}$ for each $X\subseteq C$. 
Identify $S$ with $C$ by $c_{i}\simeq s_{i}$ for $i\in[\tau]$, and 
for $X,Y$ with
$X\subseteq Y$, add the edge $(\ell_{Y},r_{X})$ to $I$. Intuitively, the meaning of this edge is that any vertex of $R^H_X$ can be mapped to any vertex of $R^G_Y$.
Considering
$|R^G_Y|$ as the supply of $\ell_{Y}$ and $|R^H_Y|$ 
as the demand of $r_{X}$, let $\mathcal{F}$ denote the set of
all feasible integral flows $h:E(I)\to\mathbb{N}$ in $I$ that exactly
satisfy the demands. 

As the total demand is $k-\tau\leq k$ and $I$ has $t=2^{\mathcal O (\tau)}$ edges, we have
$|\mathcal{F}|\leq{k+t-1 \choose
t-1} \leq k^{2^{\mathcal O (\tau)}}$ as every feasible integral flow represents a way to choose a multiset of $k-\tau$ elements among $t$ elements.
We can thus enumerate $\mathcal{F}$ by brute force.  If the integral flow has
value $m$ on the edge $(\ell_{Y},r_{X})$, then this corresponds to
mapping $m$ vertices of $R^H_X$ to $R^G_Y$. The number of ways this is
possible is given by the falling factorial expression
$(|R^G_X|)_{m}$. Therefore, it can be verified that
\[
|\mathcal{A}_{\mathbf{s}}|=\sum_{h\in\mathcal{F}}\;\prod_{(\ell_{X},r_{Y})\in E(I)}(|R^G_X|)_{h(\ell_{X},r_{Y})}.
\]
Hence $|\mathcal{A}_{\mathbf{s}}|$ can be computed in time $k^{2^{\mathcal O (\tau)}}n^{\mathcal{O}(1)}$. 
The statement for $\#\Sub(H\to G)$ follows by Remark \ref{rem:sub-auto-emb}.
\end{proof}

\section{Unbounded-treewidth graphs}
\label{sec:unbo-treew-graphs}

In this section, we recall techniques and results underlying the
hardness proofs for finding graphs of large treewidth. In particular,
we prove (Theorem~\ref{thm:hard-treewidth}) that $\pRestr(\mathcal H)$
is $\sharpWone$-hard whenever $\mathcal H$ has unbounded treewidth,
i.e., if for every $b\in\mathbb N $ there is some $H \in \mathcal H$
of treewidth at least $b$. By Lemma~\ref{lem:colorful2uncolored}(1),
the same hardness result follows for $\pSub(\mathcal H)$, proving
Theorem~\ref{th:main} for classes with unbounded treewidth.  As
already stated in the introduction, the proof uses standard techniques
and could in fact be adapted from ideas in
\cite{380867,1206036,DBLP:journals/tcs/DalmauJ04,DBLP:conf/icalp/ChenTW08}.
We nevertheless include it for sake of completeness.

First, we prove in Lemma \ref{lem:reduction-minor} that if $H$ is a
minor of $H^\dagger$ and both graphs are colored in an arbitrary
vertex-colorful way, then $\#\SubPart(H \to G)$ can be computed from
$\#\SubPart(H^\dagger \to G^\dagger)$, where $G^\dagger$ is a graph
constructed from the graphs $G$, $H$ and $H^\dagger$.  Then we invoke
this lemma on grids $H$: By an argument similar to how the
\Wone-hardness of $\pGridTiling$ is proved (see, e.g., \cite[Lemma
1]{DBLP:conf/icalp/Marx12}), counting colored $k\times k$ square grids
is $\sharpWone$-complete. Then the Excluded Grid Theorem of
Robertson and Seymour~\cite{MR854606}, which asserts that graphs
classes of unbounded treewidth contain arbitrarily large grid minors,
implies Theorem~\ref{thm:hard-treewidth}.

As a further application of the machinery developed in this section,
we show that $\pRestr$ is $\sharpWone$-hard on the class of
$3$-regular bipartite graphs, and using a result of
Marx~\cite{marx-toc-treewidth}, we also show that this
problem admits no $f(k)n^{o( k / \log k)}$ algorithm with $k = |V(H)|$,
assuming ETH.  In the next section, this will be used as the source
problem in the reduction for showing \sharpWone-hardness of counting
$k$-matchings.

\subsection{Minors}
If $H$ is a minor of $H^{\dagger}$, then computing $\pRestr(H\to G)$ can
be reduced to computing $\pRestr(H^{\dagger}\to G^{\dagger})$ for a
colored graph $G^{\dagger}$ constructed in an appropriate
way. Essentially, if a branch set $B_i\subseteq V(H^{\dagger})$
corresponds to a vertex of $H$ having color $i$, then each vertex of
$G$ having color $i$ has to be replaced by a copy of
$H^{\dagger}[B_i]$.
\begin{lemma}
\label{lem:reduction-minor}Let $\mathcal{H}$ and $\mathcal{H}^{\dagger}$
be recursively enumerable graph classes such that for every $H\in\mathcal{H}$,
there exists some $H^{\dagger}\in\mathcal{H}^{\dagger}$ with $H\preceq H^{\dagger}$.
Then $\pRestr(\mathcal{H})\leqFpt\pRestr(\mathcal{H}^{\dagger})$.
If additionally $|V(H^{\dagger})|=\mathcal{O}(|V(H)|)$ holds for every $H$,
then $\leqFpt$ can be replaced by $\leqLin$.\end{lemma}
\begin{proof}
Let $H,G$ be $[k]$-colored with colorful $H$, and let $H^{\dagger}\in\mathcal{H}'$
with $H\preceq H^{\dagger}$ and $k^{\dagger}=|V(H^{\dagger})|$
be $[k^{\dagger}]$-colorful. Given $H$, we find $H^\dagger$ by enumerating the graphs in $H^\dagger$ and testing by brute force whether $H\preceq H^\dagger$. The actual colorings of $H,H^{\dagger}$
are irrelevant as long as they are colorful. To compute $\#\SubPart(H\to G)$,
assume that we can compute $\#\SubPart(H^{\dagger}\to G^{\dagger})$
for a specific $[k^{\dagger}]$-colored graph $G^{\dagger}$ that
we construct in the following. Note that this reduction increases the parameter from $k$ to $k^\dagger$.

Since $H\preceq H^{\dagger}$, the set $V(H^{\dagger})$ admits a
partition into branch sets $B_{0},B_{1},\ldots,B_{k}$ such that the
following holds: For $i\in[1,k]$, the graph $H^{\dagger}[B_{i}]$ is
connected, and deleting $B_{0}$ and contracting each $B_{i}$ for
$i\in[k]$ to a single vertex (which we denote by $i$) yields some
supergraph of $H$ on the vertex set $[k]$. Recall that $V_i(G)$
denotes the set of vertices in $G$ with color $i$. Let $G^{\dagger}$
denote the $[k^{\dagger}]$-colored graph obtained from $G$ as follows:
\begin{enumerate}
\item For $i\in[1,k]$ and $v\in V_{i}(G)$: Replace $v$ by a copy
of $H^{\dagger}[B_{i}]$, denoted by $L_{v}$. Note that $H^{\dagger}[B_{i}]$
is a vertex-colorful graph with colors from some subset of $[k^{\dagger}]$.
\item For $\{i,j\}\in E(H)$ and $u\in V_{i}(G), v\in V_{j}(G)$ with $\{u,v\}\in E(G)$: Insert all edges
between $L_{u}$ and $L_{v}$ in $G^{\dagger}$.
\item For $\{i,j\}\notin E(H)$ and $u\in V_{i}(G), v\in V_{j}(G)$:
Insert all edges between $L_{u}$ and $L_{v}$ in $G^{\dagger}$.
\item Add a copy of $H^{\dagger}[B_{0}]$ to $G^{\dagger}$, connect it
to all other vertices of $G^{\dagger}$.
\end{enumerate}
The effect of this transformation is shown in Figure~\ref{fig:minor-model}.
\begin{figure}[t]
\begin{centering}
\includegraphics[width=0.86\textwidth]{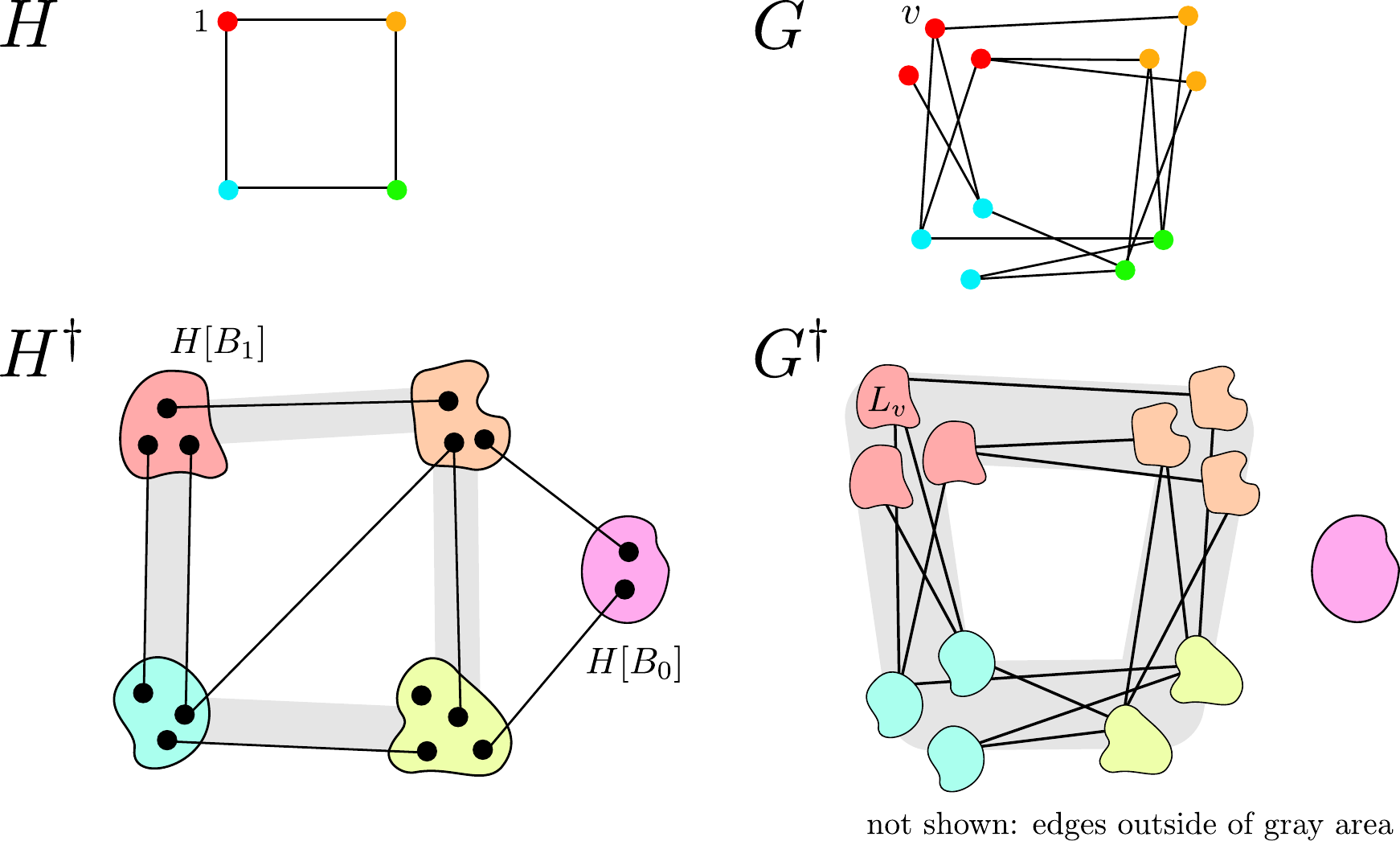}
\par\end{centering}

\caption{\label{fig:minor-model}The $[4]$-vertex-colorful graphs $H\preceq H^{\dagger}$.
From the $[4]$-vertex-colored graph $G$, we construct the graph
$G^{\dagger}$. To avoid clutter, the edges added in Step 3 of constructing
$G^{\dagger}$ are\emph{ not }shown in the figure: Additionally to
shown edges, $G^{\dagger}$ has all possible edges not contained in
the gray area.}
\end{figure}
 We show that $\SubPart(H\to G)\simeq\SubPart(H^{\dagger}\to G^{\dagger})$:
Every $F$ from the left set can be extended to a unique $F^{\dagger}$
from the right side by the graph transformation above. Conversely,
every $F^{\dagger}$ from the right set corresponds to exactly one
$F$ in the left set: Since $F^{\dagger}$ is color-isomorphic to
$H^{\dagger}$, the $B_{i}$-colored vertices of $F^{\dagger}$ induce
a graph $F_{i}^{\dagger}\simeq H[B_{i}]$, for every $i\in[0,k]$.
Since $H^{\dagger}[B_{i}]$ for $i\in[1,k]$ is connected, but $L_{u}$
and $L_{v}$ are vertex-disjoint for different $u,v\in V_{i}(G)$,
there is some $v(i)\in V_{i}(G)$ such that $F_{i}^{\dagger}=L_{v(i)}$.
Applying this to all $i\in[1,k]$ yields a colorful copy of $H$ on
vertices $v(1),\ldots,v(k)\in V(G)$.
\end{proof}

\subsection{Grids}
The $k\times k$ grid $H_{k\times k}$ is a graph with vertex set
$[k]\times [k]$ where two vertices $(i,j),(i',j')\in [k]\times [k]$
are adjacent if and only if $|i-i'|+|j-j'|=1$. We denote by
$\H_\textup{grid}$ the class containing $H_{k\times k}$ for every
$k\ge 1$. We show that $\pRestr(\H_\textup{grid})$ is \sharpWone-hard
by a proof that is essentially the same as how the \Wone-hardness of
$\pGridTiling$ can be proved (see, e.g., \cite{DBLP:conf/icalp/Marx12}).
\begin{theorem}\label{thm:partgrid}
$\pRestr(\mathcal{H}_{\textup{grid}})$ is \sharpWone-hard.
\end{theorem}
\begin{proof}
  We reduce $\pClique$ to $\pRestr(\mathcal{H}_{\textup{grid}})$. Let
  $G$ be a graph where the number of $k$-cliques has to be
  computed. We construct a colored graph $G'$ such that there is a
  one-to-one correspondence between the $k$-cliques of $G$ and the
  colored $k\times k$ grid subgraphs of $G'$.

Let $H_{k\times k}$ be the $k\times k$ grid where the vertex in row $i$ and column $j$ (denote it by $h_{i,j}$) has color $(i,j)$. 
The graph $G'$ is constructed as follows.
\begin{itemize}
\item For every $i\in [k]$ and every $x\in V(G)$, we introduce a vertex $v_{i,i,x,x}$ of color $(i,i)$.
\item For every $i,j\in [k]$, $i\neq j$ and every $x,y\in V(G)$ such that $x\neq y$ and $\{x,y\}\in E(G)$, we introduce a vertex $v_{i,j,x,y}$ of color $(i,j)$.
\item For every $i\in [k]$, $j\in [k-1]$, and $x,y,y'\in V(G)$, if $v_{i,j,x,y}$ and $v_{i,j+1,x,y'}$ both exist in $G'$, then we make them adjacent.
\item For every $i\in [k-1]$, $j\in [k]$, and $x,x',y\in V(G)$, if $v_{i,j,x,y}$ and $v_{i+1,j,x',y}$ both exist in $G'$, then we make them adjacent.
\end{itemize}
This concludes the description of the reduction. We claim that the number of $k$-cliques in $G$ is exactly $\pRestr(H_{k\times k}\to G')$.

Let $a_1$, $\dots$, $a_k$ be the vertices of a $k$-clique in $G$. Then
we can find an $H_{k\times k}$-subgraph in $G'$ by mapping vertex $h_{i,j}$ of the
grid to vertex $v_{i,j,a_i,a_j}$. It can be verified that these
vertices exist and if two vertices in $H_{k\times k}$ are adjacent, then their
images are adjacent in $G'$. For example, $h_{i,j}$ and $h_{i,j+1}$
are adjacent in $H_{k\times k}$, and the corresponding vertices $v_{i,j,a_i,a_j}$
and $v_{i,j+1,a_i,a_{j+1}}$ exist and are adjacent by definition.
Moreover, different $k$-cliques give rise to different $H_{k\times k}$-subgraphs,
thus $\pRestr(H_{k\times k}\to G')$ is at least the number of $k$-cliques in $G$.

Consider now a $H_{k\times k}$-subgraph of $G'$. As $H_{k\times k}$ has exactly one vertex of each
color $(i,j)$, the subgraph contains exactly one vertex of the form
$v_{i,j,x,y}$. As the vertices with color $(i,j)$ and $(i,j+1)$ are
adjacent, they have to be of the form $v_{i,j,x,y}$ and
$v_{i,j,x,y'}$, because only such vertices are adjacent. It follows that
for every $i\in [k]$, there is an $a_i\in V(G)$ such that the vertex
with color $(i,j)$ is of the form $v_{i,j,a_i,y}$. Similarly, by the
requirement that vertices with colors $(i,j)$ and $(i+1,j)$ have to be
adjacent, we get that for every $j\in [k]$, there is a $b_j\in V(G)$
such that the vertex with color $(i,j)$ is of the form
$v_{i,j,x,b_j}$. Therefore, the vertex with color $(i,j)$ is
$v_{i,j,a_i,b_j}$. In particular, for $i\in [k]$, the vertex with
color $(i,i)$ is $v_{i,i,a_i,b_i}$, which only exists if $a_i=b_i$. We
claim now that $a_1$, $\dots$, $a_k$ form a clique. Indeed, to see
that $a_i$ and $a_j$ are distinct and adjacent, observe that the
vertex with color $(i,j)$ is $v_{i,j,a_i,b_j}=v_{i,j,a_i,a_j}$, and
the fact that it exists implies that $a_i$ and $a_j$ are distinct and
adjacent in $G$. As it is also true that distinct subgraphs give rise
to distinct $k$-cliques in $G$ (as changing any vertex $v_{i,j,a_i,b_i}$ would change $a_i$ or $b_j=a_j$), we get that the number of $k$-cliques is at
least $\pRestr(H_{k\times k}\to G')$.  Putting together the two inequalities, we
get the required equality.
\end{proof}

The Excluded Grid Theorem, first proved by Robertson and Seymour \cite{MR854606}, shows that every graph with sufficiently large treewidth contains the grid $H_{k\times k}$ as a minor.

\begin{theorem}\label{thm:excluded-grid}
  For every $k\ge 1$, there is an integer $b(k)\ge 1$ such that every
  graph of treewidth at least $b(k)$ contains the $k\times k$ square
  grid $H_{k\times k}$ as a minor.
\end{theorem}
In the original proof of Robertson and Seymour \cite{MR854606}, as well as in the
improved proof by Diestel et
al.~\cite{DBLP:journals/jct/DiestelJGT99}, the function $b(k)$ is
exponential in $k$. Very recently, Chekuri and Chuzhoy
\cite{DBLP:journals/corr/abs-1305-6577} obtained a proof where $b(k)$ is
polynomial in $k$. However, for our application, the growth rate of
the function $b(k)$ is immaterial.

The hardness result for classes with unbounded treewidth can be obtained by a simple combination of Theorems~\ref{thm:partgrid} and \ref{thm:excluded-grid}.
\begin{theorem}
\label{thm:hard-treewidth}The problems $\pRestr(\mathcal{H})$ and
$\pSub(\mathcal{H})$ are $\sharpWone$-complete whenever $\mathcal{H}$
is recursively enumerable and has unbounded treewidth. \end{theorem}
\begin{proof}
  Since $\mathcal{H}$ has unbounded treewidth,
  Theorem~\ref{thm:excluded-grid} shows that for every $k\ge 1$, there
  exists some $H\in\mathcal{H}$ with $H_{k\times k}\preceq
  H$. Therefore, by Lemma~\ref{lem:reduction-minor},
  $\pRestr(\mathcal{H}_\textup{grid})$, which is \sharpWone-hard by
  Theorem~\ref{thm:partgrid}, can be reduced to
  $\pRestr(\mathcal{H})$. This proves the claim that $\pRestr(\mathcal{H})$ is \sharpWone-hard.
  Then the claim for $\pSub(\mathcal{H})$ follows by
  Lemma~\ref{lem:colorful2uncolored}.
\end{proof}
It was shown by Arvind and Raman \cite[Lemma 1]{DBLP:conf/isaac/ArvindR02} that the number $\#\SubPart(H\to G)$ can be computed in time  $\mathcal O (c^{b^3}k+n^{b+2}2^{b^2/2})$, where $b$ is the treewidth of $H$. Therefore, $\pRestr(\mathcal H)$ is polynomial-time solvable if $\H$ has bounded treewidth. 
Together with our $\sharpWone$-hardness result, this yields a dichotomy for $\pRestr(\mathcal H)$. Note that this algorithm for the bounded-treewidth cases of $\pRestr(\H)$ does not settle the same question for $\pSub$: the reduction in Lemma~\ref{lem:colorful2uncolored}(1) goes the opposite direction. In fact, there are bounded-treewidth classes $\H$, most notably, matchings and paths, for which $\pRestr(\H)$ is polynomial-time solvable, but $\pSub(\H)$ is \sharpWone-hard. It is precisely the bounded-treewidth classes where the complexity of the two problems can deviate.

\subsection{Bipartite 3-regular graphs}

In Section~\ref{sec:bipart-edge-colorf}, the \sharpWone-hardness proof
for bipartite $k$-matching is by a reduction from $\pRestr$. It is
essential for the hardness proof that the graph $H$ appearing in the
$\pRestr$ instance is bipartite and 3-regular. Therefore, we establish here the \sharpWone-hardness of $\pRestr(\H_{\textup{bicub}})$, where $\H_{\textup{bicub}}$ is the class of all bipartite cubic graphs.
\begin{lemma}\label{lem:minor-bicubic}
If $H$ is a graph on $n$ vertices, none of which are isolated, then there exists a bipartite $3$-regular graph $H^{\dagger}$ with $|V(H^{\dagger})|=O(n)$ such that $H$ is a minor of $H^{\dagger}$.
\end{lemma}
\begin{proof}
  First, if $v \in V(H)$ has degree $t < 3$, then attach one of the
  gadgets appearing in Figure~\ref{fig:cubic_gadget} to $v$ so as to
  increase its degree to $3$.
Then replace every vertex $v\in V(H)$ of degree $t > 3$ by a cycle of length $t$ and,
for all $i\in[t]$, attach the $i$-th edge incident with $v$ to the
$i$-th cycle vertex. The resulting graph is $3$-regular and thus
has $3t$ edges for some $t\in\mathbb{N}$. Subdivide this graph and obtain
$3t$ vertices of degree $2$. Add $t$ vertices, each of which connects
to three of the $3t$ degree-$2$ vertices. It is easy to see that the constructed graph $H^{\dagger}$ has $O(|E(H)|)$ edges and contains $H$ as minor: 
one needs to reverse the subdivisions by contractions, delete all the additionally introduced vertices, and contract each cycle to single vertex.
\begin{figure}[t]
\begin{centering}
\includegraphics[width=0.3\textwidth]{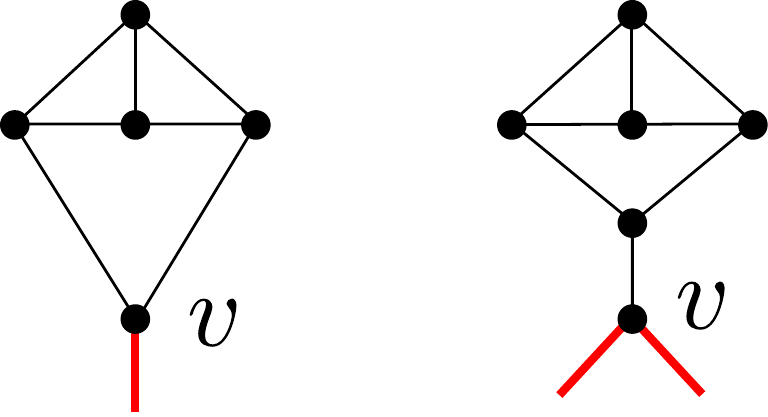}
\par\end{centering}
\caption{\label{fig:cubic_gadget}Handling a degree-1 or a degree-2 vertex in the proof of Lemma~\ref{lem:minor-bicubic}.}\end{figure}
\end{proof}
By Lemma~\ref{lem:minor-bicubic}, for every class $\H$, every $H\in
\H$ appears as the minor of some graph $H^{\dagger}\in
\H_\textup{bicub}$, hence $\pRestr(\H)$ can be reduced to
$\pRestr(\H_\textup{bicub})$.
As shown in Theorem~\ref{thm:partgrid}, the problem $\pRestr(\mathcal{H}_{\textup{grid}})$ is $\sharpWone$-hard.\footnote{The problem $\pRestr(\mathcal{K})$ on the class $\mathcal K$ of cliques could also be used here. This problem admits a simple self-contained $\sharpWone$-hardness proof from counting $k$-cliques.}
Thus, we obtain:
\begin{lemma}\label{lem:bicubic-hard}
  $\pRestr(\H_\textup{bicub})$ is \sharpWone-hard.
\end{lemma}
It is known that, assuming ETH, $\pClique$ cannot be solved in time
$f(k)n^{o(k)}$ for any computable function $k$
\cite{Chen20061346,survey-eth-beatcs}.  We would like to have a
similar lower bound for $\pRestr(\H_\textup{bicub})$ and then, via the
reduction in Section~\ref{sec:bipart-edge-colorf}, a lower bound for
counting bipartite $k$-matchings.  Note, however, that if $H$ is a
$k$-clique, then the graph $H^{\dagger}$ constructed in
Lemma~\ref{lem:minor-bicubic} has $O(k^2)$ vertices and edges.
Therefore, the additional requirement $|V(H^{\dagger})|=O(|V(H)|)$ of
Lemma~\ref{lem:reduction-minor} {\em does not} hold, and hence we {\em
  cannot} conclude that $\pClique \leqLin
\pRestr(\H_\textup{bicub})$. This means that this reduction is not sufficiently 
strong to prove that, assuming ETH,
$\pRestr(\H_\textup{bicub})$ cannot be solved in time $f(|V(H)|)\cdot
n^{o(|V(H)|)}$.

  We need a source problem different from $\pClique$ to prove (almost) tight lower
  bounds for $\pRestr(\H_\textup{bicub})$. The following
  result establishes a lower bound that holds for $\pRestr$ even if $H$ has bounded degree.
\begin{theorem}[{\cite[Corollaries 6.2--6.3]{marx-toc-treewidth}}]\label{thm:marxtreewidth}
Assuming ETH, there is a universal constant $D$ such that $\pRestr$ cannot be solved in time $f(k)n^{o(k/\log k)}$, where $k=|V(H)|$ and $f$ is any computable function, even under the restriction that $H$ has maximum degree at most $D$.
\end{theorem}
Now if $\H_D$ is the class of all graphs with maximum degree $D$, then for any $H\in \H_D$, Lemma~\ref{lem:minor-bicubic} constructs a graph $H^{\dagger}$ with $O(|E(H)|)=O(|V(H)|)$ edges (as $D$ is a universal constant). Therefore, Lemma~\ref{lem:reduction-minor} shows that $\pRestr(\H_D)\leqLin \pRestr(\H_\textup{bicup})$ holds. Together with Theorem~\ref{thm:marxtreewidth}, we get the following lower bound.
\begin{lemma}
  \label{lem:hard-bicubic}Assuming ETH, the problem
  $\pRestr(\mathcal{H}_{\textup{bicub}})$ admits no $f(k)n^{o(k/\log
    k)}$ time algorithm, where $k=|V(H)|$ and $f$ is any computable function.\end{lemma}

\section{Bipartite edge-colorful matchings}
\label{sec:bipart-edge-colorf}

In this section, we prove $\sharpWone$-hardness of counting
$k$-matchings in bipartite graphs $G$.  While this is interesting on
its own, as previously only $\sharpWone$-hardness for general graphs
$G$ was known, we mainly use this problem as a reduction source for
the next section, where it will be crucial to assume that $G$ is
bipartite.  In fact, we prove the stronger statement that counting
edge-colorful $k$-matchings is $\sharpWone$-hard (by
Lemma~\ref{lem:colorful2uncolored}(2), this statement is indeed
stronger).  This might come as a surprise as the vertex-colorful
version is fixed-parameter tractable (even on general graphs) by the
discussion in the last section.

Furthermore, our reduction bypasses the algebraic machinery of \cite{DBLP:conf/icalp/Curticapean13}, which built upon a technique introduced in \cite{MR2065338} that could only guarantee that the parameter increase in the reduction is computable.
Therefore, while showing $\sharpWone$-hardness, this proof was inherently unable to show lower bounds under ETH.
In the following proof, we reduce from the problem $\pRestr(\mathcal{H}_{\textup{bicub}})$ from the last section,
which was shown in Lemma~\ref{lem:hard-bicubic} to admit no $f(k)n^{o(k / \log k )}$ algorithm, unless ETH fails.
As our reduction will only make oracle calls to counting matchings of size $\mathcal O (k)$, we obtain the same lower bound for counting $k$-matchings in bipartite graphs.
\begin{theorem}
\label{thm:hard-matchings}
The following problems are $\sharpWone$-complete and admit no $f(k)n^{o(k/\log k)}$ algorithms, assuming ETH:
\begin{enumerate}
\item The problem $\pMatch$ of counting $k$-matchings in uncolored bipartite graphs.
\item The problem $\pColMatch$ of counting edge-colorful $k$-matchings in edge-colored bipartite graphs.
\end{enumerate}
\end{theorem}
We show the second claim, from which the first claim follows with
Lemma~\ref{lem:colorful2uncolored}(2).  The following technical lemma
will be needed in the proof, and illustrates how polynomials appear in
the context of counting matchings.
\begin{lemma}
\label{lem:matching-polynomial}Let $\Delta$ be a set of colors and let $A$ and $B$ be two edge-colorful graphs using only colors from $\Delta$.
For $n\ge 0$, let $A+n\cdot B$ denote the
graph consisting of $A$ together with $n$ vertex-disjoint copies of $B$.
Then for every $X\subseteq\Delta$, the value $\#\mathcal{M}_{X}(A+n\cdot B)$
is a polynomial in $n$ of maximum degree $|X|$.\end{lemma}
\begin{proof}
Given a partition $X=X_{A}\dot{\cup}X_{B}$, consider those $X$-colorful
matchings in $A+n\cdot B$ whose $X_{A}$-colored edges are contained
in $A$ and whose $X_{B}$-colored edges are contained in $n\cdot B$.
Their number is given by $\#\mathcal{M}_{X_{A}}(A)\cdot\#\mathcal{M}_{X_{B}}(n\cdot B)$ as $A$ and $n \cdot B$ are vertex-disjoint.
Therefore 
\[
\#\mathcal{M}_{X}(A+n\cdot B)=\sum_{X_{A}\dot{\cup}X_{B}=X}\#\mathcal{M}_{X_{A}}(A)\cdot\#\mathcal{M}_{X_{B}}(n\cdot B).
\]
As the values $\#\mathcal{M}_{X_{A}}(A)$ are constants independent of
$n$, it suffices to show that $\#\mathcal{M}_{X_{B}}(n\cdot B)$ is a
polynomial in $n$, for every fixed $X_{B}\subseteq X$. For a partition
$\rho$ of $X_{B}$, let $\mathcal{A}_\rho(n\cdot B)$ denote the set of
$X_{B}$-colorful matchings $M$ of $n\cdot B$ with the following
property: For all colors $i,j\in X_{B}$, the $i$-colored and the
$j$-colored edge of $M$ are contained in the same $B$-copy if and only
if $i$ and $j$ are both contained in the same class of $\rho$.  Let us
define $\alpha_{\rho}\in\{0,1\}$ to be $1$ if and only if no class of
$\rho$ contains two colors $i,j\in X_B$ that are incident in $B$;
clearly, $\alpha_{\rho}=0$ implies that $\mathcal{A}_\rho(n\cdot
B)=\emptyset$, as it makes it impossible to map $i$ and $j$ to the
same copy of $B$ (recall that $B$ is edge-colorful).
Therefore, we have (calling a partition with $\ell$ classes an $\ell$-partition)
\[
\#\mathcal{M}_{X_{B}}(n\cdot B)
=\sum_{\ell=1}^{|X_B|}\;\sum_{\substack{\ell\text{-partition }\\\rho\text{ of }X_{B}
}
} \#\mathcal{A}_\rho(n\cdot B)\;
= \sum_{\ell=1}^{|X_B|}
\sum_{\substack{\ell\text{-partition }\\\rho\text{ of }X_{B}
}
}\alpha_{\rho}\cdot(n)_{\ell}.
\]
Since the falling factorial expression $(n)_{\ell}$ for fixed $\ell\in\mathbb{N}$ is a polynomial
in $n$ of degree $\ell \leq |X|$, the claim is proven.
\end{proof}

The remainder of this section will comprise a proof of Theorem
\ref{thm:hard-matchings}, which we sketch in the following.  
As we reduce from $\pRestr(\mathcal{H}_{\textup{bicub}})$, let $H$
be a $3$-regular bipartite graph on vertices $[k]$ and let $G$ be
$[k]$-vertex-colored. In the setting of this reduction, we wish to determine $\#\SubPart(H\to
G)$, which is $\sharpWone$-complete by Lemma~\ref{lem:hard-bicubic}, and we are given oracle access for counting edge-colorful
matchings in bipartite graphs.

First, we transform the \emph{vertex}-colored graph
$G$ to an \emph{edge}-colored graph $G^{\triangle}$ on colors corresponding to
the \emph{edges} of $H$ and $6k$ additional colors. We denote the edge-colors of $H$ by $\Gamma = E(H)$.
The graph $G^{\triangle}$ is obtained by first coloring each edge between vertex-colors $i$ and $j$ in $G$ with the edge-color $ij \in \Gamma$. Secondly, each vertex $v\in V(G)$ is replaced by an edge-colorful gadget on six edges and with three special nodes. The edges incident with $v$ are then distributed to the three special nodes:
Recall that, by Remark~\ref{rem:partitioned-colorful}, the vertex $v$ sees exactly three vertex colors among its neighbors; draw an edge from the first special node to all neighbors of $v$ colored with the first such color, and so on.

Then we consider the $\Gamma$-edge-colorful matchings in $G^{\triangle}$, i.e., those matchings $M$ in $G^{\triangle}$ that contain exactly one edge of each color in $\Gamma$ and no other edges.
Any such $M$ can hit some number of gadgets between $k$ and $3k$.
We show that, if exactly $k$ gadgets are hit (one for each vertex-color of the original graph $G$), then $M$ is ``good'' as it corresponds to a subgraph $F \subseteq G$ that is color-preserving isomorphic to $H$.

It remains to isolate the good $\Gamma$-edge-colorful matchings of $G^{\triangle}$.
This will be achieved by setting up a linear system of equations featuring $2^{\mathcal O (k)}$ indeterminates and equations and full rank:
Each equation establishes a linear correspondence between a number we can determine by oracle calls, namely that of $\Gamma'$-edge-colorful matchings of $G^\Delta$ with $\Gamma' \supseteq \Gamma$, and a set of numbers we are looking for, namely that of $\Gamma$-edge-colorful matchings of $G^{\triangle}$ which are in certain \emph{states}. One of these states corresponds to the good matchings. The full proof follows.

\begin{proof}[Proof of Theorem \ref{thm:hard-matchings}]
  We prove the statement by a reduction from
  $\pRestr(\mathcal{H}_{\textup{bicub}})$. Let $H$ and $G$ be
  $[k]$-vertex-colored graphs such that $H$ is 3-regular, bipartite and colorful.
  Without limitation of generality, $G$ satisfies the condition stated
  in Remark~\ref{rem:partitioned-colorful}: There are no
  edges between color classes $i$ and $j$ of $G$ if there is no edge
  between the $i$-colored vertex and the $j$-colored vertex of
  $H$. 

  Moreover, let $n_0 \in \mathbb N$ with $n_0 \geq 3$ be a fixed universal constant (independent of $H$ and $G$) whose value will be determined at the end of the proof. 
  We assume that there is some $n \in \mathbb N$ such that
  $|V_{i}(G)|=n$ for all $i\in[k]$ and $n > n_0$. This can be ensured by adding isolated vertices to $G$.
  (Note that isolated vertices cannot appear in subgraphs $F$ isomorphic to $H$ as $H$ is $3$-regular.)  In the following,
  consider $H$ as a $\Gamma$-edge-colorful graph, where
  $\Gamma$ is a set of colors of size $3k/2$ corresponding to $E(H)$.

  For each vertex of $H$, let us fix an arbitrary ordering of the
  three edges incident to it.  Let $\Delta := [k]\times[6]$ and let
  $G^{\triangle}$ be the edge-colored graph with colors
  $\Gamma\cup\Delta$, which is obtained from $G$ as follows:
\begin{enumerate}
\item Replace each $v\in V(G)$ by a cycle $C_6$ on the vertices $w_{v,1}$, $z_{v,1}$, $w_{v,2}$, $z_{v,2}$, $w_{v,3}$, $z_{v,3}$. The edges of the cycle are colored with $\{i\}\times [6]$ the way it is shown in Figure~\ref{fig:triangle-states}.
\item Let us define the independent set
  $I(v)=\{w_{v,1},w_{v,2},w_{v,3}\}$.  For each vertex-color $i\in[k]$
  of $G$, define $\mathcal{I}(i)=\bigcup_{v\in V_{i}(G)}I(v)$.
\item For $e\in E(H)$ with $e=\{i,j\}$, let $a,b\in[3]$ be such that
$e$ is the $a$-th edge incident with $i$ and the $b$-th edge incident
with $j$. Replace each $\{u,v\}\in E(G)$ where $u$ is $i$-colored and $v$ is $j$-colored by the edge $\{w_{u,a},w_{v,b}\}$
of color $\gamma(e)\in\Gamma$.
\end{enumerate}

From a bipartition $V(H)=L\dot{\cup}R$, it is easy to construct a
bipartition $V(G^{\triangle})=L^{\triangle}\dot{\cup}R^{\triangle}$:
If $i\in L$ and $v\in V_{i}(G)$, put $I(v)$ into $L^{\triangle}$, and
put the remaining vertices of the $C_6$ cycle of $v$ into
$R^{\triangle}$. Proceed symmetrically for $i\in R$.

For $X\subseteq\Gamma\cup\Delta$, recall that $\mathcal{M}_{X}(G^{\triangle})$
denotes the set of matchings of $G^{\triangle}$ that contain exactly
one edge of each color in $X$. At first, we will only be
interested in $\mathcal{N}:=\mathcal{M}_{\Gamma}(G^{\triangle})$, i.e.,
in colorful matchings of the subgraph of $G^{\triangle}$ that contains
no $C_{6}$-edges. Observe that for $M\in\mbox{\ensuremath{\mathcal{N}}}$
and $i\in[k]$, the set $V(M)\cap\mathcal{I}(i)$ contains exactly
three vertices, which could be contained within a single
set $I(v)$ for some $v\in V(G)$, or they could be spread over different
such sets. That is, the three vertices can be all in the same $I(v)$, or be in three different sets $I(v_1)$, $I(v_2)$, $I(v_3)$, or one of them can be in some $I(v_1)$ and the other two in some $I(v_2)$. This last case further splits into three subcases: there is an $i\in [3]$ such that $w_{v_1,i}$ is used from $I(v_1)$ and the two vertices $w_{v_2,j}$ for $j\in [3]\setminus \{i\}$ are used from $I(v_2)$. In total, this yields five possibilities how the matching $M$ can look like from the viewpoint of the cycles representing $V_i(G)$ (see Figure~\ref{fig:triangle-states}).

\begin{figure}[t]
\begin{centering}
\includegraphics[width=0.96\textwidth]{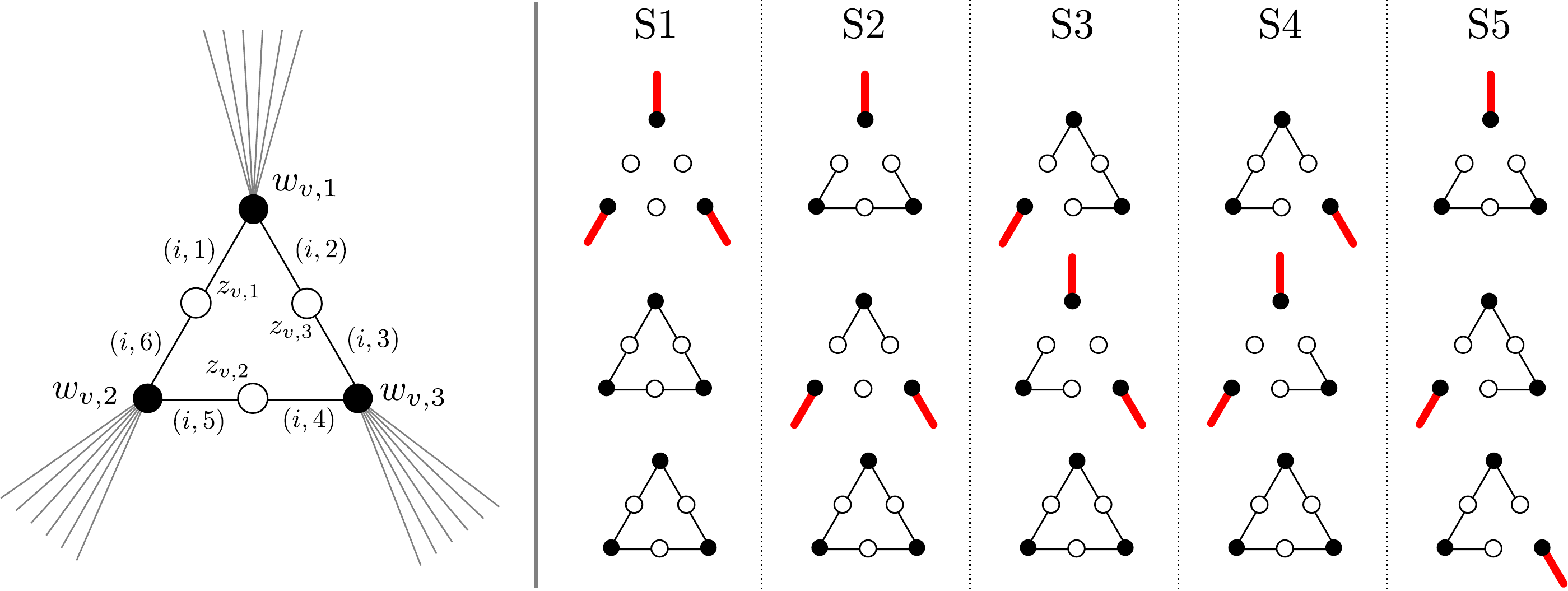}
\par\end{centering}

\caption{\label{fig:triangle-states}Each column represents one
  type. The partition of $M_{i}$ is depicted with red edges. The black
  edges show the edges of the cycles not incident to the matching; these
  edges form the graphs $R_{s}$.}
\end{figure}

We formally define the five possible
\emph{types} depicted in Figure~\ref{fig:triangle-states} as follows.
For $M\in\mbox{\ensuremath{\mathcal{N}}}$ and $i\in[k]$, call $u,v\in V(M)\cap\mathcal{I}(i)$
equivalent if there exists some $w\in V(G)$ such that $u,v\in I(w)$.
This equivalence notion induces a partition $\theta_{i}(M)$ of $V(M)\cap\mathcal{I}(i)$,
which we refer to by its index in Figure~\ref{fig:triangle-states}.
Let the vector $\theta(M)=(\theta_{1}(M),\ldots,\theta_{k}(M))$ be the \emph{type}
of $M$, and let $\Theta:=[5]^{k}$ be the set of all types. For $\theta\in\Theta$,
let $\mathcal{N}[\theta]:=\{M\in\mathcal{N}\mid\theta(M)=\theta\}$
denote the matchings of type $\theta$. Let $\theta^{*}=(1,\ldots,1)$
denote the good type. Given a type $\theta$, we use $\theta(i)$ to denote the $i$-th coordinate of $\theta$.

The set $\mathcal{N}[\theta^{*}]$ corresponds bijectively to
$\SubPart(H\to G)$: Every $M\in\mathcal{N}[\theta^{*}]$ describes a
copy of $H$, as the edges in $M$ involve exactly one vertex of color
$i$ for every $i\in [k]$; conversely every $H$-copy induces a
unique $M^{*}\in\mathcal{N}[\theta^{*}]$.  However, the matchings in
$\mathcal{N}[\theta]$ for $\theta\neq\theta^{*}$ stand in no useful
relation to $H$-copies.

In the following, we consider the edge-colorful matchings in
$\mathcal{M}_{X}(G^{\triangle})$, for certain sets $\Gamma\subseteq
X\subseteq\Gamma\cup \Delta$. Each matching in
$\mathcal{M}_{X}(G^{\triangle})$ is an extension of a matching $M\in
\mathcal{N}$. Different matchings $M\in \mathcal{N}$ have different
numbers of extensions in $\mathcal{M}_{X}(G^{\triangle})$, but we show
that the contribution of $M$ depends only on its type
$\theta(M)$. Therefore, the size of $\mathcal{M}_{X}(G^{\triangle})$
can be interpreted as a weighted sum over $M\in\mathcal{N}$ with
weights depending on $\theta(M)$. Our goal is to deduce the number of
matchings $M\in \mathcal N$ of type $\theta^*$ from the resulting system of linear equations.

This task requires a few definitions. For $t\in[5]$, define subsets
$A_{t}\subseteq[6]$ as follows.
\[
\begin{array}{ccccc}
A_{1}:=\{4,5\}\quad & A_{2}:=\{2,3\}\quad & A_{3}:=\{1,6\}\quad & A_{4}:=\{2,3,4,5\}\quad & A_{5}:=\{1,2,3,4,5,6\}.\end{array}
\]
These subsets seem somewhat arbitrary, and we remark that other (but not all) collections of subsets could be used in the proof.
For $i\in[k]$, write $A_{t}^{i}=\{i\}\times A_{t}$, which are colors appearing on the cycles representing vertices of $V_i(G)$. For $\mathbf{t}\in[5]^{k}$,
let 
\[
X(\mathbf{t}):=\Gamma\cup A_{\mathbf{t}(1)}^{1}\cup\ldots\cup A_{\mathbf{t}(k)}^{k}.
\]
For $s\in[5]$ and $i\in \Gamma$, let $C^i_6$ be the cycle representing
vertices of $V_i(G)$. We introduce a specific auxiliary graph
$R_{s}$, which is an induced subgraph of three disjoint copies of
$C^i_6$, after removing vertices incident to a matching of type
$s$; clearly, $R_s$ has exactly $3\cdot 6 -3=15$ vertices. These graphs are drawn in
Figure~\ref{fig:triangle-states}.  By
Lemma~\ref{lem:matching-polynomial}, for all $s,t\in [5]$, the quantity
\begin{equation}
\label{eq:pst}
p_{s,t}(n):=\#\mathcal{M}_{A^i_{t}}(R_{s}+n\cdot C^i_6)
\end{equation}
is a polynomial in $n$ of maximum degree $6$ which is independent of $H$ and $G$. In principle, the 25 polynomials $p_{s,t}$ could be calculated and written out explicitly. Computing these polynomials is tedious, but, as we shall see, we do not need to know these polynomials explicitly, it is sufficient to know that they are polynomials.

 With the next claim, we obtain $5^{k}$ linear
equations involving the following:
\begin{itemize}
\item the results $\#\mathcal{M}_{X(\mathbf{t})}(G^{\triangle})$ of oracle
calls on $\pColMatch$, for $\mathbf{t}\in[5]^{k}$, and 
\item products of numbers $p_{s,t}(n)$ for $s,t\in[5]$, where $p_{s,t}$ are defined above, and
\item the number of matchings $\#\mathcal{N}[\theta]$ for all $5^{k}$
types $\theta\in\Theta$, including the desired $\#\mathcal{N}[\theta^{*}]$.\end{itemize}
\begin{claim}
\label{lem:kron-comb}Let $n\geq3$, as assumed in this section. For
every $\mathbf{t}\in[5]^{k}$, it holds that 
\begin{equation}
\#\mathcal{M}_{X(\mathbf{t})}(G^{\triangle})=\sum_{\theta\in\Theta}\#\mathcal{N}[\theta]\cdot\prod_{i\in[k]}p_{\theta(i),\mathbf{t}(i)}(n-3).\label{eq:kron-equation}
\end{equation}
\end{claim}
\begin{proof}
  Let $\mathbf{t}\in[5]^{k}$. We can partition the
  $X(\mathbf{t})$-colored matchings $F$ in $G^{\triangle}$ according
  to the type of the edges of color $\Gamma$.  We claim that each
  partition class then contains
  $\#\mathcal{N}[\theta]\cdot\prod_{i\in[k]}p_{\theta(i),\mathbf{t}(i)}(n-3)$
  elements. First, the number of matchings of type $\theta$ in
  $\mathcal{N}$ is given by $\#\mathcal{N}[\theta]$. For $i\in [k]$,
  each matching $M\in \mathcal{N}[\theta]$ occupies three vertices of
  the $n$ copies of $C^i_6$ representing $V_i(G)$. What remains is
  isomorphic to $R_{\theta(i)}+(n-3)\cdot C^6_i$, since three copies
  are affected as described by $\theta(i)$ and the remaining $n-3$
  copies are unaffected.  Therefore, independently for $i\in[k]$, each
  matching in $\mathcal{N}[\theta]$ can be extended by edges of colors
  $A_{\mathbf{t}(i)}^{i}$ in $p_{\theta(i),\mathbf{t}(i)}(n-3)$ ways.
  \cqed\end{proof} For $\mathbf{t}\in[5]^{k}$, consider
(\ref{eq:kron-equation}) as a linear equation in the unknowns
$\#\mathcal{N}[\theta]$: We obtain $T$ equations in $T$ unknowns,
where $T=5^{k}$. With
$[5]^{k}=\{\mathbf{t}_{1},\ldots,\mathbf{t}_{T}\}$ and
$\Theta=\{\theta_{1},\ldots,\theta_{T}\}$, they read
\begin{equation}
\underbrace{\left(\begin{array}{ccc}
\prod_{i\in[k]}p_{\theta_{1}(i),\mathbf{t}_{1}(i)}(n-3) & \ldots & \prod_{i\in[k]}p_{\theta_{T}(i),\mathbf{t}_{1}(i)}(n-3)\\
\vdots & \ddots & \vdots\\
\prod_{i\in[k]}p_{\theta_{1}(i),\mathbf{t}_{T}(i)}(n-3) & \ldots & \prod_{i\in[k]}p_{\theta_{T}(i),\mathbf{t}_{T}(i)}(n-3)
\end{array}\right)}_{=:R_{k}(n-3)}\left(\begin{array}{c}
\#\mathcal{N}[\theta_{1}]\\
\vdots\\
\#\mathcal{\mathcal{N}}[\theta_{T}]
\end{array}\right)=\left(\begin{array}{c}
\#\mathcal{M}_{X(\mathbf{t}_{1})}(G^{\triangle})\\
\vdots\\
\#\mathcal{M}_{X(\mathbf{t}_{T})}(G^{\triangle})
\end{array}\right).\label{eq:kron-system}
\end{equation}

By Gaussian elimination, a solution to (\ref{eq:kron-system}) can
be found in time $\mathcal{O}(T^{3})$, but it is crucial to show
that this solution is unique, i.e., that $R_{k}(n)$ has full rank.
We show that there is a number $n_0 \in \mathbb N$ depending only on the fixed polynomials $p_{s,t}$, which are independent of $H$ and $G$,
such that for all $n,k\in\mathbb{N}$ with $n > n_0$,
the matrix $R_{k}(n)$ has full rank.

For $R\in\mathbb{Z}^{\ell\times\ell}$ and $k\in\mathbb{N}$, we write
$R^{\otimes k}$ for the $k$-th Kronecker power of $R$: The $\ell^{k}$
rows and columns of $R^{\otimes k}$ are indexed by the lexicographical
ordering of vectors $\mathbf{i},\mathbf{j}\in[\ell]^{k}$, and it holds
that $(R^{\otimes
  k})_{\mathbf{i},\mathbf{j}}=\prod_{s\in[k]}R_{\mathbf{i}(s),\mathbf{j}(s)}$.
Let us observe that $R_{k}(n)=(R_{1}(n))^{\otimes k}$, where
$R_{1}(n)$ is the $5\times5$ matrix with $(R_{1}(n))_{s,t}=p_{s,t}(n)$
for $s,t\in[5]$.  It is a basic property of the Kronecker product that
the $k$-th Kronecker power of a nonsingular square matrix is also
nonsingular. Therefore, we only need to verify that $R_1(n)$ is
nonsingular. By Lemma~\ref{lem:matching-polynomial}, the value $p_{s,t}(n)$ is a
  polynomial in $n$ for every $s,t\in[5]$, hence the determinant
  $\det(R_{1}(n))$ is also a polynomial in $n$. This means that it is
  either zero for every $n\in \mathbb{Z}$, or zero only for finitely
  many $n$. Recall that $(R_1(0))_{s,t}=p_{s,t}(0)=\#\mathcal{M}_{A^i_{t}}(R_s)$, that is, the number of $A^i_t$-colored matchings in a specific 15-vertex graph $R_s$, which can be computed with some effort. 
We show in the appendix that in fact
\[
R_{1}(0)=\left(\begin{array}{ccccc}
2 & 2 & 3 & 3 & 3\\
2 & 3 & 2 & 3 & 3\\
2 & 3 & 3 & 2 & 3\\
2 & 3 & 3 & 4 & 5\\
2 & 2 & 2 & 2 & 4
\end{array}\right),
\]
  and it can be calculated that $\det(R_{1}) = 12$.
  This implies $\det(R_{1}(n))$, interpreted as a polynomial in $n$, is not identically 0, which in turn implies that $R_1(n)$ is
  singular only for finitely many $n$.  We conclude that
  (\ref{eq:kron-system}) admits a unique solution if
  $n > n_0$, which we assumed in the beginning by adding isolated vertices to $G$.
\end{proof}

Let us remark on how the current form of the proof of Theorem~\ref{thm:hard-matchings} was found. In the proof, we consider five specific sets $A_1$, $\dots$, $A_5$, but we could have considered all possible $2^6=64$ subsets. This would have resulted in a larger system of equations with $5^k$ unknowns and $64^k$ equations. Again, the coefficient matrix is the $k$-th Kronecker power of a fixed $64\times 5$ matrix where each entry is a polynomial of $n$ having degree at most 6. The authors wrote a computer program to calculate all these $64\times 5$ polynomials and observed that the columns are linearly independent. If the 5 columns are linearly independent, there has to exist a set of 5 rows such that the corresponding $5\times 5$ submatrix is nonsingular; this is how the sets $A_1$, $\dots$, $A_5$ were selected. Finally, we have observed that it is not necessary to compute explicitly the 25 polynomials appearing in this matrix to make sure that it is nonsingular: it is sufficient to show that there is at least one $n$ where the matrix is nonsingular. Computing the matrix for $n=0$ requires solving 25 instances of counting colored matchings in a 15-vertex graph, which can be done conveniently with computer and is also doable by hand. We stress that while the proof was found with computer help, the only calculation that is needed to verify the proof is computing the $5\times 5$ matrix $R_1(0)$.

With a simple reduction, we can transfer the lower bound of Theorem~\ref{thm:hard-matchings} for counting bipartite $k$-matchings to counting directed $k$-cycles.
\begin{corollary}
The problem $\pDirCycle$ of counting directed cycles of length $k$ is $\sharpWone$-hard 
and admits no algorithm with runtime $f(k)n^{o(k / \log k )}$, unless ETH fails.
\end{corollary}
\begin{proof}
We show how to count $k$-matchings in a bipartite graph $G$ with bipartition
$V(G)=L\dot{\cup}R$ by counting directed $2k$-cycles in the graph
$G'$ obtained from $G$ by directing each edge from $L$ to $R$
and adding each edge from $R$ to $L$. This also implies the lower bound under ETH.

As every $2k$-cycle $C$ of $G'$ must alternate between $L$ and
$R$, it contains $k$ edges from $L$ to $R$. As $C$ is a simple cycle
(it visits no vertex twice), the edges from $L$ to $R$ induce
a $k$-matching of $G$. Conversely, every $k$-matching $M$ of $G$
can be extended to a $2k$-cycle in $G'$ by (i) orienting the edges
of $M$ from $L$ to $R$ and (ii) adding edges from $R$ to $L$.
Step (ii) fixes a permutation of $M$ up to cyclic equivalence, hence
each $M$ corresponds to $(k-1)!$ cycles of length $2k$ in $G'$.

Thus, if $t$ denotes the number of $k$-matchings in $G$ and $s$ denotes the number of $2k$-cycles in $G'$, then $s=(k-1)! \cdot t$, which proves the claim.
\end{proof}

Next we show how to prove hardness for counting undirected cycles.

\begin{theorem}
The problems $\pDirCycle, \pUndirCycle, \pDirPath, \pUndirPath$ of counting directed/undirected paths/cycles of length $k$ are all $\sharpWone$-hard and admit no algorithm with runtime $f(k)n^{o(k / \log k )}$, unless ETH fails.
\end{theorem}

\begin{proof}
We reduce $\pDirCycle \leqLin \pUndirCycle \leqLin \pUndirPath \leqLin \pDirPath$.
The last two reductions were shown by Flum and Grohe~\cite[Lemma 22]{MR2065338} and can be observed to preserve the parameter.\footnote{In fact, a reduction from $\pDirCycle$ to $\pUndirCycle$ was also shown by Flum and Grohe~\cite{MR2065338}, but it incurs a polynomial increase of the parameter.}
We thus only show the first reduction.
Let $D$ be a directed graph in which we want to count directed $k$-cycles.
We transform $D$ to an undirected graph $G$ as follows:

\begin{enumerate}
\item Replace each $v\in V(D)$ by vertices $v^{\textup{in}}$ and $v^{\textup{out}}$,
and replace each directed edge $uv\in E(D)$ by the undirected edge
$u^{\textup{out}}v^{\textup{in}}$ in $G$. Consider these edges to
be colored with $0$.
\item For each $v\in V(D)$, add $k$ parallel edges, called \emph{internal edges at $v$}, between vertices $v^{\textup{in}}$
and $v^{\textup{out}}$. Assign colors $1,\ldots,k$ to these edges. (We will later show how to obtain simple graphs that feature no parallel edges.)
\end{enumerate}
Let $\mathcal{C}$ denote the set of $k$-cycles in $D$, and let
$\mathcal{B}$ denote the set of $2k$-cycles in $G$ that choose
at least one edge of each color $[k]$.

\begin{claim}

It holds that $|\mathcal{B}|=k!\cdot|\mathcal{C}|$.

\end{claim}

\begin{proof}

For each cycle $C\in\mathcal{C}$, we can define a subset $\mathcal{B}_{C}\subseteq\mathcal{B}$,
which includes $B\in\mathcal{B}_{C}$ if and only if $B$ contains all edges
$u^{\textup{out}}v^{\textup{in}}$ for $uv\in E(C)$ and some internal
edge $w^{\textup{out}}w^{\textup{in}}$ for each $w\in V(C)$. It
is obvious that $\mathcal{B}_{C}\cap\mathcal{B}_{C'}=\emptyset$ if
$C\neq C'$. Independently of $C$, the set $\mathcal{B}_{C}$ has
size $k!$ as there are $k!$ ways of choosing internal edges for
each $v\in V(C)$: Note that exactly $k$ internal edges
are present in $B$, and each color in $[k]$ must be the color of
some internal edge. Therefore, we have $|\mathcal{B}|\ge k!\cdot |\mathcal{C}|$.

We now show $\mathcal{B}\subseteq\bigcup_{C\in\mathcal{C}}\mathcal{B}_{C}$,
which implies $|\mathcal{B}|\le k!\cdot|\mathcal{C}|$ and proves the claim. Since each color in $[k]$ is present in
$B$, the cycle $B$ passes through $s\geq k$ internal edges. But
since internal edges are vertex-disjoint, at least $s$ additional $0$-colored
edges are required for $B$ to be a cycle. Since $B$ has length $2k$,
this is possible only when $s=k$. Let $v_{1},\ldots, v_{k}\in V(D)$
be the vertices whose internal edges $B$ visits, then $B$ contains
both $v_{i}^{\textup{in}}$ and $v_{i}^{\textup{out}}$ for all $i\in[k]$.
Thus, if we orient edges from out-vertices to in-vertices and contract
each internal edge, then we obtain a directed $k$-cycle of $D$, i.e., some
element of $\mathcal{C}$, and it holds that $B\in\mathcal{B}_{C}$.
\cqed\end{proof}

This allows to compute $|\mathcal{C}|$ from $|\mathcal{B}|$. We determine $|\mathcal{B}|$
from the number of $2k$-cycles in certain uncolored subgraphs of $F'$, using
an inclusion-exclusion argument similar to that of Lemma \ref{lem:colorful2uncolored}: For each
$S\subseteq\Gamma$, let $F'_{S}$ denote the subgraph of $F'$ that
contains only edges of color $S$. Then a $2k$-cycle $C$ in $F'$
is contained in $\mathcal{B}$ if and only if $C$ is contained in
$F'_{\Gamma}$, but in none of $F'_{S}$ for $S\subsetneq\Gamma$.
Together with oracle calls for counting $2k$-cycles, this allows
to apply the inclusion-exclusion formula \eqref{eq:incl-excl}. Note that, in each oracle
call, the parameter is $2k$, which implies the lower bound under ETH.

In the above reduction, the oracle for $2k$-cycles is called on graphs which may feature parallel edges,
but we can easily reduce these to graphs without parallel edges.
Let $G'$ denote the graph obtained from $G$ by subdividing each edge.
Then it is clear that $G'$ contains no parallel edges, while for $t >2$, the $2t$-cycles of $G'$ stand in bijection with the $t$-cycles of $G$.
\end{proof}

\section{Hereditary classes}
\label{sec:hereditary-classes}

In this section, we prove Theorem~\ref{th:main} in the special case
when $\H$ is a hereditary class, that is, when $H\in \H$ implies that every
induced subgraph is in $\H$. We use the multicolored version of Ramsey's Theorem,
stating that if we color the edges of a sufficiently large clique with
a constant number of colors, then a large monochromatic clique appears,
that is, there is a large set of vertices such that every edge
between them has the same color.

\begin{theorem}\label{theorem:multiramsey}
  Let $c$, $r$, and $n$ be positive integers with $n\ge
  (r+1)^{rc}$. Given a $c$-coloring of the edges of an $n$-clique,
  there is a subset of $r$ vertices inducing a monochromatic
  $r$-clique in the coloring.
\end{theorem}

Using Theorem~\ref{theorem:multiramsey}, we can give a surprisingly simple
induced-subgraph characterization of graphs with large matchings.
\begin{lemma}\label{lem:hereditaryramsey}
  For every $k\ge 1$, there is an $m\ge 1$ such that following
  holds. If a graph $G$ contains a (not necessarily induced) matching of size $m$, then $G$ contains
  as an induced subgraph either
\begin{itemize}
\item a clique on $k$ vertices,
\item a biclique (i.e., a complete bipartite graph) on $k+k$ vertices, or
\item a matching with $k$ edges.
\end{itemize}
\end{lemma}
\begin{proof}
Let $r=2k$, $c=16$, and let $m=(r+1)^{rc}$ be the bound appearing in Theorem~\ref{theorem:multiramsey}. 
Let $M$ be a matching of size $m$ in $G$, and for $1 \leq i \leq m$, let $x_i,y_i\in V(G)$ denote the endpoints of the $i$-th edge of $M$. 
We define four graphs $H_{xx}$, $H_{xy}$, $H_{yx}$, $H_{yy}$ on the same vertex set $[m]$. For every $1\le i < j \le m$, 
\begin{itemize}
\item $\{i,j\}\in E(H_{xx})$ if and only if $x_i$ and $x_j$ are adjacent in $G$,
\item $\{i,j\}\in E(H_{xy})$ if and only if $x_i$ and $y_j$ are adjacent in $G$,
\item $\{i,j\}\in E(H_{yx})$ if and only if $y_i$ and $x_j$ are adjacent in $G$,
\item $\{i,j\}\in E(H_{yy})$ if and only if $y_i$ and $y_j$ are adjacent in $G$.
\end{itemize}
For any pair $i,j\in[m]$ with $i \neq j$, there are $2^4=16$ possibilities for the edge $\{i,j\}$
to be present in a subset of the four graphs defined above.
This induces a coloring of the edges of a clique on $[m]$ with
$2^4$ colors. Then the choice of $m$ and
Theorem~\ref{theorem:multiramsey} imply that there is a subset
$S=\{s_{1},\dots,s_{2k}\}\subseteq [m]$ such that all edges between
these vertices have the same color. That is, in each of the four
graphs, $S$ is either a clique or an independent set. If $S$ is a
clique in $H_{xx}$, then $\{x_s \mid s\in S\}$ is a $2k$-clique in
$G$, and we are done. Similarly, if $S$ is a clique in $H_{yy}$, then
$\{y_s\mid s\in S\}$ is a $2k$-clique in $G$. Therefore, we can assume
that $S$ is an independent set in both $H_{xx}$ and $H_{yy}$. Suppose
that $S$ is a clique in $H_{xy}$. Let $S_1$ and $S_2$ be the smallest and
largest $k$ elements of $S$, respectively. Then the fact that $S$ is a
clique in $H_{xy}$ implies that $x_i$ and $y_j$ are adjacent for every
$i\in S_1$ and $j\in S_2$ (note that $i<j$ holds). As $S_1$ is
independent in $H_{xx}$, the set $\{x_s\mid s\in S_1\}$ is independent
in $G$. Similarly, $\{y_s\mid s\in S_2\}$ is independent in $G$. Thus
$S_1\cup S_2$ induces a biclique with $k+k$ vertices in $G$. The
argument is similar if $S$ is a clique in $H_{yx}$. Therefore, we can
assume that $S$ is an independent set in all four graphs. This means
that $\{x_sy_s\mid s\in S\}$ is an induced matching of size $2k$.
\end{proof}
It follows that if $\H$ is a hereditary class of graphs such that
there is no bound on the size of the largest matching (equivalently,
there is no bound on the vertex-cover number), then either arbitrarily
large cliques, arbitrarily large induced bicliques, or arbitrarily large
induced matchings appear in $\H$.
\begin{corollary}\label{cor:hereditaryramsey}
Let $\H$ be a hereditary class of graphs. Then at least one of the following holds:
\begin{itemize}
\item $\H$ has bounded vertex-cover number.
\item For all cliques $H$, it holds that $H \in \H$.
\item For all bicliques $H$, it holds that $H \in \H$.
\item For all matchings $H$, it holds that $H \in \H$.
\end{itemize}
\end{corollary}

If $\H$ contains every clique, then the canonical $\sharpWone$-hard problem
$\pClique$ can be trivially reduced to $\pSub(\H)$. If $\H$ contains every matching, then the $\sharpWone$-hard problem $\pMatch$
(see \cite{DBLP:conf/icalp/Curticapean13} and
Theorem~\ref{th:main-matching}) can be reduced to $\pSub(\H)$.  If
$\H$ contains every biclique, then $\pBiclique$ (that is, the problem of counting the number of
occurrences of a given biclique) can be reduced to $\pSub(\H)$. The
$\sharpWone$-hardness of $\pBiclique$ has been established in the
Diploma Thesis of Thurley \cite{thurley-diploma}. We can also argue
the following way. Dell and Marx \cite{DBLP:conf/soda/DellM12}
presented a parameterized reduction from finding cliques to finding colored
bicliques. One can observe that the reduction is parsimonious, implying that $\pRestr(\H)$ is $\sharpWone$-hard for the class $\H$ of bicliques.
Then Lemma~\ref{lem:colorful2uncolored}(1) implies the $\sharpWone$-hardness of $\pSub(\H)$ for the class $\H$ of bicliques.
\begin{theorem}\label{thm:countingbiclque}
$\pBiclique$ is $\sharpWone$-hard.
\end{theorem}
Thus, if a hereditary class $\H$ has unbounded vertex-cover number, then $\pSub(\H)$ is $\sharpWone$-hard in each of the three possible cases of Corollary~\ref{cor:hereditaryramsey}. Together with Theorem~\ref{thm:algo-vc}, this gives a simple proof of Theorem~\ref{th:main} in the special case of hereditary classes.
\begin{theorem}\label{th:main-hereditary}
  Let $\H$ be a hereditary class of graphs. Assuming $\FPT\neq \sharpWone$,
  the following are equivalent:
\begin{enumerate}
\item $\pSub(\H)$ is polynomial-time solvable.
\item $\pSub(\H)$ is fixed-parameter tractable parameterized by $|V(H)|$.
\item $\H$ has bounded vertex-cover number.
\end{enumerate}
\end{theorem}
\section{Reducing matchings to $\pSub(\H)$ via gadgets}
\label{sec:reduc-match-psubh}
The arguments of the previous section allowed for a simple treatment of the problem $\pSub(\H)$ if $\H$ is hereditary.
However, as already stated in the introduction, this approach fails even on very simple graph classes $\H$, such as the class of disjoint copies of triangles, that we intuitively expect to be ``close'' to the class of matchings.

To show hardness for non-hereditary classes $\H$, we develop a general machinery of $k$-matching gadgets, which are graphs $H\in\H$ together with a partition of $V(H)$ into an induced matching $M$ and some remainder $C$. These gadgets satisfy certain technical properties which will be used in Theorem \ref{th:reduce}, which is the main reduction of this paper.
It states that, if $\H$ is a class of graphs that contains $k$-matching gadgets for all $k\in\mathbb N$,
then there is a parameterized Turing reduction from the problem of counting (uncolored) $k$-matchings in bipartite graphs $G$ to the problem $\pSub(\H)$.

In the remainder of this section, we define $k$-matching gadgets formally, give some first examples of their properties, and then prove Theorem \ref{th:reduce}.
Proving the actual existence of $k$-matching gadgets in graph classes $\H$ will be the task of the subsequent sections.

\begin{definition}
  Let $H$ be a graph and let $C\subseteq V(H)$ be a subset of vertices. Then we
  denote by $\boundary_H(C)$ the set of vertices in $C$ that have a
  neighbor in $V(H)\setminus C$. If $f$ is an isomorphism from $H[C]$
  to $H[C']$ for some $C,C'\subseteq V(H)$ such that
  $f(\boundary_H(C))=\boundary_H(C')$, then we say that $f$ is {\em
    boundary preserving.}
\end{definition}
Observe that $X\subseteq Y$ implies
$(X\setminus \partial_H(X))\subseteq (Y\setminus \partial_H(Y))$: if
$v\in X$ has no neighbor outside $X$, then it has no neighbor outside
$Y$ either.

The following definition formulates the properties of the gadgets we need in the main reduction (Theorem~\ref{th:reduce}).
\begin{definition}\label{def:gadget}
  Let $H$ be a graph, $M$ be an induced $k$-matching in $H$, and let
  $C:=V(H)\setminus V(M)$. We say that $(H,M)$ is a {\em $k$-matching
    gadget} if whenever an isomorphism $f$ from
  $H[C]$ to $H[C']$ for some $C'\subseteq V(H)$ satisfies the conditions
\begin{enumerate}[label=(C\arabic*)]
\item\label{i:cisolated} $H\setminus C'$ has no isolated vertex,
\item\label{i:cbipartite} $H\setminus C'$ is bipartite, and
\item\label{i:boundary} $f$ is boundary preserving,
\end{enumerate}
then it is also true that $H\setminus C'$ is a $k$-matching, i.e., $H \setminus C'$ is isomorphic to the graph on $2k$ vertices that contains $k$ vertex-disjoint edges.
\end{definition}
Using a rather extensive graph-theoretical analysis, we will show in Sections \ref{sec:bound-degr-graphs}-\ref{sec:bound-treew-graphs}:

\begin{theorem}
\label{thm:gadgetexists}
Let $\H$ be a graph class of unbounded vertex-cover number and bounded treewidth.
Then, for all $k \in \mathbb N$, there exists a graph $H\in \H$ and a subset $M \subseteq V(H)$ such that $(H,M)$ is a $k$-matching gadget.
\end{theorem}

It indeed suffices to consider classes $\H$ covered by this theorem: 
Recall that, by Theorem~\ref{thm:algo-vc}, the problem $\pSub(\H)$ admits a polynomial-time algorithm if $\H$ has bounded vertex-cover number.
If $\H$ has unbounded treewidth, then $\pSub(\H)$ is $\sharpWone$-complete by Theorem~\ref{thm:hard-treewidth}. 

It will be convenient to know that if a $k$-matching gadget exists,
then a $k_0$-matching gadget also exists for every $k_0<k$. This is
not obvious from the definition and requires a nontrivial proof (which
also serves as an illustration of Definition~\ref{def:gadget} and how
it is used, e.g., in the proof of Claim~\ref{cl:minmatching} in the next section).
\begin{lemma}\label{lem:subgadget}
If $(H,M)$ is a $k$-matching gadget and $M_0\subseteq M$ is a $k_0$-matching, then $(H,M_0)$ is a $k_0$-matching gadget.
\end{lemma}
\begin{proof}
  Let $C=V(H)\setminus V(M)$ and $C_0=V(H)\setminus V(M_0)$; we have
  $C\subseteq C_0$. Let $f_0$ be an isomorphism from $H[C_0]$ to
  $H[C'_0]$ satisfying \ref{i:cisolated}--\ref{i:boundary} of
  Definition~\ref{def:gadget} (see Figure~\ref{fig:smallergadget}). We
  have to show that $H\setminus C'_0$ is a matching (and then clearly
  it is a $k_0$-matching, since $H\setminus C'_0$ has the same size as
  $H\setminus C_0$).  Let $f$ be the restriction of $f_0$ to $C$ and
  let $C'=f(C)$; then $f$ is an isomorphism from $H[C]$ to $H[C']$.
\begin{figure}
\begin{center}
{\footnotesize \svg{0.65\linewidth}{smallergadget}}
\caption{Proof of Lemma~\ref{lem:subgadget}.}\label{fig:smallergadget}
\end{center}
\end{figure}

We claim that $f$ satisfies \ref{i:cisolated}--\ref{i:boundary} of Definition~\ref{def:gadget} with respect to $(H,M)$.
\begin{enumerate}[label=(C\arabic*)]
\item Suppose that there is an isolated vertex $v\in H\setminus C'$.
  As $H[C'_0\setminus C']$ is isomorphic to $H[C_0\setminus C]$, which
  induces a matching, we have $v\not\in C'_0$.  By \ref{i:cisolated}
  on $f_0$, we have that there is no isolated vertex in $H\setminus
  C'_0$, a contradiction.
\item The graph $H\setminus C'_0$ is bipartite by \ref{i:cbipartite}
  on $f_0$.  The graph $H[C'_0\setminus C']$ is isomorphic to
  $H[C_0\setminus C]$, which induces a matching, hence bipartite.
  There are no edges between $C_0\setminus C$ and $V(H)\setminus C_0$
  (as there are no edges between $M_0$ and $M\setminus M_0$). Thus the
  fact that $f_0$ is boundary preserving implies that there is no edge
  between $C'_0\setminus C'$ and $V(H)\setminus C'_0$. It follows that
  $H\setminus C'$ is bipartite.

\item Consider first a vertex $v\in C\setminus \boundary_H(C)$.
  As every neighbor of $v$ is in $C$, every neighbor of $f(v)$ is in
  $C'$, that is, $f(v)\in C'\setminus \boundary_H(C')$. Consider now a
  vertex in $v\in \boundary_H(C)$. If $v$ has a neighbor $u\in
  C_0\setminus C$, then $f(u)\in C'_0\setminus C'$ is a neighbor of
  $f(v)$ outside $C'$, implying $v\in \boundary_H(C')$. Finally, if
  $v$ has a neighbor $u\not\in C_0$, then $v\in \boundary_H(C_0)$ and
  hence (as $f_0$ is boundary-preserving by property \ref{i:boundary})
  we have $f(v)\in \boundary_H(C'_0)$. Since $f(v)$ is in $C'$, we
  also have $f(u)\in \boundary_H(C')$.
\end{enumerate}
As $(H,M)$ is a $k$-matching gadget and $f$ satisfies
\ref{i:cisolated}--\ref{i:boundary} of Definition~\ref{def:gadget}, it
follows that $H\setminus C'$ is a matching. As $H[C'_0\setminus C']$
is isomorphic to the matching $H[C_0\setminus C]$, this is only
possible if $H\setminus C'_0$ is also a matching.
\end{proof}

The following lemma shows as an example a simple condition that guarantees the correctness of a $k$-matching gadget.
\begin{lemma}\label{lem:nocommon}
  Let $M$ be an induced $k$-matching in a graph $H$ such that every
  vertex of $C:=V(H)\setminus V(M)$ is adjacent to at most one vertex
  of $V(M)$. Then $(M,H)$ is a $k$-matching gadget.
\end{lemma}
\begin{proof}
  Suppose that $f$ is an isomorphism from $H[C]$ to $H[C']$ for some
  $C'\subseteq V(H)$ satisfying \ref{i:cisolated}--\ref{i:boundary} of
  Definition~\ref{def:gadget}, but $H\setminus C'$ is not a matching.
  As $H[C]$ and $H[C']$ are isomorphic, the number of edges in $H[C]$
  and $H[C']$ are the same. Every vertex $v\in C$ has at most one edge
  to $V(H)\setminus C$ and if $v$ has such an edge (that is, $v\in \boundary_H(C)$), then \ref{i:boundary}
  implies that $f(v)\in \boundary_H(C')$ has at least one edge to $V(H)\setminus
  C'$. Therefore, the number of edges between $C'$ and $V(H)\setminus
  C'$ is at least the number of edges between $C$ and $V(H)\setminus
  C$. It follows that the number of edges in $H\setminus C'$ is at most
  the number of edges in $H\setminus C$, that is, at most
  $k$. Therefore, the $2k$-vertex graph $H\setminus C'$ has at most
  $k$ edges and \ref{i:cisolated} implies that it does not have isolated
  vertices.\footnote{This the point where we crucially use \ref{i:cisolated} of Definition~\ref{def:gadget}.} This is only possible if $H\setminus C'$ is a
  $k$-matching.
\end{proof}
The condition in Lemma~\ref{lem:nocommon} can be also formulated as
requiring the following two properties:
\begin{itemize}
\item[(P1)] no two edges of $M$ have a common neighbor in $C$ and
\item[(P2)] the two endpoints of an edge of $M$ have no common neighbor in
$C$.
\end{itemize}
As we shall see, condition (P1) is usually easy to achieve in the
graphs we care about (by making $M$ somewhat smaller), but
we will have more difficulty in ensuring condition (P2).

The main reduction is described by the following lemma and theorem, which 
provide a reduction from counting bipartite $k$-matchings to $\Sub(\H)$ whenever $\H$ contains $k$-matching gadgets of all sizes.
\begin{lemma}\label{lem:reduce}
Let $G$ be a graph and let $(H,M)$ be a $k$-matching gadget of size $t = |V(H)|$.
Then we can compute the number of $k$-matchings in $G$ from $2k \cdot 2^{\mathcal O (t^2)}$ oracle queries of the form $\#\Sub(H\to G')$, where $G'$ is an arbitrary graph.
\end{lemma}
\begin{proof}

For $\ell\in\mathbb{N}$, let $G^{(\ell)}$ be defined as the following
supergraph of $G$. To avoid notational overhead, we write $F[X]$
for $F[X\cap V(F)]$.
\begin{figure}[t]
\centering{}\includegraphics[width=0.8\textwidth]{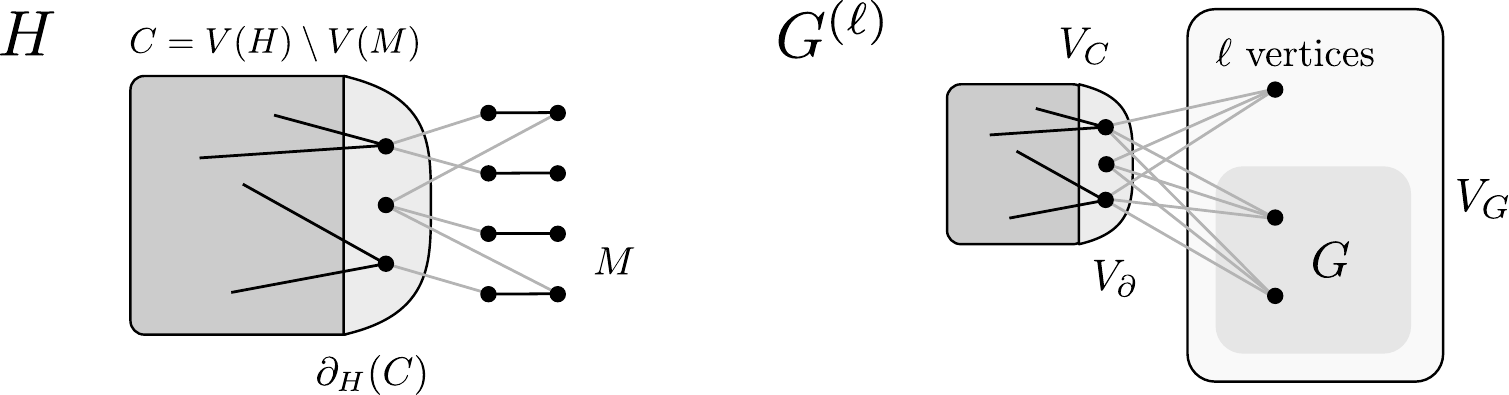}\caption{A $k$-matching gadget $(H,C)$, and the graph $G^{(\ell)}$ constructed
from $(H,C)$ and some graph $G$.}
\end{figure}

\begin{enumerate}
\item Starting from the empty graph, place a copy of $G$ into $G^{(\ell)}$
and add $\ell$ isolated vertices. 
Let $V_{G}$ denote the set that contains the isolated vertices and the vertices of the $G$-copy.
\item Add a copy of $H[C]$ on a vertex set $V_{C}$ disjoint from $V_{G}$.
Denote the copy of $\partial_{H}(C)$ by $V_{\partial}\subseteq V_{C}$.
\item Add all edges between $V_{\partial}$ and $V_{G}$, call this edge
set $I$.
\end{enumerate}
Using the principle of inclusion and exclusion, we will count the
copies of $H$ in $G^{(\ell)}$ that fully use $G^{(\ell)}[V_C]$ and
where every vertex of $V_{\partial}$ has an edge to $V_G$. If such a copy
$F$ of $H$ in $G^{(\ell)}$ additionally satisfies that $F[V_G]$ contains no isolated vertices, then
it can be observed that conditions \ref{i:cisolated}--\ref{i:boundary}
imply that $F[V_G]$ is actually a $k$-matching. Therefore, the number
of such copies can be put into relationship with the number of
$k$-matchings in $G$. To calculate the number of copies with this
additional restriction, we count the number of copies for different
values of $\ell$ (which correspond to different numbers of introduced isolated vertices)
and use an interpolation argument to isolate the contribution of those
copies where $F[V_G]$ has no isolated vertices.

We first aim at determining the size of
\[
T(\ell)=\{F\in\Sub(H\to G^{(\ell)})\mid F[V_{C}]\simeq H[C]\:\wedge\:\forall v\in V_{\partial}:\deg_{I\cap F}(v)>0\},
\]
that is, where the copy $F$ fully uses $G^{(\ell)}[V_C]$ and every
vertex of $V_{\partial}$ has a neighbor outside $V_C$ in the copy.
This will be achieved by inclusion-exclusion in the following claim.
Building upon this, we will later use interpolation to count those
$F\in T(\ell)$ where $F[V_{G}]$ additionally contains no isolated
vertices. (Note that $\ell$ is irrelevant if this additional property
holds.)
\begin{claim}
\label{cl:enforce-C}
We can compute $\#\mathsf{T}(\ell)$ with at most $2^{\mathcal O (t^2)}$ oracle calls to $\Sub(H\to G')$, where $G'$ is a subgraph of $G^{(\ell)}$.
\end{claim}
\begin{proof}
Observe that for $F\in\Sub(H\to G^{(\ell)})$, the condition $F[V_{C}]\simeq H[C]$
is equivalent to requiring that $G^{(\ell)}[V_{C}]\subseteq F$. For
subgraphs $X\subseteq G^{(\ell)}[V_{C}]$ and subsets $Y\subseteq V_{\partial}$,
let 
\[
\mathcal{A}_{X,Y}=\{F\in\Sub(H\to G^{(\ell)})\mid F[V_{C}]\subseteq X\ \wedge\ \forall v\in V_{\partial}:\deg_{I\cap F}(v)>0\Rightarrow v\in Y\}.
\]
Note that the number of sets $\mathcal{A}_{X,Y}$ is bounded by $2^{\mathcal O (t^2)}$.
We have $\mathcal{A}_{X,Y}\simeq\Sub(H\to G_{X,Y}^{(\ell)})$, where
$G_{X,Y}^{(\ell)}$ is obtained from $G^{(\ell)}$ as follows: Delete
all edges and vertices from $G^{(\ell)}[V_{C}]$ not contained in $X$, and
delete all edges in $I$ incident with vertices from $V_\partial \setminus Y$.
With $X^{*}:=G'[V_{C}]$ and $Y^{*}:=V_{\partial}$, it then holds
that 
\[
\mathsf{T}(\ell)=\mathcal{A}_{X^{*},Y^{*}}\setminus\bigcup_{X\subsetneq X^{*},Y\subsetneq Y^{*}}\mathcal{A}_{X,Y}.
\]
The number $\#\Sub(H\to G_{X,Y}^{(\ell)})$ can be computed by an
oracle call. Hence $\#\mathsf{T}(\ell)$ can be computed by the inclusion-exclusion formula \eqref{eq:incl-excl}. To compute the size of the intersection of different $\mathcal{A}_{X,Y}$'s, we can use $\mathcal{A}_{X,Y}\cap\mathcal{A}_{X',Y'}=\mathcal{A}_{X\cap X',Y\cap Y'}$.
\cqed\end{proof}
Consider the graphs that can be written as $H-C'$, where $C'\subseteq V(H)$
is such that $H[C]\simeq H[C']$ via an isomorphism $f$ satisfying
$(\mathrm{C2})$-$(\mathrm{C3})$ of Definition~\ref{def:gadget}. Let $\mathcal{R}$
denote the set of such graphs, modulo isomorphism. 
Note that every $R\in \mathcal R$ is a graph on $2k$ vertices. Slightly abusing
notation, we will henceforth write $A\in\mathcal{R}$ if there is
some $A'\in\mathcal{R}$ with $A\simeq A'$. 
\begin{claim}
If $F\in T(\ell)$, then $F[V_{G}]\in\mathcal{R}$.\end{claim}
\begin{proof}
Let $F\in T(\ell)$ and let $f$ be an isomorphism from $F$ to $H$
that maps $F[V_{C}]\simeq H[C]$ to $H[C']$ for some $C'\subseteq V(H)$.
We claim that $f$ satisfies $(\mathrm{C2})$-$(\mathrm{C3})$, which
implies $F[V_{G}]\in\mathcal{R}$.

Condition $(\mathrm{C2})$ holds as $F[V_{G}]$ is contained in (a copy
of) the bipartite graph $G$.  Concerning condition $(\mathrm{C3})$, as
$F\in T(\ell)$, each vertex in $V_{\partial}$ is forced to be adjacent
to at least one vertex in $F[V_{G}]$, and thus
$\partial_{F}(V_{C})\supseteq V_{\partial}$.  By construction of
$G^{(\ell)}$, only vertices in $V_{\partial}$ can be adjacent to
$F[V_{G}]$, thus $\partial_{F}(V_{C})\subseteq V_{\partial}$.  The
isomorphism $f$ must therefore map $V_{\partial}$ to
$\partial_{H}(C')$.  \cqed\end{proof} Let $\iota(R)$ denote the set of
isolated vertices of $R$. If $R\in\mathcal{R}$ satisfies
$\iota(R)=\emptyset$, then $R\simeq M$: This holds because then the
isomorphism that establishes $R\in\mathcal{R}$ satisfies
$(\mathrm{C1})$-$(\mathrm{C3})$ of Definition~\ref{def:gadget}. However, since
$\iota(R)=\emptyset$ cannot generally be assumed, it may occur that
$R\not\simeq M$. In the remainder of this proof, we will therefore
``interpolate out'' the isolated vertices. To this goal, we first
describe $\#\mathsf{T}(\ell)$ as a weighted sum over $\mathcal{R}$,
which will later be interpreted as a polynomial to perform
interpolation on.

For $R\in\mathcal{R}$, let $\alpha_{R}\in\mathbb{N}$ be defined
as follows: Given the vertex-disjoint union of $H[C]$ and $R$, let
$\alpha_{R}$ denote the number of ways of adding edges between $\delta_{H}(C)$
and $V(R)$ such that the resulting graph is isomorphic to $H$. Note
that we can compute $\alpha_{R}\in\mathbb{N}$ by brute force in time
$2^{\mathcal O (t^2)}$, and that $\alpha_{M}>0$ since $H$ contains $M$ as an induced subgraph.
\begin{claim}
With the abbreviation $\pure(R):=R-\iota(R)$, it holds that 
\begin{equation}
\label{eq:pure-isol-binomial}
\#\mathsf{T}(\ell)=\sum_{R\in\mathcal{R}}\alpha_{R}\cdot\#\Sub(\pure(R)\to G)\cdot{n+\ell-2k+|\iota(R)| \choose |\iota(R)|}.
\end{equation}
\end{claim}
\begin{proof}
By the previous claim, $F[V_{G}]\in\mathcal{R}$ for every $F\in T(\ell)$.
This induces a partition of $\mathsf{T}(\ell)$ according to the isomorphism
type of $F[V_{G}]$ for $F\in\mathsf{T}(\ell)$: For $R\in\mathcal{R}$,
let $\mathcal{A}_{R}$ denote the set of $F\in T(\ell)$ with $F[V_{G}]\simeq R$.
Since the sets $\mathcal{A}_{R}$ partition $T(\ell)$, it holds that 
\begin{equation}
\label{eq:classesA}
\#\mathsf{T}(\ell)=\sum_{R\in\mathcal{R}}\#\mathcal{A}_{R}.
\end{equation}
We show \eqref{eq:pure-isol-binomial} by computing $\#\mathcal{A}_{R}$ for $R \in \mathcal R$ and observing that it is equal to the corresponding summand in \eqref{eq:pure-isol-binomial}.
To this goal, observe first that every $F\in\mathcal{A}_{R}$ can be written as a union of a uniquely determined 
$F'\in\Sub(\pure(R)\to G)$ and the following extensions.
\begin{enumerate}[label=\roman*]
\item the graph $G^{(\ell)}[V_{C}] \simeq H[C]$, which must be contained in $F$ since $F\in T(\ell)$, 
\item a set of $|\iota(R)|$ isolated vertices in $V_G$, which may be chosen from vertices of the $G$-copy as well as from the $\ell$ added isolated vertices, and
\item a set of edges from $I$.
\end{enumerate}

For fixed $F'\in\Sub(\pure(R)\to G)$, we determine the number of such extensions.
In (i), there is only one choice.
In (ii), we extend by adding isolated vertices in $V_{G}$. 
As $V_{\partial}$ connects to all of $V_{G}$, these $|\iota(R)|$ isolated vertices can be chosen arbitrarily among the $n+\ell-|V(\pure(R))|$ vertices
of $V_{G}$ not already contained in $F'[V_{G}]$. Since $\pure(R)$ has exactly $2k-|\iota(R)|$ vertices,
the number of such extensions is given by the binomial coefficient
in (\ref{eq:pure-isol-binomial}). 

So far, this procedure has fixed $F[V_{G}]$, and we reach step (iii), where we extend the vertex-disjoint parts
$F[V_{G}]\simeq R$ and $F[V_{C}]\simeq H[C]$ to a graph $F$
by including edges between $V_{\partial}$ and $V_{G}$. The number
of possible extensions in this step is $\alpha_{R}$, by definition of $\alpha_{R}$. 
Thus, each $F'\in\Sub(\pure(R)\to G)$ can be extended to some $F\in\mathcal{A}_{R}$ by $\alpha_R \cdot {n+\ell-2k+|\iota(R)| \choose |\iota(R)|}$ possible extensions, and each $F$ is the extension of a unique $F'$, which shows the claim.
\cqed\end{proof}

The value $\#\mathsf{T}(\ell)$ can be interpreted as a polynomial
$p\in\mathbb{Z}[x]$ of maximum degree $2k$ in the indeterminate
$x:=n+\ell-2k$. Note that $n$ and $k$ are fixed by the input while
$\ell$ can be varied. 

Considering the binomial coefficient in (\ref{eq:pure-isol-binomial}) as a polynomial in $x$,
we can observe that $R\in\mathcal{R}$ with differing $|\iota(R)|$ yield polynomials of different degrees.
In particular, the binomial coefficient (as a polynomial in $x$) has degree $0$ (and is then equal to $1$) iff $|\iota(R)|=0$, i.e., when $R$ contains no isolated vertices.

We consider again the polynomial $p$ and observe that it can be interpolated by evaluating
$\#\mathsf{T}(0),\ldots,\#\mathsf{T}(2k)$,
where each evaluation can be performed with $2^{\mathcal O (t^2)}$ oracle calls by Claim~\ref{cl:enforce-C}.
This yields the representation of $p$ in the standard monomial basis $\{x^{i}\mid i\in\mathbb{N}\}$.
We then represent $p$ as a linear combination of $\mathcal{B}:=\{{x+i \choose i}\mid i\in\mathbb{N}\}$.
Since the polynomials in $\mathcal{B}$ have different degrees, $\mathcal{B}$
is linearly independent. Thus, the coefficients of $p$ over the basis
$\mathcal{B}$ are uniquely determined and we can extract the constant
coefficient $c_{0}=\alpha_{M}\cdot\#\Sub(M\to G)$. Recall that $\alpha_{M}>0$ since $M$ is an induced matching in $H$.
\end{proof}
\begin{remark}
The final interpolation step could equivalently be achieved by solving
a linear system of equations whose system matrix consists of the binomial
coefficients in (\ref{eq:pure-isol-binomial}).
In fact, this system appears implicitly in the argument above.
\end{remark}

This readily implies the reduction.

\begin{theorem}
\label{th:reduce}
If $\H$ is a recursively enumerable graph class that contains a $k$-matching gadget for every $k\in \mathbb N$,
then $\pSub(\H)$ is $\sharpWone$-complete.
\end{theorem}
\begin{proof}
We reduce $\pMatch \leqFpt \pSub(\H)$: Given a bipartite graph $G$ and $k \in \mathbb N$, we wish to compute the number of $k$-matchings in $G$.

For $H$ and $M \subseteq V(H)$, we can test whether $(H,M)$ is a $k$-matching gadget by checking all isomorphisms appearing in Definition~\ref{def:gadget} by brute force. The runtime required for this step depends only on $|V(H)|$.
To determine the number of $k$-matchings in $G$, enumerate the graphs $H \in \H$ and subsets $M \subseteq H$ until we find a $k$-matching gadget $(H,M)$, which is guaranteed to exist by assumption. Then invoke Lemma~\ref{lem:reduce} with the $k$-matching gadget $(H,M)$.

Let $g(k) = |V(H)|$, where $H$ is the $k$-matching gadget found by the enumeration procedure above.
Then the function $g$ is computable, which implies together with Lemma~\ref{lem:reduce} that the reduction to $\pSub(\H)$ is indeed a parameterized Turing reduction.
\end{proof}

\section{Bounded-degree graphs}
\label{sec:bound-degr-graphs}
The goal of this section is to prove Theorem~\ref{thm:gadgetexists},
the existence of $k$-matching gadgets, for the special case of graph
classes $\H$ with bounded maximum degree and unbounded vertex-cover
number, or equivalently, containing graphs with arbitrarily large
matchings. The results in Sections~\ref{sec:graphs-with-no} and
\ref{sec:bound-treew-graphs} for other graph classes are based on this
result for bounded-degree graphs. The basic idea is that in
bounded-degree graphs we are close to the situation described by
Lemma~\ref{lem:nocommon}: clearly, the two endpoints of an edge in the
matching can have only a bounded number of common neighbors; in this
sense property (P2) ``almost holds.'' We choose a candidate $(H,M)$
for the $k$-matching gadget and see how it can fail. If for every $C'$
satisfying \ref{i:cisolated}--\ref{i:boundary}, the graph $H\setminus
C'$ still has many components of size 2 (so it is ``almost a
matching''), then we can extract a correct $k'$-matching gadget for
some relatively large $k'<k$. Suppose therefore that $(H,C)$
``spectacularly fails'': $H\setminus C'$ has only few components of
size 2. As $H\setminus C'$ has no isolated vertices, this is only
possible if $H\setminus C'$ has many more edges than the $k$-matching
$M$. Then we argue that now the total degree on the boundary of $C'$
is much smaller than on the boundary of $C$, and we can use this to
find an induced matching in $H\setminus C'$ where the endpoints of the
edges have strictly fewer common neighbors than in $M$. As the graph
has bounded degree, repeating this argument a constant number of times
eventually leads to a matching where the endpoints of the edges have
no common neighbors, hence Lemma~\ref{lem:nocommon} can be invoked.

In a bounded-degree graph, any sufficiently large set of edges
contains a large matching and in fact a large induced matching: we can
greedily select edges and we need to throw away only a bounded number
of edges after each selection. Moreover, in order to move closer to
the situation described in Lemma~\ref{lem:nocommon}, we may also
satisfy the requirement that the selected edges have no common
neighbors (but it is possible that the two endpoints of an edge have
common neighbors).
\begin{lemma}\label{lem:makeinduceddegree}
  Let $F$ be a set of edges in a graph $G$ with
  maximum degree $D$.
\begin{enumerate}
\item There is an induced matching $M'\subseteq F$ of size at least
  $|F|/(2D^2)$.
\item There is an induced matching $M''\subseteq F$ of size at least
  $|F|/(2D^3)$ such that every vertex of
  $V(G)\setminus V(M'')$ is adjacent to at most one edge of $M''$.
\end{enumerate}
\end{lemma}
\begin{proof}
  (1) Let $uv$ be an edge of $F$. As the graph has maximum degree at
  most $D$, there are less than $2D^2$ edges having an endpoint in the
  closed neighborhood of $\{u,v\}$. We perform the following greedy
  procedure: we select an arbitrary edge $uv$ of $F$ into $M'$ and
  then remove every edge from $F$ that has an endpoint in the closed
  neighborhood of $\{u,v\}$. In each step, we reduce the size of $F$
  by at most $2D^2$. Therefore, the procedure runs for at least
  $|F|/(2D^2)$ steps, producing an induced matching of the required
  size.

  (2) Observe that there are less than $2D^3$ edges of $F$ at distance
  at most 2 from $uv$. Then a greedy algorithm can produce a set of
  edges of size at least $|F|/(2D^3)$ such that the edges are at
  pairwise distance at least 3, that is, $M'$ is an induced matching
  and every vertex outside $M$ is adjacent to at most one edge of $M$.
\end{proof}
For bounded-degree graphs, Lemma~\ref{lem:makeinduceddegree} implies
that there is not much difference between having a large set of edges,
a large matching, a large induced matching, or a large induced
matching satisfying the requirement that every vertex outside the
matching is adjacent to at most one edge of the matching.

\begin{lemma}\label{lem:gadgetmainlow}
There is a function $\fdegree(k_0,D)$ such that the following holds. If $H$ is a graph with maximum degree at most $D$ and contains a matching of size at least $\fdegree(k_0,D)$, then there is a $k_0$-matching gadget $(H,M_0)$.
\end{lemma}
\begin{proof}
We prove the following statement by induction on $c$:
\begin{quote}
  There is a function  $\fdegree'(c,k_0,D)$ such that the following
  holds. If $H$ is a graph with maximum degree at most $D$ and having a set $F$ of at least
  $\fdegree'(c,k_0,D)$ edges such that $u$ and $v$ have at most $c$
  common neighbors for every $uv\in F$, then there is a $k_0$-matching gadget $(H,M_0)$.
\end{quote}
If this statement is true for every $c$, then we can set
$\fdegree(k_0,D)=\fdegree'(D,k_0,D)$: if the graph has maximum degree
of $H$ is $D$, then it is clear that every edge of a matching has the
property that the endpoints have at most $D$ common neighbors.

For $c=0$, we can set $\fdegree'(0,k_0,D)=2k_0D^3$. Then Lemma~\ref{lem:makeinduceddegree}(2) implies that $F$ contains an induced matching $M_0$ of size $k_0$ such that every vertex of $V(H)\setminus V(M_0)$ is adjacent to at most one edge of $M_0$. Furthermore, $c=0$ implies that every vertex of $V(H)\setminus V(M_0)$ can be adjacent to at most one endpoint of at most edge of $M_0$. Therefore, Lemma~\ref{lem:nocommon} implies that $(H,M_0)$ is a $k_0$-matching gadget.

Suppose now that the statement is true for $c-1$ with some value of
$\fdegree'(c-1,k_0,D)$; we show that the statement is true for $c$ with 
\[
\fdegree'(c,k_0,D)=2D^3(5c\cdot \fdegree'(c-1,k_0,D)+100k_0).
\]
If $F$ has at least this size, then Lemma~\ref{lem:makeinduceddegree}(2)
implies that there is an induced matching $M\subseteq F$ of size at least
$5c\cdot \fdegree'(c-1,k_0,D)+100k_0$ such that every vertex in
$V(H)\setminus V(M)$ is adjacent to at most one edge of $M$.

If $(H,M)$ is a matching gadget, then we are done by Lemma~\ref{lem:subgadget} as $|M| \geq k_0$. Otherwise,
let $C:=V(H)\setminus V(M)$ and let $f$ be an isomorphism from $H[C]$
to $H[C']$ satisfying \ref{i:cisolated}--\ref{i:boundary} of
Definition~\ref{def:gadget} such that $H\setminus C'$ is not a
matching. For any graph $G$, let us denote by $\ktwo(G)$ the number of components of $G$
with exactly two vertices. Let us choose $f$ such that
$k':=\ktwo(H\setminus C')$ is minimum possible.
\begin{claim}\label{cl:minmatching}
If $k'\ge k_0$, then there is a $k_0$-matching gadget $(H,M_0)$.
\end{claim}
\begin{proof}
  Let $M_0$ be the induced matching obtained as the union of all the
  $k'=\ktwo(H\setminus C')\ge k_0$ components of $H\setminus C'$ having
  two vertices each (see Figure~\ref{fig:k2min}). We show that $(H,M_0)$ is a $k'$-matching gadget; then the existence of a $k_0$-matching gadget
  follows from Lemma~\ref{lem:subgadget}.

\begin{figure}
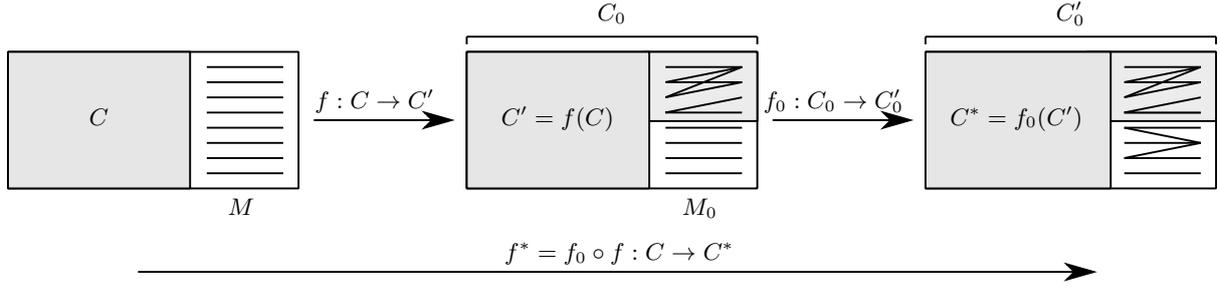

\begin{center}
{\footnotesize \svg{\linewidth}{k2min}}
\caption{Proof of Claim~\ref{cl:minmatching} in Lemma~\ref{lem:gadgetmainlow}.}\label{fig:k2min}
\end{center}
\end{figure}
Let $C_0:=V(H)\setminus V(M_0)$; we have $C'\subseteq C_0$.  Let us
point out that there is no edge between $V(M_0)$ and $V(H)\setminus
(C'\cup M_0)=C_0 \setminus C'$: the edges of $M_0$ are components of
$V(H)\setminus C'$.  Suppose that $f_0$ is an isomorphism from
$H[C_0]$ to some $H[C'_0]$ showing that $(H,M_0)$ is not a
$k'$-matching gadget: $f_0$ satisfies properties
\ref{i:cisolated}--\ref{i:boundary}, but $H\setminus C'_0$ is not a
matching. Let $f^*=f_0 \circ f$ (that is, applying $f_0$ after $f$) be
the isomorphism from $H[C]$ to $H[C^*]$ with
$C^*=f_0(f(C))=f_0(C')\subseteq C'_0$. Let us verify that $f^*$
satisfies the three conditions of Definition~\ref{def:gadget} with
respect to $(H,M)$:
\begin{enumerate}[label=(C\arabic*)]
\item Consider a vertex $v\in H\setminus C^*$. If $v\not\in C'_0$,
  then property \ref{i:cisolated} of $f_0$ implies that $v$ is not
  isolated in $H\setminus C'_0$ and hence it is not isolated in $H\setminus C^*$ either. Suppose therefore that $v\in C'_0\setminus C^*$, that is, there is a
  $u\in C_0\setminus C'$ with $f_0(u)=v$. As $H\setminus C'$ has no
  isolated vertex, there has to be a neighbor $w$ of $u$ in
  $H\setminus C'$. If $w\in C_0\setminus C'$, then $f_0(w)$ is a
  neighbor of $f_0(u)=v$ in $H\setminus C^*$. Suppose now that $w\not
  \in C_0$, which implies that $u$ is in $\boundary_H(C_0)$. Therefore, property~\ref{i:boundary} of $f_0$ implies that $f_0(u)$ has a neighbor outside
  $C'_0$, hence $f_0(u)$ has a neighbor in $H\setminus C^*$.
\item By \ref{i:cbipartite} on $f$, the graph $H[C_0\setminus C]$ is
  bipartite. Mapping $f_0$ is an isomorphism between $H[C_0]$ and
  $H[C'_0]$ that maps $C_0\setminus C$ to $C'_0\setminus C^*$, thus
  $H[C'_0\setminus C^*]$ is bipartite as well. We have observed above
  that there is no edge between $C_0\setminus C'$ and $M_0$, thus
\ref{i:boundary} on $f_0$ implies that there is no edge between
  $C'_0\setminus C^*$ and $V(H)\setminus C'_0$. Finally,
  $V(H)\setminus C'_0$ is bipartite by \ref{i:cbipartite} on $f_0$,
  thus we get that $H\setminus C^*$ is bipartite.
\item Let $v$ be a vertex in $C$. By \ref{i:boundary} on $f$,
  we have that $v\in \boundary_H(C)$ if and only if $f(v)\in \boundary_H(C')$.
As there are no edges going from $C_0\setminus C'$ to outside $C_0$, we have $\boundary_H(C')=\boundary_H(C_0)$.
Furthermore, by
  \ref{i:boundary} on $f_0$, we have that $u\in \boundary_H(C_0)$ if and only if $f_0(u)\in \boundary_H(C'_0)$.
Putting together, we have that $v\in \boundary_H(C)$ if and only if $f^*(v)=f_0(f(v))\in \boundary_H(C'_0)$.
\end{enumerate}
We claim that $\ktwo(H\setminus C^*)<\ktwo(H\setminus C')$, contradicting the minimal choice of
$f$. 
We have stated above that
there is no edge between $C_0\setminus C'$ and $V(M_0)$; as $f_0$ is boundary preserving, it also follow that there is no edge between $C'_0\setminus C^*$ and $V(H)\setminus C'_0$.  Therefore, 
the components of $H\setminus C^*$ are exactly the components of $H[C'_0\setminus C^*]$ and the components of $H\setminus C'_0$.
Recall that $H[C_0'\setminus C^*]$
is isomorphic to $H[C_0\setminus C']$, which has no 2-vertex component by the definition of $M_0$.
As $H\setminus C'_0$ is not a matching, we have that $H\setminus C'_0$ (and hence $H\setminus C^*$) has strictly less than $|M_0|$ 2-vertex components, that is, $\ktwo(H\setminus
C^*)< \ktwo(H\setminus C')$.  \cqed\end{proof}

In the following, we suppose that $\ktwo(H\setminus C')<k_0$. Let
$s=10c\cdot \fdegree'(c-1,k_0,D)+200k_0$ be the number of vertices of
$H\setminus C$ and $H\setminus C'$. Note that $H\setminus C=M$ is an
induced matching having exactly $s/2$ edges.  We can give a lower
bound on the number of edges of $H\setminus C'$ by giving an upper
bound on the number of additional edges required to make the graph
connected.  By \ref{i:cisolated} of $f$, every component has size at
least two in $H\setminus C'$.\footnote{This the point where we
  crucially use \ref{i:cisolated} of Definition~\ref{def:gadget}.}
There are less than $k_0$ components with exactly two vertices, thus
by adding at most $k_0$ edges, we can ensure that every component contains at
least three vertices. Then there are at most $s/3$ components with at
least three vertices, hence we can connect them with $s/3$ additional
edges.  It follows that there are at least $s-(s/3+k_0)=2s/3-k_0\ge
0.6s$ edges in $H\setminus C'$ (using $s\ge 200k_0$). As the number of
edges in $H[C]$ and $H[C']$ are the same and $H\setminus C$ has $0.5s$
edges, it follows that the number of edges going out of $C'$ is
less than the number of edges going out of $C$ by at least
$0.1s$.

Let $B:=\boundary_H(C)$ and $B':=\boundary_H(C')$; by \ref{i:boundary} of Definition~\ref{def:gadget}, we have $|B|=|B'|$.  Recall that every
vertex of $B$ has either 1 or 2 edges to $V(M)$. As each of the $s/2$
edges of $M$ has at most $c$ common neighbors in $B$, there are at
most $c(s/2)$ vertices of $B$ with two edges to $V(M)$, implying that
there are at most $|B|+c(s/2)$ edges going out of $C$. Therefore,
there are at most $|B|+c(s/2)-0.1s$ edges going out of $C'$.
Every vertex of $B'=\boundary_H(C')$ has at least one edge going to
$V(M)$. Let us remove one such edge from each vertex of $B'$, then a
set $T$ of at most $c(s/2)-0.1s$ edges remain.\footnote{This the point where we crucially use \ref{i:boundary} of Definition~\ref{def:gadget}: we need that $|B|=|B'|$.} Let $B'_{\ge 2}$ be the
subset of $B'$ containing those vertices that have at least two edges
going to $V(H)\setminus C'$. Then the total number of edges going from $B'_{\ge 2}$
to $V(H)\setminus C'$ is at most $2|T|\le cs-0.2s=(c-0.2)s$: each vertex of
$B'_{\ge 2}$ has at least one edge in $T$ and, in addition to that,
can have at most one edge not in $T$.

As the number of edges between $B'_{\ge 2}$ and $V(H)\setminus C'$ is at most
$(c-0.2)s$, at most $(c-0.2s)s/c$ vertices of $V(H)\setminus C'$ can have at
least $c$ neighbors in $B'_{\ge 2}$.  Let $X$ be the set of vertices
in $V(H)\setminus C'$ with at most $c-1$ neighbors in $B'_{\ge 2}$, we
have that $|X|\ge s-((c-0.2)s)/c= 0.2s/c$. By \ref{i:cisolated} of
Definition~\ref{def:gadget} on $f$, there are no isolated vertices in
$H\setminus C'$. For each vertex in $X$, let us select an edge of $H\setminus C'$
incident to it; this way, we select a set $F^*$ of least $|X|/2\ge
0.1s/c\ge \fdegree'(c-1,k_0,D)$ distinct edges.  Consider an edge
$uv$ in $F^*$. Vertices $u$ and $v$ have no common neighbors in
$H\setminus C'$: this would contradict \ref{i:cbipartite} of
Definition~\ref{def:gadget} stating that $H\setminus C'$ is
bipartite.\footnote{This the point where we crucially use \ref{i:cbipartite} of Definition~\ref{def:gadget}.} Therefore, every such common neighbor is in $B'_{\ge
  2}$. One of $u$ and $v$ is in $X$, which means that it has at most
$c-1$ neighbors in $B'_{\ge 2}$, implying that $u$ and $v$ can have at
most $c-1$ common neighbors. As $|F^*|\ge \fdegree'(c-1,k_0,D)$, the
induction assumption implies that there is a $k_0$-gadget.
\end{proof}

\section{Graphs with no large subdivided stars}
\label{sec:graphs-with-no}

A {\em subdivided $\ell$-star} consists of a center vertex $v$ and
$\ell$ paths of length 2 starting at $v$ that do not share any vertex
other than $v$. We denote by $\stwo(v)$ the largest integer $\ell$ such
that $v$ is the center of a subdivided $\ell$-star. We denote by $\stwo(G)$ the maximum of $\stwo(v)$ for every $v\in V(G)$.
The goal of this section is to prove Theorem~\ref{thm:gadgetexists}, the existence of $k$-matching gadgets, for graphs where $\stwo(G)$ is bounded.

We develop a technology that allows us to ``ignore'' certain sets
$Q$ of vertices: if $H\setminus Q$ has a $k$-matching gadget, then so
does $H$. This works for sets $Q$ where the vertices have some characteristic property (e.g., based on degrees) that allows us to distinguish them from the vertices 
not in $Q$ (see below). We use this technique to reduce the problem to bounded-degree graphs. If we have a large induced matching where every vertex has small degree, then we define $Q$ to be the vertices of ``large degree.'' Now $H\setminus Q$ is clearly a bounded-degree graph and hence Lemma~\ref{lem:gadgetmainlow} can be invoked. Suppose therefore that we have an induced matching where every vertex has large degree. Then we define $Q$ to be the vertices of ``small degree.'' Somewhat unexpectedly, $H\setminus Q$ is a bounded-degree graph also in this case: this follows from the fact that if $\stwo(G)$ is bounded, then a vertex cannot have many neighbors of large degree.
\begin{proposition}\label{prop:starneighbor}
Every vertex $v\in V(G)$ has at most $\stwo(v)$ neighbors with degree at least $2\stwo(v)+2$.
\end{proposition}
\begin{proof}
  Let $\ell=\stwo(v)+1$ and suppose that vertex $v$ is adjacent to
  vertices $\alpha_1$, $\dots$, $\alpha_\ell$, each having degree at
  least $2\ell$.  Then for every $1\le i \le \ell$, as the degree of
  $\alpha_i$ is at least $2\ell$, we can find a vertex $\beta_i$ that
  is adjacent to $\alpha_i$, but is not in the set
  $\{v,\alpha_1,\dots, \alpha_\ell,\beta_1,\dots, \beta_{i-1}\}$ (note
  that this set of at most $2\ell$ vertices contains less than $2\ell$
  vertices different from $\alpha_i$). This creates a subdivided
  $\ell$-star centered at $v$, a contradiction.
\end{proof}
Therefore, we can reduce the problem to bounded-degree graphs also
in the case of a matching with large degree vertices. Finally, if we
have a large induced matching with ``mixed'' edges, that is, each
having both a small-degree and a large-degree endpoint, then we can reduce
to one of the previous two cases by looking at the common neighbors of
the endpoints.

The following definition will be crucial for the clean treatment of
the problem. We show that if a set is ``well identifiable'' (for
example, based on degrees etc.) then we can remove it from the graph
and it is sufficient to show that the remaining part of the graph has a
$k$-matching gadget. The definition formulates this condition as
invariance under certain isomorphisms.

\begin{definition}\label{def:strong}
Let $H$ be a graph and let $X\subseteq C\subseteq V(H)$ be two subsets of vertices. We say that $X$ is a {\em strong set} with respect to $C$ if whenever $f$ is a boundary-preserving isomorphism from $H[C]$ to $H[C']$ for some $C'\subseteq V(H)$, then $f(X)=X$ (in particular, this implies $X\subseteq C'$).
\end{definition}
Observe that $f(X)$ and $X$ have the same size, thus to prove
$f(X)=X$, it is sufficient to prove $f(X)\subseteq X$, that is, $v\in
X$ implies $f(v)\in X$.

As a simple example, suppose that every vertex in $H$ has either
degree at most $d$ or degree at least $d+2k+1$ and $M\subseteq H$ is a
$k$-matching with every vertex having degree at most $d$ in $H$. Let
$C=V(H)\setminus V(M)$ and let $X\subseteq C$ be the set of vertices
with degree at least $d+2k+1$. Then $X$ is a strong set: every vertex
$x\in X$ has at least $d+2k+1-|V(M)|=d+1$ neighbors in $C$, hence
$f(v)$ has at least $d+1$ neighbors in $C'$, implying that $f(v)\in X$
(as we assumed that degree larger than $d$ implies that the degree is
at least $d+2k+1$). In fact, it is sufficient to enforce the degree
requirement only for vertices $v\in \boundary_H(C)$: it is sufficient
if we require that the degree of every vertex in $\boundary_H(C)$ is
either at most $d$ or at least $d+2k+1$, but the degrees of the
vertices in $C\setminus \boundary_H(C)$ can be arbitrary. This is
sufficient, as if $v\in C\setminus \boundary_H(C)$, then every
neighbor of $v$ is in $C$ and \ref{i:boundary} of $f$ implies that
every neighbor of $f(v)$ is in $C'$, hence (as $H[C]$ and $H[C']$ are
isomorphic) vertices $v$ and $f(v)$ have exactly the same degree.

We show now that removing a strong set disjoint from $M$ does not
affect whether a $k$-matching gadget is correct.
\begin{lemma}\label{lem:removestrong}
  Let $H$ be a graph containing an induced $k$-matching $M$,
  let $C:=V(H)\setminus V(M)$, and let $X\subseteq C$ be a strong set with respect to $C$. 
  If $(H\setminus X,M)$ is a $k$-matching gadget, then so is $(H,M)$.
\end{lemma}
\begin{proof}
  Suppose that $f$ is an isomorphism from $H[C]$ to $H[C']$ for some
  $C'\subseteq V(H)$ satisfying \ref{i:cisolated}--\ref{i:boundary} of
  Definition~\ref{def:gadget}, but $H\setminus C'$ is not a
  matching. Let $H^*=H\setminus X$ and let $f^*$ be the restriction of
  $f$ to $C\setminus X$. As $X$ is a strong set, we have $f(X)=X$ and
  hence $f(C\setminus X)=C'\setminus X$, that is, $f^*$ induces an an
  isomorphism from $H^*[C\setminus X]$ to $H^*[C'\setminus
  X]$. Observe that $H^*\setminus (C'\setminus X)=H\setminus C'$ has
  no isolated vertex and bipartite (as $f$ satisfies properties
  \ref{i:cisolated} and \ref{i:cbipartite}). Furthermore, $f^*$ is
  boundary preserving (as $\boundary_{H\setminus X}(C\setminus
  X)=\boundary_H(C)$ and $\boundary_{H\setminus X}(C'\setminus
  X)=\boundary_H(C')$).  Therefore, the fact that $(H^*,M)$ is a
  $k$-matching gadget implies that $H^*\setminus (C'\setminus
  X)=H\setminus C'$ is a $k$-matching, what we had to show.
\end{proof}

Similarly to bounded-degree graphs
(Lemma~\ref{lem:makeinduceddegree}), we can use a bound on $\stwo(H)$
to argue that not too many edges can be in the neighborhood of an edge
and therefore a large set of edges implies a large induced
matching. However, all we need now is that a large induced matching
implies that there is a large induced matching such that every vertex
outside the matching is adjacent to at most one edge of the matching.

\begin{lemma}\label{lem:makeinduced}
  Let $M$ be an induced matching of size at least $2kL^2$ in a graph $H$ with
  $\stwo(G)\le L$. Then there is an $M'\subseteq M$ of size at least
  $k$ such that every vertex of
  $V(G)\setminus V(M')$ is adjacent to at most one edge of $M'$.
\end{lemma}
\begin{proof}
  Let $uv$ be an edge of $M$.  As $M$ is an induced matching, every
  other edge of $M$ is at distance at least 2 from $uv$. We claim that
  at most $2L^2$ edges of $M$ are at distance exactly 2 from
  $uv$. Assume, without loss of generality, that there is a set
  $M^*\subseteq M$ of at least $L^2$ edges different from $uv$ at
  distance exactly 2 from $u$. If a neighbor $w$ of $u$ is adjacent to
  $L$ distinct edges of $M$, then there is a subdivided $L$-star
  centered at $w$, a contradiction. Thus every neighbor of $u$ is
  adjacent to at most $L$ distinct edges of $M$, which means that
  there are at least $|M^*|/L=L$ neighbors of $u$, each adjacent to a
  distinct edge of $M^*$. Then $u$ is the center of an $L$-star, a
  contradiction. Thus we have shown that for every edge of $M$, there
  are at most $2L^2$ edges at distance exactly 2 from it. This means
  that we can greedily select a subset $M'\subseteq M$ of edges at
  pairwise distance is more than 2, that is, $M'$ is an induced
  matching and every vertex outside $M'$ is adjacent to at most one
  edge of $M'$.
\end{proof}
Recall the example after Definition~\ref{def:strong}: if there is a
sufficiently large ``gap'' in the degrees of the vertices of $N(V(M))$
for a matching $M$, then we can define a strong set simply based on
the degrees. The following lemma creates such a gap of arbitrary large
size by throwing away at most half of the edges of a matching.
 
\begin{lemma}\label{lem:degreegap}
  Let $F$ be an induced matching in a graph $H$ with $\stwo(H)\le L$. For every $x\ge 2L+2$, $y\ge 1$, there is an induced matching
  $F'\subseteq F$ of size at least $|F|/2$ and an $x \le g \le
  x+4(2L+2)y$ such that $N(V(F'))$ has no vertex whose degree in $H$ is in the
  range $\{g,\dots, g+y-1\}$.
\end{lemma}
\begin{proof}
  Let   $B:=N(V(F))$ and let $B_i$ be the subset of $B$ containing those vertices that have
  degree exactly $i$ in $H$ and let $x_i$ be the number of edges
  between $B_i$ and $V(F)$.  Let $X_j=\sum_{i=x+jy}^{x+(j+1)y-1}x_i$.
  Proposition~\ref{prop:starneighbor} implies that for every $i\ge x
  \ge 2L+2$, every vertex in $V(F)$ has at most $2L+2$ neighbors in $B_i$,
  thus $\sum_{j=0}^{4(2L+2)-1}X_j=\sum_{i=x}^{x+4(2L+2)y-1}x_i\le
  (2L+2)\cdot 2|F|$. Thus by an averaging argument, there is a $0\le j^* \le
  4(2L+2)$ such that $X_{j^*}\le |F|/2$. Therefore, if we construct a matching $F'\subseteq F$ by throwing away every edge
  of $F$ adjacent to $B_i$ for some $x+j^*y\le i \le x+(j^*+1)y-1$, then we
  get a matching $F'$ of size at least $|F|/2$ such that no vertex in
  $N(V(F'))\subseteq B$ (recall that $F$ is an induced matching) has degree in the range $\{x+j^*y,\dots,x+j^*y+y-1\}$. That
  is, the statement is true with $g=x+j^*y$.
\end{proof}

Now we are ready to prove that main result for graphs not having large
subdivided stars. The proof uses Lemma~\ref{lem:removestrong} to
remove a set of vertices, making the graph bounded degree, and then the
bounded-degree result Lemma~\ref{lem:gadgetmainlow} can be invoked.
\begin{lemma}\label{lem:gadgetmainnostar}
  There is a function $\fstar(k_0,L)$ such that if graph $H$ with
  $\stwo(H)\le L$ has an induced matching of size $\fstar(k_0,L)$, then there is a
  $k_0$-matching gadget $(H,M_0)$.
\end{lemma}
\begin{proof}
We define the following constants 
(the function $\fdegree$ is from Lemma~\ref{lem:gadgetmainlow}):
\begin{alignat*}{1}
\kh&=2\fdegree(k_0,L)\\
\Dh&= 2L+2+8(2L+2)\kh+2\kh\\
\Dl&=\Dh+8(2L+2)L\\
\kl&=2\fdegree(k_0,\Dl)\\
\kx&=2k_0L^2+2\Dh^2 \kl+2L^2\kh\\
\fstar(k_0,L)&=\kh+\kl+\kx.
\end{alignat*}
Let $M$ be an induced matching of size $\fstar(k_0,L)$ in $H$.  The
edges of $M$ are of three types: either both endpoint have degree at
most $\Dh$, or only one of them, or neither. We show that if $M$
contains a large number of edges from any of the three types, then the
$k_0$-matching gadget exists. In the case when there is a large
matching with degrees bounded by $\Dh$, then we identify a strong set
containing the high-degree vertices, remove them by applying
Lemma~\ref{lem:removestrong}, and invoke Lemma~\ref{lem:gadgetmainlow}
on the remaining bounded-degree graph.
\begin{claim}\label{cl:low}
  If there is an induced matching $F$ of $\kl$ edges in $H$ such that
  the endpoints of these edges have degree at most $\Dh$ in $H$, then
  there is a $k_0$-matching gadget.
\end{claim}
\begin{proof}
  Let us apply Lemma~\ref{lem:degreegap} on $F$ with
  $x=\Dh+1$ and $y=2L$. Then we get a matching $F'\subseteq F$ and a $\Dh+1\le g \le \Dh+1+8(2L+2)L=\Dl+1$ such that
  there is no vertex in $N(V(F'))$ with degree in the range
  $\{g,\dots,g+2L-1\}$. Let $Q$ be the set of vertices with degree
  at least $g$ (observe that $g\ge \Dh+1$ implies that $Q\cap
  V(F)=\emptyset$). We claim that $Q$ is a strong set with respect to
  $C:=V(H)\setminus V(F')$. Let $f$ be a boundary-preserving isomorphism from $H[C]$ to
  $H[C']$. Consider a vertex $v\in Q$. If $v\not\in\boundary_H(C)$, then
  every neighbor of $v$ is in $C$, hence $f(v)$ has at least $g$
  neighbors in $C'$, implying $f(v)\in Q$.  Suppose now that
  $v\in \boundary_H(C)$. If $v$ is adjacent to $L+1$ distinct edges of $F'$, then $v$ is the
  center of an $(L+1)$-star, a contradiction. Therefore, there are at
  most $2L$ edges between $v$ and $V(F')$. As $v\in N(V(F'))$, the way $F'$ was defined by 
  Lemma~\ref{lem:degreegap} ensures that the degree of $v$ is at least $g+2L$,
  thus it has at least $g$ neighbors inside $C$. It follows that
  $f(v)$ has at least $g$ neighbors inside $C'$ and hence $f(v)\in
  Q$. We have shown that $v\in Q$ implies $f(v)\in Q$ and hence $Q$ is
  a strong set. The graph $H\setminus Q$ has maximum degree at most
  $g-1\le \Dh+8(2L+2)L=\Dl$ and has a matching $F'$ of size $|F'|\ge
  |F|/2=\fdegree(k_0,\Dl)$. Therefore, Lemma~\ref{lem:gadgetmainlow}
  implies that there is a $k_0$-matching gadget in $H\setminus Q$ and
  then it follows by Lemma~\ref{lem:removestrong} that $H$ has a
  $k_0$-matching gadget as well.  \cqed\end{proof}

In the case when every vertex of the matching has high degree in $H$,
then we identify a strong set containing all the low-degree vertices.
We argue that after removing this strong set, the remaining graph
has bounded degree: as we have a bound on $\stwo(H)$,
Proposition~\ref{prop:starneighbor} implies that a vertex cannot have
many high-degree neighbors. Therefore, we are again in a situation
when Lemma~\ref{lem:gadgetmainlow} can be invoked.

\begin{claim}\label{cl:high}
  If there is an induced matching $F$ of $\kh$ edges in $H$ such that
  the endpoints of these edges have degree at least $\Dh$, then
  there is a $k_0$-matching gadget.
\end{claim}
\begin{proof}
Let us apply
  Lemma~\ref{lem:degreegap} on $F$ with $x=2L+2$ and $y=2|F|$. Then we get a matching $F'\subseteq F$ and a
  $2L+2\le g \le 2L+2+8(2L+2)|F|$ such that there is no vertex in $N(V(F'))$
  with degree in the range $\{g,\dots,g+2|F|-1\}$. Let $Q$ be the set
  of vertices with degree less than $g+2|F|$ (observe that $g+2|F|\le
 2L+2+8(2L+2)|F|+2|F|=\Dh$ implies that $Q\cap V(F)=\emptyset$). We claim that $Q$ is
  a strong set with respect to $C:=V(H)\setminus V(F')$. Let $f$ be a boundary-preserving 
  isomorphism from $H[C]$ to $H[C']$. Consider a vertex $v\in Q$. If
  $v\not\in \boundary_H(C)$, then every neighbor of $v$ is in $C$. As $f$ is boundary preserving, we have $f(v)\not\in \boundary_H(C')$, hence every
  neighbor of $f(v)$ is in $C'$. Furthermore, $f$ is an isomorphism between $H[C]$ and $H[C']$, hence $f(v)$ has the same degree as $v$, that is, less
  than $g+2|F|$, implying $f(v)\in Q$.  Suppose now that $v\in
 \boundary_H(C)$. By the way $F'$ was defined by Lemma~\ref{lem:degreegap}, the degree of
  $v$ is actually less than $g$, thus it has at most $g$ neighbors inside
  $C$. It follows that $f(v)$ has at most $g$ neighbors inside
  $C'$. Clearly, $v$ can have at most $2|F'|\le 2|F|$ neighbors outside $C'$,
  hence $f(v)\in Q$. We have shown that $v\in Q$ implies $f(v)\in Q$
  and hence $Q$ is a strong set. The graph $H\setminus Q$ has maximum
  degree at most $L$: otherwise $H$ has a vertex with more than $L$
  neighbors with degree at least $g+2|F|\ge 2L+2$ adjacent to it,
  contradicting Proposition~\ref{prop:starneighbor}.  As $F'$ is a
  matching of size $|F'|\ge |F|/2=\fdegree(k_0,L)$,
  Lemma~\ref{lem:gadgetmainlow} implies that there is a $k_0$-matching
  gadget in $H\setminus Q$. It follows by Lemma~\ref{lem:removestrong}
  that $H$ has a $k_0$-matching gadget as well.  \cqed\end{proof}

The remaining case is when we have a large induced matching where each edge
has both a high-degree and a low-degree endpoint. If we have many edges
where the endpoints have no common neighbors, then we can invoke
Lemma~\ref{lem:nocommon}. Otherwise, if many of the edges have
low-degree common neighbors, then we can use this common neighbor and
the low-degree endpoint to create a matching where every vertex has
low degree and apply Claim~\ref{cl:low}. Similarly, if many of the
edges have high-degree common neighbors, then we can create a matching
with high-degree vertices and apply Claim~\ref{cl:high}.
\begin{claim}\label{cl:mixed}
 If there is a matching $F$ of $\kx$ edges in $H$ such that
  the every edge in $F$ has an endpoint with degree at most $\Dl$ and a degree at least $\Dh$, then
  there is a $k_0$-matching gadget.
\end{claim}
\begin{proof}
  Let $F_0\subseteq F$ contain those edges $uv$ in $F$ for which $u$
  and $v$ have no common neighbors. If $|F_0|\ge 2k_0L^2$, then
  Lemma~\ref{lem:makeinduced} implies that there is an induced
  matching $F'_0\subseteq F_0$ of size $k_0$ such that every vertex of
  $V(H)\setminus V(F'_0)$ is adjacent to at most one edge of $F'_0$.
  In fact, as the endpoints of the edges in $F'_0$ have no common
  neighbors, every vertex of $V(H)\setminus V(F'_0)$ is adjacent to at
  most one vertex of $V(F'_0)$.  Then Lemma~\ref{lem:nocommon} implies
  that there is a $k_0$-matching gadget.

Suppose therefore that $F\setminus F_0$ has size at least
$\kx-2k_0L^2=2\Dh^2 \kl+2L\kh$. For every edge $e\in F$, let us pick a
common neighbor $w_e$ of the endpoints of $e$ and let us partition
$F\setminus F_0$ into $F_L$ and $F_H$ depending on whether $w_e$ has
degree less than $\Dh$ or at least $\Dh$, respectively.  If $|F_L|\ge
2\Dh^2 \kl$, then, for every $e\in F_L$, let $F^*_L$ contain the edge
between $w_e$ and the endpoint of $e$ with degree at most $\Dh$.  As
every vertex of $F^*_L$ has degree at most $\Dh$ in $H$, Lemma~\ref{lem:makeinduceddegree}(1) implies that we can select a subset of $F^*_L$ of size at least
$|F^*_L|/(2\Dh^2)\ge \kl$ that forms an induced matching, and then
Claim~\ref{cl:low} implies that there is a $k_0$-matching gadget.

Otherwise, we have that $|F_H|\ge 2L^2\kl$ and, for every $e\in F_H$,
we let $F^*_H$ contain the edge between $w_e$ and the endpoint of $e$
with degree at least $\Dh$. As every endpoint of every edge in $F^*_H$
has degree at least $\Dh\ge 2L+2$, Proposition~\ref{prop:starneighbor}
implies that the graph induced by these vertices has maximum degree at
most $L$. Thus by Lemma~\ref{lem:makeinduceddegree}(1), we can select
a subset of $F^*_H$ that forms an induced matching of size at least
$|F^*_L|/(2L^2)\ge \kh$, and then Claim~\ref{cl:high} implies that
there is a $k_0$-matching gadget.  \cqed\end{proof}

Claims \ref{cl:low}--\ref{cl:mixed} prove the lemma: if $M$ has size
at least $f_s(k_0,L)$, then at least one the three cases hold.
\end{proof}

\section{Bounded-treewidth graphs}
\label{sec:bound-treew-graphs}

In this section, we complete the proof of Theorem~\ref{thm:gadgetexists} by showing that if
a bounded-treewidth graph has large vertex-cover number, then it
contains a $k$-matching gadget. First, we need an induced
matching. Lemma~\ref{lem:hereditaryramsey} shows that if a graph has
large vertex-cover number, then it contains either a large clique, a
large induced biclique, or a large induced matching. As the first two
conclusions are not possible in bounded-treewidth graphs, it follows
that there is a large induced matching. We give another proof (Lemma~\ref{lem:nicematching}) of this
fact by looking at the tree decomposition instead of using Ramsey arguments.
This proof gives a better correspondence between the size of the
induced matching and the vertex-cover number (but this does not matter
for our purposes). More importantly, the proof presented below
finds an induced matching such that
$\stwo(v)$ is bounded for every vertex $v$ of the matching, that is, there
are no large subdivided stars centered on them.  Then we define $Q$ to
be the set of vertices with large $\stwo$-number (this require some care) and use the
technology developed in Section~\ref{sec:graphs-with-no}
(Lemma~\ref{lem:removestrong}) to argue that it is sufficient to find
a $k$-matching gadget in $H\setminus Q$. Clearly, $\stwo(H\setminus
Q)$ is bounded, hence Lemma~\ref{lem:gadgetmainnostar} can be
invoked.

\begin{lemma}\label{lem:nicematching}
Let $w$ and $k$ be integers and let $H$ be a graph of treewidth at most $w$ and vertex cover number greater than $3k(w+1)$. Then there is an induced matching $M=\{u_1v_1,\dots, u_kv_k\}$ such that $\stwo(u_i),\stwo(v_i)\le 2(w+1)$ for every $1\le i \le k$.
\end{lemma}
\begin{proof}
  Consider a rooted tree decomposition $(T,\B)$ of $H$ having width $w$.  For every $t\in
  V(T)$, we denote by $H_t$ the graph induced by the union of every bag $B_{t'}$ for every
  descendant $t'$ of $t$ (including $t$ itself). For $i=0,1,\dots$, we iteratively construct an induced matching $M_i$ and a subset $X_i$ of $V(T)$; initially, $M_0=X_0=\emptyset$.
We define (see Figure~\ref{fig:case1})
\begin{itemize}
\item  $S_i:=\bigcup_{t\in X_i}B_t$,
\item $V_i:=\bigcup_{t\in X_i}\bigcup_{\text{$t'$ is a descendant of $t$}}B_{t}$,
\item $\hX_i\subseteq X_i$ to be the maximal elements of $X_i$ (i.e., those vertices of $X_i$ that have no proper ancestor in $X_i$), and
\item $\hS_i:=\bigcup_{t\in \hX_i}B_t$.
\end{itemize}
Note that if a vertex $v$ of $V_i$ has a neighbor outside $V_i$, then $v\in \hS_i$, that is, $\hS_i$ separates $V_i\setminus \hS_i$ from $V(H)\setminus V_i$.

We maintain the following invariant properties:
\begin{enumerate}
\item\label{i:cover}$H[V_i\setminus S_i]$ has a vertex cover $C_i$ of size at most $(w+1)(3|M_i|-|X_i|-|\hX_i|)$.
\item\label{i:diffcomponents} Each edge in $M_i$ is in a different component of $H[V_i\setminus S_i]$ (in particular, $S_i$ is disjoint from $V(M_i)$).
\item\label{i:nostar}$\stwo(u),\stwo(v)\le 2(w+1)$ for every edge $uv\in M_i$.
\end{enumerate}
For $i=0$, all three conditions hold vacuously.
If $|M_{i-1}|\ge k$, then we stop the process: a subset of $k$ edges of $M_{i-1}$ is the required induced matching. Otherwise, given $M_{i-1}$
and $X_{i-1}$, we compute $M_i$ and $X_i$ the following way. Let us
choose a node $t^*$ such that $H_{t^*}\setminus (B_{t^*}\cup
V_{i-1})$ contains at least one edge and the distance of $t^*$ from
the root $r$ is maximum possible. We claim that at least one such node
exists; in particular, the root $r$ is such a node. Otherwise, if
$H_r\setminus (B_{r}\cup V_{i-1})$ has no edge, then
$B_{r}\cup S_{i-1}\cup C_{i-1}$ is covers of $H_r=H$ (note that
$\hS_{i-1}\subseteq S_i$ covers every edge between $V_{i-1}$ and
$V(H)\setminus V_{i-1}$, and $C_{i-1}$ is covers every edge in $H[V_{i-1}\setminus S_{i-1}]$) and its size is at most
$w+1+(w+1)|X_{i-1}|+(w+1)(3|M_{i-1}|-|X_{i-1}|-|\hX_{i-1}|)\le
3k(w+1)$ (using property~\ref{i:cover} on the size of $C_{i-1}$ and $|M_{i-1}|<k$), contradicting our assumption on the vertex cover number
of $H$.  Let $t_1$, $\dots$, $t_p$ be the children
of $t^*$ for which $H_{t_j}\setminus (B_{t^*}\cup V_{i-1})$ contains
at least one edge. Note that $p\ge 1$: if there is an edge in $H_{t^*}\setminus (B_{t^*}\cup V_{i-1})$, then it has to appear in $H_{t_j}$ for some child $t_j$ of $t^*$. We consider two cases.

Case 1: there is a $t_j$ that has more than one descendants in $\hX_{i-1}$ (see Figure~\ref{fig:case1}). Let
$\hX'_{i-1}$ be the descendants of $t_j$ in $\hX_{i-1}$.  Let us find
a node $t$ that is at maximum distance from the root and has at least
two descendants in $\hX'_{i-1}$ (it is possible that $t=t_j$). In other words, for every pair of
nodes in $\hX'_{i-1}$, we find the least common ancestor of the two
nodes and we take $t$ to be a node of maximum distance from the root
among these common ancestors.  We let $M_i=M_{i-1}$ and
$X_i=X_{i-1}\cup \{t\}$. Observe that $t_j$ is an ancestor of $t$ and therefore $t^*$ is a proper ancestor of $t$. Thus by the choice of $t^*$, there is no edge in $H_{t}\setminus (B_{t}\cup V_{i-1})$. Therefore, $C_i:=C_{i-1}$ is a vertex cover of $H[V_i\setminus
S_i]$: edges with both endpoint in $V_{i-1}$ are covered by $C_i$,
edges with exactly one endpoint in $V_{i-1}$ have one endpoint in
$\hS_{i-1}\subseteq S_i$, and edges with both endpoints in
$V_i\setminus V_{i-1}$ have one endpoint in $B_t\subseteq
S_i$. Observe that the descendants of $t$ in $\hX'_{i-1}$ (there are at least two such nodes) are no longer maximal nodes
in $\hX_i$ after adding $t$ and we have added only one new node to
$X_i$, hence $|\hX_{i}|\le |\hX_{i-1}|-1$. Together with $|M_i|=|M_{i-1}|$
and $|X_i|=|X_{i-1}|+1$, this implies that $3|M_i|-|X_i|-|\hX_i|\ge
3|M_{i-1}|-|X_{i-1}|-|\hX_{i-1}|$. Therefore, $C_i$ satisfies the size
bound of property~\ref{i:cover}. For 
property~\ref{i:diffcomponents}, 
observe first that $S_i\cap V_{i-1}\subseteq S_{i-1}$ (if a vertex appears in $B_t$ and $V_{i-1}$, then it has to appear in a bag of $\hX_{i-1}$) and therefore the fact that $S_{i-1}$ is disjoint from $V(M_{i-1})$ implies that $S_i$ disjoint from $V(M_i)$. Together with $S_{i-1}\subseteq S_i$, it follows that the edges in $M_i$ are indeed in different components of $H[V_i\setminus S_i]$.
 Property~\ref{i:nostar} follows from $M_i=M_{i-1}$.
\begin{figure}
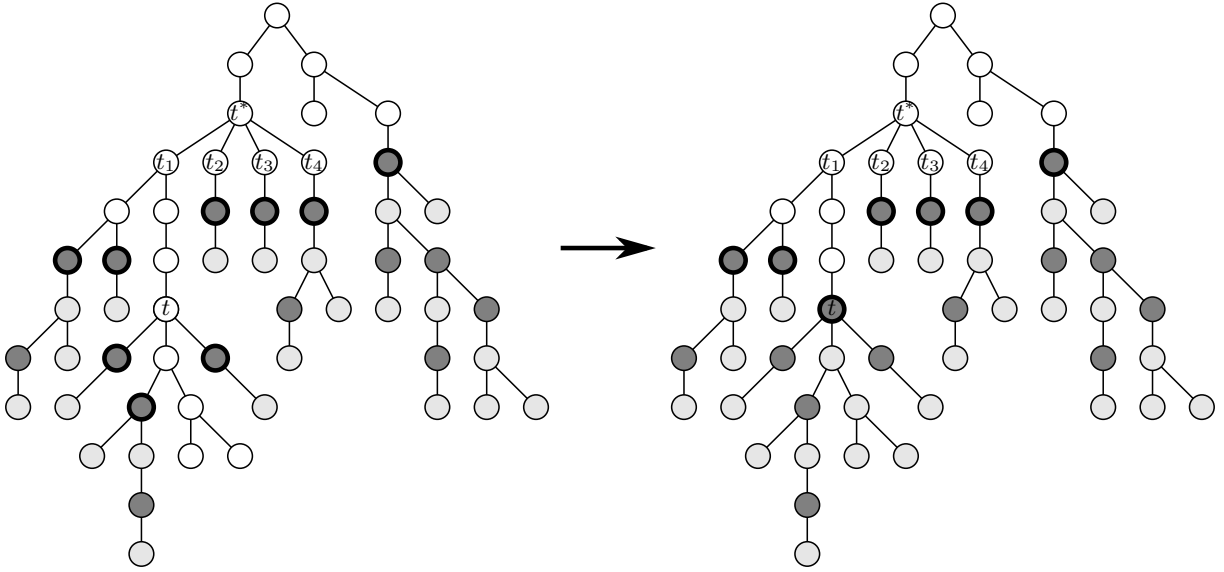

\begin{center}
{\footnotesize \svg{\linewidth}{case1}}
\caption{Case 1 in the proof of Lemma~\ref{lem:nicematching} with $p=4$ and $j=1$. On the left, shaded nodes show $V_{i-1}$, dark shaded nodes show $X_{i-1}$, dark circled nodes show $\hX_{i-1}$. We have $|X_i|=|X_{i-1}|+1$ and $|\hX_i|=|\hX_{i-1}|-2<|\hX_{i-1}|$.}\label{fig:case1}
\end{center}
\end{figure}

Case 2: Every $t_j$ has at most one descendant in $\hX_{i-1}$ (Figure~\ref{fig:case2}). For
every $1\le j \le p$, we let $u_jv_j$ to be an edge of
$H_{t_j}\setminus (B_{t^*}\cup V_{i-1})$. We let
$M_i=M_{i-1}\cup \bigcup_{j=1}^p\{u_jv_j\}$ and $X_i=X_{i-1}\cup
\{t^*\}$.  To prove property~\ref{i:cover}, we show that $C_i:=C_{i-1}
\cup \bigcup_{j=1}^p B_{t_j}$ is a vertex cover of $H[V_{i}\setminus
S_i]$.
 Consider an
edge $u'v'$ of $H[V_i\setminus S_i]$ not covered by $C_i$.  If both
$u'$ and $v'$ are in $V_{i-1}\setminus S_{i-1}$, then $u'v'$ is covered by $C_{i-1}$;
if, say, $u'\in V_{i-1}$ and $v'\not\in V_{i-1}$, then $u'\in
\hS_{i-1}\subseteq S_{i-1}$. Therefore, we may assume that $u'v'$ is an
edge of $H\setminus V_{i-1}$. Neither $u'$ nor $v'$ can be in $B_{t^*}\subseteq S_i$. Thus $u'v'$ is an edge of $H_{t^*}\setminus
(B_{t^*}\cup V_{i-1})$. By the way we defined $t_1$, $\dots$, $t_p$, this
means that $u'v'$ is an edge of $H_{t_j}\setminus (B_{t^*}\cup
V_{i-1})$ for some $1\le j \le p$. As $B_{t_j}\subseteq C_i$, it is in fact an edge of 
$H_{t_j}\setminus (B_{t_j}\cup
V_{i-1})$ as well. However, this contradicts the choice of $t^*$.
Let us prove that the size bound of property~\ref{i:cover} holds for $C_i$.
As we add only the single new node $t^*$ to $X_i$, we have $|X_i|=
|X_{i-1}|+1$ and $|\hX_i|\le |\hX_{i-1}|+1$ (note that the equality
$|\hX_i|=|\hX_{i-1}|+1$ is possible, but only if $t^*$ has no
descendant in $X_{i-1}$). Together with $|M_i|=|M_{i-1}|+p$, it
follows that
\begin{align*}
|C_i|&\le |C_{i-1}|+p(w+1)\\ &\le (w+1)(3|M_{i-1}|-|X_{i-1}|-|\hX_{i-1}|)+p(w+1)
\\&\le (w+1)(3(|M_i|-p)-(|X_i|-1)-(|\hX_i|-1))+p(w+1)\\&=
(w+1)( 3|M_i|-|X_i|-|\hX_i|-2p+2)\\&\le
(w+1)(3|M_i|-|X_i|-|\hX_i|),
\end{align*}
where we used the property~\ref{i:cover} on $C_{i-1}$ in the first inequality and $p\ge 1$ in the last inequality.
\begin{figure}
\begin{center}
{\footnotesize \svg{\linewidth}{case2}}
\caption{Case 2 in the proof of Lemma~\ref{lem:nicematching} with $p=3$. On the left, shaded nodes show $V_{i-1}$, dark shaded nodes show $X_{i-1}$, circled nodes show $\hX_{i-1}$. We have $|X_i|=|X_{i-1}|+1$ and $|\hX_i|=|\hX_{i-1}|-2\le |\hX_{i-1}|+1$.
}\label{fig:case2}
\end{center}
\end{figure}

To prove property~\ref{i:diffcomponents}, observe first 
that every $u_jv_j$ is disjoint from $B_{t^*}$ by definition and $B_{t^*}\cap V_{i-1}\subseteq \hS_{i-1}$, thus $S_i$ is disjoint from $V(M_i)$. Notice that $\hS_{i-1}\subseteq S_{i-1}\subseteq S_i$ separates 
$V_{i-1}$ from $V_i\setminus V_{i-1}$ and 
every edge of $M_{i}\setminus M_{i-1}$ is in $V_i\setminus V_{i-1}$
Therefore, no  
edge of $M_{i-1}$ can be in the same component of $H\setminus  S_i$ as an edge $M_{i}\setminus M_{i-1}$. Furthermore, the edges of $M_i\setminus M_{i-1}$ are separated by $B_{t^*}\subseteq S_i$.
To prove property~\ref{i:nostar}, suppose that, for $\ell=2(w+1)+1$,
there is a subdivided $\ell$-star centered at $u_j$ (the argument is
the same for $v_j$); let $u_j\alpha_1\beta_1$, $\dots$,
$u_j\alpha_\ell \beta_\ell$ be the paths of length 2 starting at
$u_j$. Node $t_j$ can have at most one descendant in $\hX_{i-1}$.  If
there is such a descendant $x_j\in \hX_{i-1}$, then there is an $1\le
q \le \ell$ such that $\alpha_q, \beta_q\not \in B_{t_j}\cup B_{x_j}$;
if there is no such descendant, then let us choose $q$ such that
$\alpha_q, \beta_q\not \in B_{t_j}$. It follows that $\alpha_q$ and
$\beta_q$ only appear in bags that are proper descendants of $t_j$,
but they do not appear in any bag of $X_{i-1}$, i.e.,
$\alpha_q,\beta_q\not\in V_{i-1}$. It follows that $\alpha_q\beta_q$
is an edge of $H_{t_j}\setminus (B_{t_j}\cup V_{i-1})$, contradicting
the selection of node $t^*$ and the edge $uv$. Thus we have
$\stwo(u_j),\stwo(v_j)\le 2(w+1)$, as required by
property~\ref{i:nostar}.
\end{proof}

The following two technical lemmas will be used in the proof of Lemma~\ref{lem:gadgetmainboundedtw}.
\begin{lemma}\label{lem:onlyone}
If $\C$ is a multiset of at least $(1+z\cdot r)k$ subsets of a universe $U$, each having size at most $r$, then there is a subcollection $\C'\subseteq \C$ of size $k$ such that for every $x\in U$, either there is at most one set in $\C'$ containing $x$, or there are at least $z$ sets in $\C\setminus \C'$ containing $x$.
\end{lemma}
\begin{proof}
  We prove the statement by induction on $k$. Let us select an
  arbitrary set $X\in\C$. For every $x\in X$, let us arbitrarily
  select $z$ sets of $\C\setminus \{X\}$ that contain $x$ (or all of
  them, if there are less than $z$ such sets); we define $\C_X$ as
  these selected sets; we have $|\mathcal{H}_X|\le z \cdot r$. Let us apply the
  induction hypothesis on the multiset $\C_{k-1}:=\C\setminus (\C_X\cup \{X\})$ and
  $k-1$ (note that $\C_{k-1}$ has size at least $(1+z\cdot
  r)(k-1)$); let $\C'_{k-1}$ be the resulting subcollection of $k-1$ sets. We claim
  that $\C'=\C'_{k-1}\cup \{X\}$ is the desired collection of $k$ sets. Indeed,
  for every vertex $x\in X$, if $x$ appears in a set of $\C'$, then it
  appears in at least $z$ sets of $\C_X\subseteq \C \setminus 
  (\C'_{k-1}\cup \{X\}))$ and the statement is true for every $x\not\in X$
  by the induction hypothesis.
\end{proof}

\begin{lemma}\label{lem:starclique}
  Let $Z$ be a set of vertices in a graph $H$ of treewidth at most
  $w$. If for every $v\in Z$ there is a subdivided star $S_v$ centered
  at $v$ covering every vertex of $Z$, then $|Z|\le w+1$.
\end{lemma}
\begin{proof}
  Consider a rooted tree decomposition $(T,\B)$ of $H$.  For every
  $t\in V(T)$, we denote by $V_t$ the union of every bag $B_{t'}$ for
  every descendant $t'$ of $t$ (including $t$ itself). For every
  vertex $v\in Z$, consider the node $t_v$ closest to the root that contains $v$, and let us select a $v\in Z$ such that $t_v$ has maximum
  distance from the root. Then $B_{t'}\cap Z\subseteq B_{t_v}\cap Z$
  for every proper descendant $t'$ of $t_v$, otherwise there would be
  a vertex $u\in Z$ such that $t_{u}$ is a proper descendant of
  $t_v$. The subdivided star $S_v$ covers $Z$ by
  assumption, hence there is a path of length at most two between $v$
  and each vertex of $Z\setminus \{v\}$ such that $v$ is the only
  vertex shared by these paths.  We claim that each such path has to
  contain a vertex of $B_{t_v}\setminus \{v\}$: otherwise, the
  vertices of the path appear either only in bags that are proper
  descendants of $t_v$ (contradicting the maximality of $t_v$) or only
  in bags where $v$ does not appear (contradicting that a vertex of
  the path is a neighbor of $v$).  Therefore, we have $|Z\setminus
  \{v\}|\le |B_{t_v}\setminus \{v\}|\le w$ and $|Z|\le w+1$ follows.
\end{proof}

We are now ready to prove the main result for bounded-treewidth graphs, which completes the proof Theorem~\ref{thm:gadgetexists}.
\begin{lemma}\label{lem:gadgetmainboundedtw}
  There is a function $\fmain(k,w)$ such that if a graph $H$ with treewidth
  at most $w$ has vertex cover number greater than $\fmain(k,w)$, then there is
  a $k$-matching gadget $(H,M)$.
\end{lemma}
\begin{proof}
We define the following constants (the function $\fstar$ is from Lemma~\ref{lem:gadgetmainnostar}):
\begin{alignat*}{1}
L&=2(w+1)\\
r&=  2L+2(L(w+2)+1)\\
L_1&=2L+2\\
L_2&=  4r+2+L_1\\
z&=  L_2\\
k_2&=  \fstar(k,L_2)\\
k_1&=  2k_2\\
k_0&=   (1+z\cdot r)k_1\\
f(k,w)&=3k_0(w+1).
\end{alignat*}
 By
  Lemma~\ref{lem:nicematching}, if $H$ has vertex cover number greater
  than $f(k,w)$, then $H$ has an induced matching $M_0$ of size $k_0$ such that
  $\stwo_H(v)\le L$ for every $v\in V(M)$.

  For every $v\in V(H)\setminus V(M_0)$, let us fix a subdivided
  star $S_v$ centered at $v$ with $\min\{\stwo_H(v),L_2\}$ leaves and 
  having the minimum number of vertices in $V(M_0)$. For every $e\in M_0$,
  we let $v\in V(H)\setminus V(M_0)$ be in $X_e$ if $\stwo_H(v)\ge L_1$
  and $S_v$ uses an endpoint of $e$. 

\begin{claim}\label{cl:boundXe}
For every $e\in M_0$, $|X_e|\le r$.
\end{claim}
\begin{proof}
  Suppose first that $v\in X_e$ and $v$ is adjacent to $e$. As $v$ has
  degree at least $\stwo_H(v)\ge L_1=2L+2$,
  Proposition~\ref{prop:starneighbor} implies that each endpoint of
  $e$ can have at most $L$ such neighbors, hence there are at most
  $2L$ such vertices in $X_e$. Suppose therefore that $|X_e|$ has at least $r-2L=2(L(w+2)+1)$
  vertices not adjacent to $e$. For each such vertex $v\in X_e$, the
  subdivided star $S_v$ contains one or two paths $vxy$ with $y$
  being one of the two endpoints of $e$. Let us fix such an $x$ and $y$ for each
  vertex $v\in X_e$ not adjacent to $e$; for at least $L(w+2)+1$ vertices we
  have the same $y$. If there are $L+1$ such vertices $v\in X_e$ with
  distinct $x$'s, then this shows that there is a subdivided $L+1$
  star centered at $y$, a contradiction. Therefore, there are $w+2$
  vertices $v_1$, $\dots$, $v_{w+2}$ in $X_e$ sharing the same $x$. If
  $S_{v_i}$ does not use vertex $v_{j}$ for some $j\neq i$, then we
  can replace $y$ by $v_j$ in the subdivided star $S_{v_i}$,
  contradicting the minimality of $S_{v_i}$ with respect to the number of vertices of $V(M_0)$ used. Thus every $S_{v_i}$
  covers the set $Z=\{v_1,\dots,v_{w+2}\}$, contradicting
  Lemma~\ref{lem:starclique}.
\cqed\end{proof}

We construct a matching $M_1\subseteq M_0$ the following way.  Let $\C$ be
the multiset containing $X_e$ for every $e\in M_0$; we have
$|\C|=|M|=(1+z\cdot r)k_1$. Let us invoke Lemma~\ref{lem:onlyone} to
obtain a subcollection $\C'$ of size $k_1$ such that for every $x\in
V(H)\setminus V(M_0)$, either there is at most one set in $\C'$
containing $x$ or at least $z$ sets of $\C\setminus \C'$ contain
$x$. Let $M_1\subseteq M_0$ be the subset of $k_1$ edges corresponding
to the subcollection $\C'$.

\begin{claim}\label{cl:twoXe-ok}
If $v\in X_e$ for at least two different $e\in M_1$, then $\stwo_{H\setminus V(M_1)}(v)\ge L_2$.
\end{claim}
\begin{proof}
  By the way the subcollection $\C'$ is constructed, we have that $v$
  is in $X_e$ for at least $z=L_2$ edges $e\in M\setminus M_1$; let
  $F\subseteq M\setminus M_1$ be this set of edges. Let $vx_1y_1$,
  $\dots$, $vx_\ell y_\ell$ be the paths in the subdivided star
  $S_v$. Every edge of $F$ is intersected by 
one or two of these paths. Furthermore, as $M_0$ is an induced matching, 
if $vx_iy_i$ intersects an edge of $F\subseteq M\setminus M_1$, then it cannot intersect $V(M_1)$. Therefore, at least $|F|\ge L_2$ such paths are disjoint from $V(M_1)$. 
These paths form a subdivided $L_2$-star centered at $v$, implying $\stwo_{H\setminus V(M_1)}(v)\ge L_2$.
\cqed\end{proof}

Let $v\in V(H)\setminus V(M_1)$ be in $V_i$ if $\stwo_H(v)=i$ and
$v\in X_e$ for some $e\in M_1$.  As $|X_e|\le r$ for every $e\in M_e$ (Claim~\ref{cl:boundXe}),
we have $\sum_{i=L_1}^{L_2-2}|V_i|\le r k_1$. Therefore, there is an $
L_1 \le i^* \le L_2-2$ such that $|V_{i^*}|+|V_{i^*+1}|\le 2r
k_1/(L_2-2-L_1) = k_1/2$.  
Let $M_2$ contain every $e\in M_1$ with
$X_e\cap (V_{i^*}\cup V_{i^*+1})=\emptyset$.
Note that every $v\in V_{i^*}\cup
V_{i^*+1}$ is in $X_e$ for at most one $e\in M_1$: otherwise,
Claim~\ref{cl:twoXe-ok} implies that $ \stwo_H(v)\ge\stwo_{H\setminus V(M_1)}(v)\ge L_2$, contradicting
the choice $i^*\le L_2-2$.  Therefore, we have that the size
of $M_2$ is at least $k_1-(|V_{i^*}|+|V_{i^*+1}|)|\ge k_1/2=k_2$.
Furthermore, if $e\in M_2$, then $X_e$ is disjoint from $V_{i^*}\cup V_{i^*+1}$.

Let us define $Q$ as the set containing every $v\in V(H)$ with $\stwo(v)\ge i^*$ and let $C=V(H)\setminus V(M_2)$.

\begin{claim}\label{cl:starstrong}
Set $Q$ is a strong set of $H$ with respect to $C$.
\end{claim}
\begin{proof}
  Let $f$ be a boundary-preserving
  isomorphism from $H[C]$ to $H[C']$ for some $C'\subseteq V(H)$.
We show that $\stwo_{H[C]}(v)\ge i^*$ for every $v\in Q$.
 Then, as $f$ is an isomorphism between
  $H[C]$ and $H[C']$, we have that $\stwo_{H}(f(v))\ge
  \stwo_{H[C']}(f(v))=\stwo_{H[C]}(v)\ge i^*$, implying $f(v)\in
  Q$.

  Consider a vertex $v\in Q$ and the subdivided $S_v$ star. Recall
  that, as we have $\stwo_H(v)\ge i^*\ge L_1$, the subdivided star
  $S_v$ contains an endpoint of $e\in M_0$ only if $v\in X_e$.
  Suppose first that $v\in Q$ is not in $X_e$ for any $e\in M_2$; in
  particular, this is the case for any $v$ with $\stwo_H(v)\in
  \{i^*,i^*+1\}$. Then the subdivided star $S_v$ is disjoint from
  $V(M_2)$, that is, $S_v$ is fully contained
  in $C$. Then $\stwo_{H[C]}(v)\ge \min\{\stwo_{H}(v),L_2\}\ge i^*$.

  Suppose now that $v\in Q$ is in $X_e$ for exactly one edge $e\in
  M_2$; this is only possible if $\stwo_H(v)\ge i^*+2$. Then the
  subdivided star $S_v$ intersects the endpoints of at most one edge
  of $M_2$, that is, at most two paths of $S_v$ intersect
  $V(M_2)$. Therefore, a subdivided star centered at $v$ with
  $\min\{\stwo_{H}(v),L_2\}-2\ge i^*$ leaves appears in $H[C]$.

  Finally, suppose that $v\in Q$ is in $X_e$ for at least two edges
  of $M_2\subseteq M_1$. Then Claim~\ref{cl:twoXe-ok} implies that
  $\stwo_{H[C]}(v)\ge \stwo_{H\setminus V(M_1)}(v)\ge L_2 \ge i^*$.
  \cqed\end{proof} We have $\stwo(H\setminus Q)<i^*<L_2$ and $M_2$ is
a matching of size at least $k_2$ in $H\setminus Q$. Therefore,
Lemma~\ref{lem:gadgetmainnostar} implies that there is a $k$-matching
gadget $(H\setminus Q,M)$. As $Q$ is a strong set,
Lemma~\ref{lem:removestrong} implies that $(H,M)$ is also a
$k$-matching gadget.
\end{proof}

\bibliographystyle{abbrv}
\bibliography{h-counting}

\appendix
\section {Computations for Section \ref{sec:bipart-edge-colorf}}

In this section, we compute the values 
\[
p_{s,t}(0)=\#\mathcal{M}_{A_{t}}(R_{s}+0\cdot C_{6})=\#\mathcal{M}_{A_{t}}(R_{s})
\]
defined in \eqref{eq:pst} of Section  \ref{sec:bipart-edge-colorf} for all $s,t\in[5]$, where $R_{s}$ is the graph defined
in Figure~\ref{fig:triangle-states}. For each $t\in[5]$, we give a table on the following page that contains the
following information: 
\begin{enumerate}
\item All partitions of $A_{t}$ into $B_{1}\dot{\cup}\ldots\dot{\cup}B_{\ell}=A_{t}$
with $\ell\leq3$, with the additional property that, if $a,b\in B_{i}$
for any $i\in[\ell]$, then $a-b\not\equiv1\mod\ 6$. Note that these
properties are required for an $A_{t}$-colorful matching to exist
in the edge-colored graph $3\cdot C_{6}$ (and thus in $R_{s}$, for
any $s\in[5]$, since $R_{s}\subseteq3\cdot C_{6}$)
\item For each partition $B_{1}\dot{\cup}\ldots\dot{\cup}B_{\ell}=A_{t}$,
we compute the number of $A_{t}$-colorful matchings in $R_{s}$ with
the following property: Edges of colors $a,b\in A_{t}$ in the matching
are contained in the same component of $R_{s}$ if and only if $a,b\in B_{i}$
for some $i$. (This is essentially the same decomposition used in the proof of Lemma \ref{lem:matching-polynomial}.)
\item In the last row of the table, we compute the sum of each column. The
$s$-th value in this row then counts the number of $A_{t}$-colorful
matchings in $R_{s}$, i.e., it is equal to the value $p_{s,t}(0)$.
\end{enumerate}

\begin{table}[H!]
\begin{centering}
\begin{tabular}{|c||c|c|c|c|c|}
\hline 
 & $R_{1}$ & $R_{2}$ & $R_{3}$ & $R_{4}$ & $R_{5}$\tabularnewline
\hline 
\hline 
$\{4\}\{5\}$ & 2 & 2 & 3 & 3 & 3\tabularnewline
\hline 
\hline 
$\mathbf{\Sigma}$ & 2 & 2 & 3 & 3 & 3\tabularnewline
\hline 
\end{tabular}
\par\end{centering}

\caption{The table for $t=1$.}
\end{table}

\begin{table}[H!]
\begin{centering}
\begin{tabular}{|c||c|c|c|c|c|}
\hline 
 & $R_{1}$ & $R_{2}$ & $R_{3}$ & $R_{4}$ & $R_{5}$\tabularnewline
\hline 
\hline 
$\{2\}\{3\}$ & 2 & 3 & 2 & 3 & 3\tabularnewline
\hline 
\hline 
$\mathbf{\Sigma}$ & 2 & 3 & 2 & 3 & 3\tabularnewline
\hline 
\end{tabular}
\par\end{centering}

\caption{The table for $t=2$.}
\end{table}

\begin{table}[H!]
\begin{centering}
\begin{tabular}{|c||c|c|c|c|c|}
\hline 
 & $R_{1}$ & $R_{2}$ & $R_{3}$ & $R_{4}$ & $R_{5}$\tabularnewline
\hline 
\hline 
$\{1\}\{6\}$ & 2 & 3 & 3 & 2 & 3\tabularnewline
\hline 
\hline 
$\mathbf{\Sigma}$ & 2 & 3 & 3 & 2 & 3\tabularnewline
\hline 
\end{tabular}
\par\end{centering}

\caption{The table for $t=3$.}
\end{table}

\begin{table}[H!]
\begin{centering}
\begin{tabular}{|c||c|c|c|c|c|}
\hline 
 & $R_{1}$ & $R_{2}$ & $R_{3}$ & $R_{4}$ & $R_{5}$\tabularnewline
\hline 
\hline 
$\{2,4\}\{3,5\}$ & 2 & 1 & 1 & 0 & 1\tabularnewline
\hline 
$\{2,4\}\{3\}\{5\}$ & 0 & 0 & 2 & 1 & 1\tabularnewline
\hline 
$\{2,5\}\{3\}\{4\}$ & 0 & 0 & 0 & 2 & 2\tabularnewline
\hline 
$\{2\}\{3,5\}\{4\}$ & 0 & 2 & 0 & 1 & 1\tabularnewline
\hline 
\hline 
$\mathbf{\Sigma}$ & 2 & 3 & 3 & 4 & 5\tabularnewline
\hline 
\end{tabular}
\par\end{centering}

\caption{The table for $t=4$.}
\end{table}

\begin{table}[H!]
\begin{centering}
\begin{tabular}{|c||c|c|c|c|c|}
\hline 
 & $R_{1}$ & $R_{2}$ & $R_{3}$ & $R_{4}$ & $R_{5}$\tabularnewline
\hline 
\hline 
$\{1,3,5\}\{2,4,6\}$ & 2 & 0 & 0 & 0 & 0\tabularnewline
\hline 
$\{1,3,5\}\{2,4\}\{6\}$ & 0 & 0 & 1 & 0 & 0\tabularnewline
\hline 
$\{1,3\}\{2,4,6\}\{5\}$ & 0 & 0 & 1 & 0 & 0\tabularnewline
\hline 
$\{1,3,5\}\{2,6\}\{4\}$ & 0 & 0 & 0 & 0 & 0\tabularnewline
\hline 
$\{1,3,5\}\{2\}\{4,6\}$ & 0 & 0 & 0 & 1 & 0\tabularnewline
\hline 
$\{1,3\}\{2,5\}\{4,6\}$ & 0 & 1 & 0 & 0 & 0\tabularnewline
\hline 
$\{1,4\}\{2,5\}\{3,6\}$ & 0 & 0 & 0 & 0 & 1\tabularnewline
\hline 
$\{1,4\}\{2,6\}\{3,5\}$ & 0 & 0 & 0 & 0 & 1\tabularnewline
\hline 
$\{1,5\}\{2,4,6\}\{3\}$ & 0 & 0 & 0 & 0 & 1\tabularnewline
\hline 
$\{1,4\}\{2,5\}\{3,6\}$ & 0 & 0 & 0 & 1 & 0\tabularnewline
\hline 
$\{1,5\}\{2,4\}\{3,6\}$ & 0 & 0 & 0 & 0 & 1\tabularnewline
\hline 
$\{1\}\{2,4,6\}\{3,5\}$ & 0 & 1 & 0 & 0 & 0\tabularnewline
\hline 
$\mathbf{\Sigma}$ & 2 & 2 & 2 & 2 & 4\tabularnewline
\hline 
\end{tabular}
\par\end{centering}

\caption{The table for $t=5$.}
\end{table}

Then it can be read off the tables that
\[
R_{1}(0)=\left(\begin{array}{ccccc}
2 & 2 & 3 & 3 & 3\\
2 & 3 & 2 & 3 & 3\\
2 & 3 & 3 & 2 & 3\\
2 & 3 & 3 & 4 & 5\\
2 & 2 & 2 & 2 & 4
\end{array}\right),
\]

The determinant of this matrix is $12$ and $12 \neq 0$ under standard complexity-theoretic assumptions.

\end{document}